\documentclass[12pt]{amsart}
\usepackage{ulem}
\usepackage{amssymb,amsthm}
\usepackage{amsmath}
\usepackage{hyperref}
\usepackage{tikz}
\usepackage{graphicx}
\usepackage{subcaption}
\usepackage{booktabs,tabularx}
\newcolumntype{C}{>{\centering\arraybackslash}X}
\usepackage[shortlabels]{enumitem}
\usetikzlibrary{arrows}
\usepackage{float}
\usepackage{graphicx}
\usepackage{subcaption}
\graphicspath{{MicCompFigsNew/}}
\usetikzlibrary{decorations.pathmorphing}
\tikzset{snake it/.style={decorate, decoration=snake}}

\makeatletter
\newcommand*{\rom}[1]{\expandafter\@slowromancap\romannumeral #1@}
\makeatother

\numberwithin{equation}{section}

\theoremstyle{plain}

\newtheorem{theorem}{Theorem}
\numberwithin{theorem}{section}
\newtheorem{proposition}[theorem]{Proposition}

\theoremstyle{definition}
\newtheorem{definition}[theorem]{Definition}
\theoremstyle{remark}
\newtheorem{remark}[theorem]{Remark}
\theoremstyle{remark}

\theoremstyle{remark}
\newtheorem{example}[theorem]{Example}

\DeclareMathOperator*{\argmin}{arg\,min}

\usepackage{fullpage}

\newcommand{\smo}{\setminus \left\{\mathbf{0}\right\}}

\newcommand{\partyf}[2]{\frac{\partial #2}{\partial y_{#1}}}

\newcommand{\abs}[1]{\left|#1\right|}                 % Absolutbetrag
\newcommand{\paren}[1]{\left(#1\right)}               % Klammern
\newcommand{\bparen}[1]{\left[#1\right]}               % eckige Klammern
\newcommand{\sparen}[1]{\left\{#1\right\}}      % Mengenklammer
\newcommand{\dd}{\mathrm{d}}  % without the space
\newcommand{\Cc}{\mathcal{C}}

\newcommand{\Dc}{\mathcal{D}}
\newcommand{\Ec}{\mathcal{E}}
\newcommand{\Fc}{\mathcal{F}}

\newcommand{\Sc}{\mathcal{S}}
\newcommand{\Tc}{\mathcal{T}}

\newcommand{\WF}{\mathrm{WF}}                         % Wavefront set
\newcommand{\wf}{\mathrm{WF}}                         % Wavefront set

\newcommand{\vc}{\mathbf{c}}

\newcommand{\vO}{\mathbf{O}}

\newcommand{\vx}{{\mathbf{x}}}
\newcommand{\vy}{{\mathbf{y}}}

\newcommand{\vxi}{{\boldsymbol{\xi}}}
\newcommand{\vxio}{{\boldsymbol{\xi}_0}}

\newcommand{\vs}{{\boldsymbol{\sigma}}} %when $\sigma$ is one dimensional, 
\newcommand{\xo}{x_0}

\newcommand{\xoj}{(x_0)_j} %play with how this looks                            
\newcommand{\xoo}{(x_0)_1} %play with how this looks                            
\newcommand{\xot}{(x_0)_2} %play with how this looks                            

\newcommand{\rr}{{{\mathbb R}}}

\newcommand{\rtwo}{{{\mathbb R}^2}}

\newcommand{\rn}{{{\mathbb R}^n}}

\newcommand{\st}{\hskip 0.3mm : \hskip 0.3mm}

\newcommand{\be}{\begin{equation}}
\newcommand{\ee}{\end{equation}}
\newcommand{\bea}{\begin{eqnarray}}
\newcommand{\eea}{\end{eqnarray}}
\newcommand{\bean}{\begin{eqnarray*}}
\newcommand{\eean}{\end{eqnarray*}}

\newcommand{\bel}[1]{\begin{equation}\label{#1}}

\newcommand{\eel}[1]{{\label{#1}\end{equation}}}

\newcommand{\Lot}{\Lambda_{12}}
\newcommand{\Lto}{\Lambda_{21}}
\newcommand{\Lij}{\Lambda_{ij}}
\newcommand{\lot}{\lambda_{12}}
\newcommand{\lto}{\lambda_{21}}
\newcommand{\lij}{\lambda_{ij}}
\newcommand{\xxi}{(\vx,\vxi)}
\newcommand{\yeta}{(\vy,\boldsymbol{\eta})}

\renewcommand{\emph}[1]{\textit{#1}}

\usepackage{datetime}
\usetikzlibrary{quotes, angles}

\title[A joint reconstruction and lambda tomography regularization technique]{A joint reconstruction and lambda tomography regularization technique for energy-resolved X-ray imaging
\\{\footnotesize\ddmmyyyydate\today~\currenttime}}

\author{James Webber}
\address{Department of Electrical and Computer
Engineering, Tufts University, Medford, MA USA}
\email{James.Webber@tufts.edu}
\author{Eric Todd Quinto}
\address{Department
of Mathematics, Tufts University, Medford, MA USA}
\email{Todd.Quinto@tufts.edu}
\author{Eric L. Miller}
\address{Department of Electrical and Computer
Engineering, Tufts University, Medford, MA USA}
\email{elmiller@ece.tufts.edu}

\begin{document}
%\titlep
%\tableofcontents
\begin{abstract}
We present new joint reconstruction and regularization techniques
inspired by ideas in microlocal analysis and lambda tomography, for
the simultaneous reconstruction of the attenuation coefficient and
electron density from X-ray transmission (i.e., X-ray CT) and
backscattered data (assumed to be primarily Compton scattered). To
demonstrate our theory and reconstruction methods, we consider the
``parallel line segment" acquisition geometry of
\cite{webber2019compton}, which is motivated by system architectures
currently under development for airport security screening. We first
present a novel microlocal analysis of the parallel line geometry
which explains the nature of image artefacts when the attenuation
coefficient and electron density are reconstructed separately. We next
introduce a new joint reconstruction scheme for low effective $Z$
(atomic number) imaging ($Z<20$) characterized by a regularization
strategy whose structure is derived from lambda tomography principles
and motivated directly by the microlocal analytic results.  Finally we
show the effectiveness of our method in combating noise and image
artefacts on simulated phantoms.
\end{abstract}

\maketitle

\section{Introduction} In this paper we introduce new joint
reconstruction and regularization techniques based on ideas in
microlocal analysis and lambda tomography \cite{FFRS, FRS} (see also
\cite{Lo2008} for related work). We consider the simultaneous
reconstruction of the attenuation coefficient $\mu_E$ and electron
density $n_e$ from joint X-ray CT (transmission) and Compton scattered
data, with particular focus on the parallel line segment X-ray scanner
displayed in figures \ref{fig1} and \ref{fig1.1}. The acquisition
geometry in question is based on a new airport baggage scanner
currently in development, and has the ability to measure X-ray CT and
Compton data simultaneously.  The line segment geometry was first
considered in \cite{webber2019compton}, where injectivity results are
derived in Compton Scattering Tomography (CST). We provide a stability
analysis of the CST problem of \cite{webber2019compton} here, from a
microlocal perspective. The scanner depicted in figure \ref{fig1}
consists of a row of fixed, switched, monochromatic fan beam sources
($S$), a row of detectors ($D_A$) to measure the transmitted photons,
and a second (slightly out of plane) row of detectors ($D_C$) to
measure Compton scatter. The detectors are assumed to energy-resolved,
a common assumption in CST
\cite{palamodov2011analytic,truong2019compton,RigaudComptonSIIMS2017,
rigaud20183d, webber2020microlocal}, and the sources are fan-beam (in
the plane) with opening angle $\pi$ (so there is no restriction due to
cropped fan-beams). 
\begin{figure}[!h]
\centering
\begin{tikzpicture}[scale=5]%6
\draw [very thick] (-0.75,0)--(0.75,0)node[right] {$D_A$ at $\{z=2-r_m\}$};
\draw [dashed] (-0.75,1)--(0.75,1)node[right] {$D_C$};
\draw [very thick] (-0.75,1.5)--(0.75,1.5)node[right] {$S$};
\draw [->,line width=1pt] (0,0.75)--(0,1.7)node[right] {$x_2$};
\draw [->,line width=1pt] (0,0.75)--(0.9,0.75)node[below] {$x_1$};
%\node at (0.04,1.75) {$x_2$};
%\node at (1.04,0.45+0.25) {$x_1$};
\node at (0,0.47+0.23) {$O$};
\draw [->] (-0.75,1.5)--(0.75,0);
\draw [->] (0.75,1.5)--(-0.75,0);
\draw [->] (0.3,1.5)--(-0.1,0);
\draw [->] (-0.3,1.5)--(0.5,0);
\draw [->] (-0.7,1.5)--(-0.4,0);
\node at (-0.6,0.75) {$L$};
\draw [->] (0.5,1.5)--(0.65,0);
\coordinate (origo1) at (0,0.75);
\coordinate (pivot1) at (-0.75,0);
\coordinate (bob1) at (-0.75,1.5);
\draw pic[fill=orange, <->,"$C_R$", angle eccentricity=1.7] {angle = bob1--origo1--pivot1};
\draw [<->] (0,1.55)--(0.75,1.55);
\node at (0.375,1.58) {$a$};
\end{tikzpicture}
\caption{Parallel line X-ray CT geometry. Here $S$, $D_C$ and $D_A$ denote the source and detector rows. The length of the detector (and source) array is $2a$. A cone $C_R\subset S^1$ is highlighted in orange. We will refer to $C_R$ later for visualisation in section \ref{microsec2}.}
\label{fig1}
\end{figure}
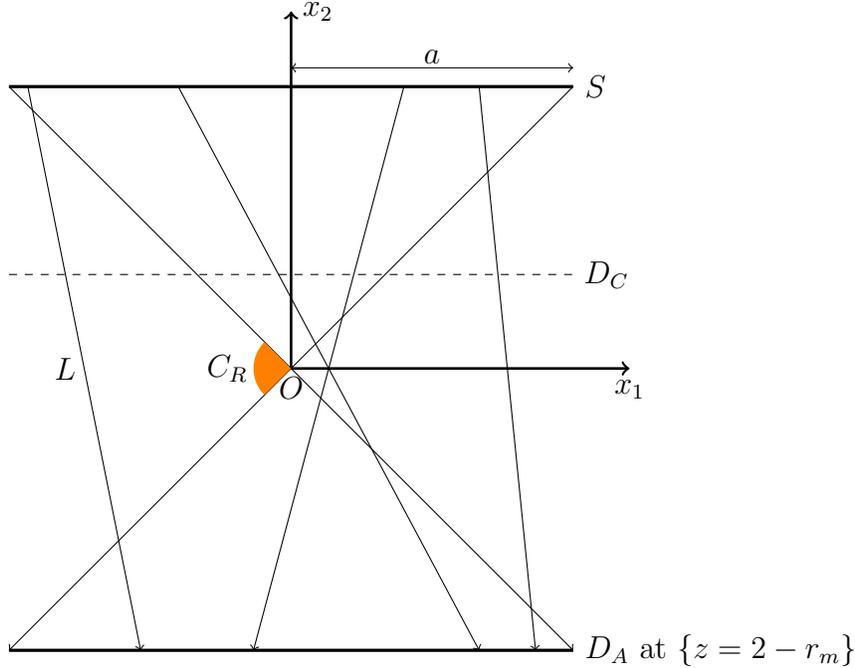

%\begin{equation}
%
%\end{equation}

The attenuation coefficient relates to the X-ray transmission data by
the Beer-Lambert law,
$\log\left(\frac{I_0}{I_A}\right)=\int_L\mu_E\mathrm{d}l$ \cite[page
2]{krishnan2014microlocal} where $I_A$ is the photon intensity
measured at the detector, $I_0$ is the initial source intensity and
$\mu_E$ is the attenuation coefficient at energy $E$. Here $L$ is a
line through $S$ and $D_A$, with arc measure $\mathrm{d}l$. Hence the
transmission data determines a set of integrals of $\mu_E$ over lines,
and the problem of reconstructing $\mu_E$ is equivalent to inversion
of the line Radon transform with limited data (e.g.,
\cite{krishnan2014microlocal, natterer}). Note that we need not
account for the energy dependence of $\mu_E$ in this case as the
detectors are energy-resolved, and hence there are no issues due to
beam-hardening. See figure \ref{fig1}.

When the attenuation of the incoming and scattered rays is ignored,
the Compton scattered intensity in two-dimensions can be modelled as
integrals of $n_e$ over toric sections
\cite{palamodov2011analytic,truong2019compton,webber2020microlocal}
\begin{equation}
I_C=\int_{T}n_e\mathrm{d}t,
\end{equation}
where $I_C$ is the Compton scattered intensity measured at a point
on $D_C$. A toric section
$T=C_1\cup C_2$ is the union of two intersecting circles of the same
radii (as displayed in figure \ref{fig1.1}), and $\mathrm{d}t$ is the arc measure on
$T$. The recovery of $n_e$ is equivalent to inversion of the
toric section Radon transform
\cite{norton1994compton, palamodov2011analytic, truong2019compton,
webber2020microlocal}.  See figure \ref{fig1.1}. See also
\cite{RigaudComptonSIIMS2017, rigaud20183d} for alternative
reconstruction methods. We now discuss the approximation made above to
neglect the attenuative effects from the CST model. When
the attenuation effects are included, the inverse scattering problem
becomes nonlinear \cite{rigaud20193d}. 
We choose to focus on the analysis of the
idealised linear case here, as this allows us to apply the well
established theory on linear Fourier Integral Operators (FIO) and microlocal
analysis to derive expression for the image artefacts.
Such analysis will likely give valuable insight into the expected
artefacts in the nonlinear case. The nonlinear models and their inversion properties are left for future work.
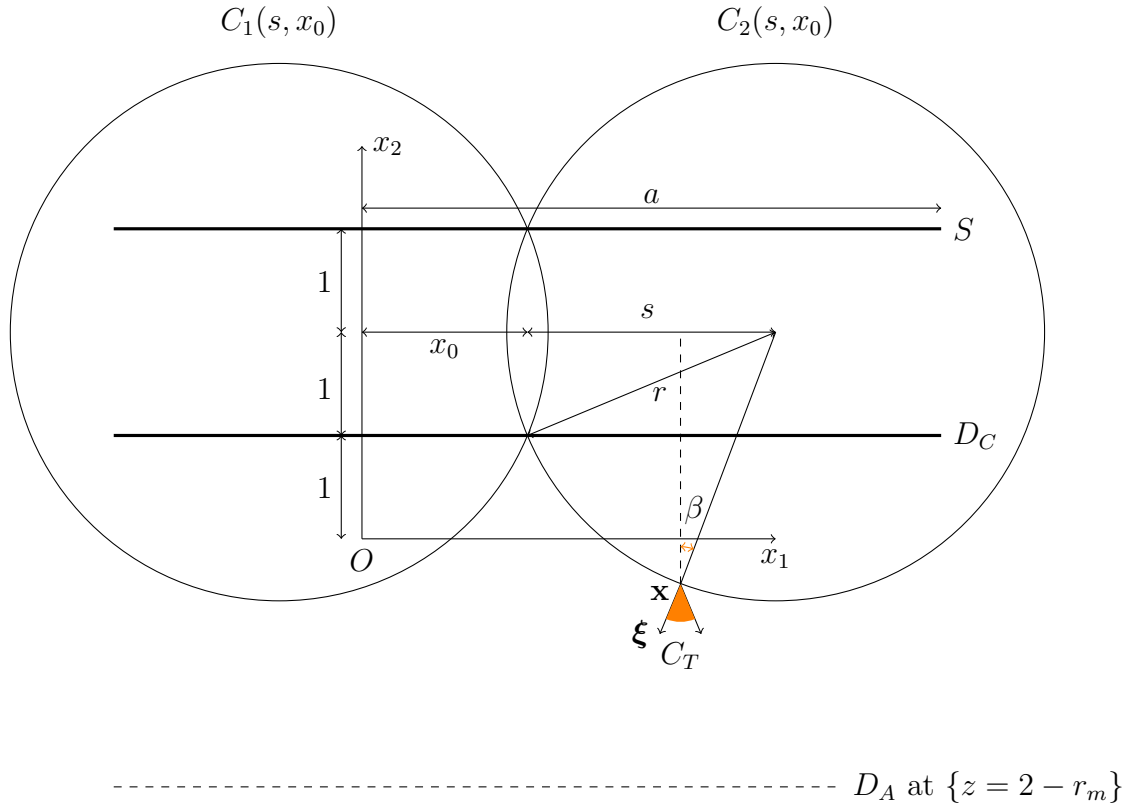
\begin{figure}[!h]
\centering
\begin{tikzpicture}[scale=5.5]
\draw [dashed] (-1,0.15)--(0.76,0.15)node[right] {$D_A$ at $\{z=2-r_m\}$};%origianlly y=0
\draw [very thick] (-1,1)--(1,1)node[right] {$D_C$};
\draw [<->] (-0.1-0.35,1)--(-0.1-0.35,1.25);
\node at (-0.14-0.35,1.1125) {$1$};
\draw [<->] (-0.1-0.35,1.25)--(-0.1-0.35,1.5);
\node at (-0.14-0.35,1.3725) {$1$};
\draw [<->] (-0.1-0.35,0.75)--(-0.1-0.35,1);
\node at (-0.14-0.35,0.8725) {$1$};
\draw [very thick] (-1,1.5)--(1,1.5)node[right] {$S$};
\draw [<->] (-0.4,1.55)--(1,1.55);
\node at (0.3,1.58) {$a$};
\draw [->] (-0.4,0.75)--(-0.4,1.7)node[right] {$x_2$};
\draw [->] (-0.4,0.75)node[below] {$O$}--(1-0.4,0.75)node[below] {$x_1$};
\draw [<->] (-0.4,1.25)--(0,1.25);
\node at (-0.2,1.21) {$x_0$};
\draw (0.6,1.25) circle (0.65);
\draw (-0.6,1.25) circle (0.65);
\draw [<->] (0.5+0.1,1.25)--(0,1);
\node at (0.24+0.08,1.1) {$r$};
\draw [<->] (0.5+0.1,1.25)--(0,1.25);
\node at (0.24+0.05,1.3) {$s$};
%\node at (0.04,1.75) {$x_2$};
%\node at (1.04,0.45+0.25) {$x_1$};
%\node at (-0.03,0.47+0.25) {$O$};
\node at (0.32,0.62) {$\vx$};
\draw [->]  (0.37,0.64)--(0.32,0.52)node[left] {$\vxi$};
\draw [->]  (0.37,0.64)--(0.42,0.52);
\coordinate (origo1) at (0.37,0.64);
\coordinate (pivot1) at (0.42,0.52);
\coordinate (bob1) at (0.32,0.52);
\draw pic[fill=orange, <->,"$C_T$", angle eccentricity=1.9] {angle = bob1--origo1--pivot1};
\draw  (0.37,0.64)--(0.6,1.25);
\node at (-0.6,2) {$C_1(s,x_0)$};
\node at (0.6,2) {$C_2(s,x_0)$};
\draw  [dashed] (0.37,0.64)--(0.37,1.25);
\coordinate (origo) at (0.37,0.64);
\coordinate (pivot) at (0.37,1.25);
\coordinate (bob) at (0.6,1.25);
\draw pic[draw=orange, <->,"$\beta$", angle eccentricity=2] {angle = bob--origo--pivot};
%\draw [->] (-0.6,1.5)--(0.6,0);
%\draw [->] (0.6,1.5)--(-0.6,0);
%\coordinate (origo1) at (0,0.75);
%\coordinate (pivot1) at (0.6,0);
%\coordinate (bob1) at (0.6,0.75);
%\draw pic[draw=orange, <->,"$\alpha_m$", angle eccentricity=1.7] {angle = pivot1--origo1--bob1};
\end{tikzpicture}
\caption{Parallel line CST geometry. $S$, $D_C$ and $D_A$ denote the source and detector rows. The remaining labels are referenced in the main text. A cone $C_T\subset S^1$ is highlighted in orange. We will refer to $C_T$ later for visualisation in section \ref{microsec1}. Note that we have cropped out part of the left side (left of $O$) of the scanner of figure \ref{fig1} in this picture.}
\label{fig1.1}
\end{figure}

The line and circular arc Radon transforms with full data, are known
\cite{palamodov2011analytic,kalender2006x} to have
inverses that are continuous in some range of
Sobolev norms. Hence with adequate regularization we can reconstruct
an image free of artefacts. With limited data however
\cite{krishnan2014microlocal,truong2019compton}, the solution is
unstable and the image wavefront set (see Definition \ref{WF}) is not
recovered stably in all directions. We will see later in section
\ref{microsec1} through simulation that such data limitations in the
parallel line geometry cause a blurring artefact over a cone in the
reconstruction. There may also be nonlocal artefacts specific to the
geometry (as in \cite{webber2020microlocal}), which we shall discover
later in section \ref{microsec1} in the geometry of figure
\ref{fig1.1}. 

% to yield an exact solution modulo smoothing. 

%hyperplane 

The main goal of this paper is to combine limited datasets in X-ray CT
and CST with new lambda tomography regularization techniques, to
recover the image edges stably in all directions. We focus
particularly on the geometry of figure \ref{fig1}. In lambda
tomography the image reconstruction is carried out by filtered
backprojection of the Radon projections, where the filter is chosen to
emphasize boundaries.  This means that the jump singularities
in the lambda reconstruction have the same location and direction to
those of the target function, but the smooth parts are undetermined. A
common choice of filter is a second derivative in the linear variable
\cite{denisyuk1994inversion, rigaud20183d}. The application of the
derivative filter emphasizes the singularities in the Radon
projections, and this is a key idea behind lambda tomography
\cite{FFRS, FRS,VKK}, and the microlocal view on lambda CT (e.g.,
\cite{denisyuk1994inversion, QO2008:siap,rigaud20183d}). The
regularization penalty we propose aims to minimize the difference
$\|\frac{\mathrm{d}^m}{\mathrm{d}s^m}R(\mu_E-n_e)\|_{L^2(\mathbb{R}\times
S^1)}$ for some $m\geq 1$, where $R$ denotes the Radon line
transform. Therefore, with a full set of Radon projections, the lambda
penalties enforce a similarity in the locations and direction of the
image singularities (edges) of $\mu_E$ and $n_e$. Further
$\frac{\mathrm{d}^m}{\mathrm{d}s^m}R$ for $m\geq 1$ is equivalent to
taking $m-1/2$ derivatives of the object (this operation is continuous
of positive order $m-1/2$ in Sobolev scales), and hence its inverse is
a smoothing operation, which we expect to be of aid in combatting the
measurement noise. In addition, the regularized inverse problem we
propose is linear (similarly to the Tikhonov regularized inverse
\cite[page 99]{hansen2005rank}), which (among other benefits of
linearity) allows for the fast application of iterative least squares
solvers in the solution.

The literature considers joint image reconstruction and
regularization in for example,
\cite{aghasi2013geometric,guven2012multi,rezaee2017fusion,semerci2014tensor,semerci2012parametric,bubba2019learning,tracey2015stabilizing,JR1,JR2,JR4,JR5,JR6,JR7,JR8}.
See also the special issue \cite{JR3} for a more general review of
joint reconstruction techniques. In \cite{JR1} the authors consider
the joint reconstruction from Positron Emission Tomography (PET) and
Magnetic Resonance Imaging (MRI) data and use a Parallel Level Set
(PLS) prior for the joint regularization. The PLS approach (first
introduced in \cite{JR6}) imposes soft constraints on the equality of
the image gradient location and direction, thus enforcing structural
similarity in the image wavefront sets. This follows a similar
intuition to the ``Nambu" functionals of \cite{JR5} and the
``cross-gradient" methods of
\cite{gallardo2004joint,gallardo2007joint} in seismic imaging, the
latter of which specify hard constraints that the gradient cross
products are zero (i.e. parallel image gradients). The methods of
\cite{JR1} use linear and quadratic formulations of PLS, denoted by
Linear PLS (LPLS) and Quadratic PLS (QPLS). The LPLS method will be a
point of comparison with the proposed method. We choose to compare
with LPLS as it is shown to offer greater performance than QPLS in the
experiments conducted in \cite{JR1}. 

In \cite{JR2} the authors consider a class of techniques in joint
reconstruction and regularization, including inversion through
correspondence mapping, mutual information and Joint Total Variation
(JTV). In addition to LPLS, we will compare against JTV as the
intuition of JTV is similar to that of lambda regularization (and
LPLS), in the sense that a structural similarity is enforced in the
image wavefronts. Similar to standard Total Variation (TV), which
favours sparsity in the (single) image gradient, the JTV penalties
(first introduced in \cite{JR7} for colour imaging) favour sparsity in
the joint gradient. Thus the image gradients are more likely to occur
in the same location and direction upon minimization of JTV. The JTV
penalties also have generalizations in colour imaging and
vector-valued imaging \cite{JR8}.

In \cite{webber2019compton} the authors introduce a new toric section
transform $\Tc$ in the geometry of figure \ref{fig1.1}. Here explicit
inversion formulae are derived, but the stability analysis is lacking.
We aim to address the stability of $\Tc$ in this work from a
microlocal perspective. Through an analysis of the canonical relations
of $\Tc$, we discover the existence of nonlocal artefacts in the
inversion, similarly to \cite{webber2020microlocal}.  In
\cite{rezaee2017fusion} the joint reconstruction of $\mu_E$ and $n_e$
is considered in a pencil beam scanner geometry. Here gradient descent
solvers are applied to nonlinear objectives, derived from the
physical models, and a weighted, iterative Tikhonov type penalty is
applied. The works of \cite{bubba2019learning} improve the wavefront
set recovery in limited angle CT using a partially learned, hybrid
reconstruction scheme, which adopts ideas in microlocal analysis and
neural networks. The fusion with Compton data is not considered
however. In our work we assume an equality in the wavefront sets of
$n_e$ and $\mu_E$ (in a similar vein to \cite{semerci2012parametric}),
and we investigate the microlocal advantages of combining Compton and
transmission data, as such an analysis is lacking in the literature.

%We show that the reduction in sinogram space is described by a
%``diamond-shaped" cutoff region (see figure \ref{Fcut}) and hence is a
%special case of \cite{borg2018analyzing}.  

The remainder of this paper is organized as follows.  In section
\ref{microsec}, we recall some definitions and theorems from
microlocal analysis. In section \ref{microsec1} we present a
microlocal analysis of $\mathcal{T}$ and explain the image artefacts
in the $n_e$ reconstruction. Here we prove our main theorem (Theorem
\ref{mainthm}), where we show that the canonical relation $\Cc$ of
$\mathcal{T}$ is 2--1. This implies the existence of nonlocal image
artefacts in a reconstruction from toric section integral data.
Further we find explicit expressions for the nonlocal artefacts and
simulate these by applying the normal operations
$\mathcal{T}^*\mathcal{T}$ to a delta function.  In section
\ref{microsec2} we consider the microlocal artefacts from X-ray
(transmission) data. This yields a limited dataset for the Radon
transform, whereby we have knowledge the line integrals for all $L$
which intersect $S$ and $D_A$ (see figure \ref{fig1}).  We use
the results in \cite{borg2018analyzing} to describe the resulting
artifacts in the X-ray CT reconstruction. In section \ref{results},
we detail our joint reconstruction method for the simultaneous
reconstruction of $\mu_E$ and $n_e$. Later in section \ref{RnD} we
present simulated reconstructions of $\mu_E$ and $n_e$ using the
proposed methods and compare against JTV \cite{JR2} and LPLS
\cite{JR1} from the literature. We also give a comparison to a
separate reconstruction using TV.

% \tc{\tred{Read through carefully to make sure the following changes
% occur: $X\mapsto Y\times X$ and $\vxi\mapsto \vs$ (except for the
% dual variable to $\vx) $and that the definitions are written in
% terms of functions of $(\vy,\vx,\vs)$.  Note the magenta text below.
% Use bold for variables in $\rn$ for $n>1$ but not in $\rr$ (which
% will come up in section 3).  Please check for consistency!}}

\section{Microlocal definitions}\label{microsec} We next provide some
notation and definitions.  Let $X$ and $Y$ be open subsets of
$\rn$.  Let $\Dc(X)$ be the space of smooth functions compactly
supported on $X$ with the standard topology and let $\mathcal{D}'(X)$
denote its dual space, the vector space of distributions on $X$.  Let
$\Ec(X)$ be the space of all smooth functions on $X$ with the standard
topology and let $\mathcal{E}'(X)$ denote its dual space, the vector
space of distributions with compact support contained in $X$. Finally,
let $\Sc(\rn)$ be the space of Schwartz functions, that are rapidly
decreasing at $\infty$ along with all derivatives. See \cite{Rudin:FA}
for more information.

\begin{definition}[{\cite[Definition 7.1.1]{hormanderI}}]
For a function $f$ in the Schwartz space $\Sc(\mathbb{R}^n)$, we define
the Fourier transform and its inverse 
as
\begin{equation}
%\begin{split}
\mathcal{F}f(\vxi)=\int_{\mathbb{R}^n}e^{-i\vx\cdot\vxi}f(\vx)\mathrm{d}\vx,
%\\
\qquad\mathcal{F}^{-1}f(\vx)=(2\pi)^{-n}\int_{\mathbb{R}^n}e^{i\vx\cdot\vxi}f(\vxi)\mathrm{d}\vxi.
%\end{split}
\end{equation}
\end{definition}

We use the standard multi-index notation: if
$\alpha=(\alpha_1,\alpha_2,\dots,\alpha_n)\in \sparen{0,1,2,\dots}^n$
is a multi-index and $f$ is a function on $\rn$, then
\[\partial^\alpha f=\paren{\frac{\partial}{\partial
x_1}}^{\alpha_1}\paren{\frac{\partial}{\partial
x_2}}^{\alpha_2}\cdots\paren{\frac{\partial}{\partial x_n}}^{\alpha_n}
f.\] If $f$ is a function of $(\vy,\vx,\vs)$ then $\partial^\alpha_\vy
f$ and $\partial^\alpha_\vs f$ are defined similarly.

  We identify cotangent
spaces on Euclidean spaces with the underlying Euclidean spaces, so we
identify $T^*(X)$ with $X\times \rn$.

If $\phi$ is a function of $(\vy,\vx,\vs)\in Y\times X\times \rr^N$
then we define $\dd_{\vy} \phi = \paren{\partyf{1}{\phi},
\partyf{2}{\phi}, \cdots, \partyf{n}{\phi} }$, and $\dd_\vx\phi$ and $
\dd_\vs \phi $ are defined similarly. We let $\dd\phi =
\paren{\dd_{\vy} \phi, \dd_{\vx} \phi,\dd_\vs \phi}$.

% \[\begin{gathered}\dd_{\vy} \phi = \paren{\partyf{1}{\phi},
% \partyf{2}{\phi}, \cdots, \partyf{n}{\phi} },\ \dd_\vs \phi =
% \paren{\partsif{1}{\phi},\partsif{2}{\phi}, \cdots, \partsif{N}{\phi}
% }\\ \text{ and }\ \dd\phi(\vx,\vs) = \paren{\dd_{\vy} \phi(\vy,
% \vx,\vs), \dd_{\vx} \phi(\vy,\vx,\vs),\dd_\vs
% \phi(\vy,\vx,\vs)}\in \rn\times \rn\times\rr^N.\end{gathered}\]

The singularities of a function and the directions in which they occur
are described by the wavefront set \cite[page
16]{duistermaat1996fourier}:
\begin{definition}
\label{WF} Let $X$ Let an open subset of $\rn$ and let $f$ be a
distribution in $\mathcal{D}'(X)$.  Let $(\vx_0,\vxi_0)\in X\times
(\mathbb{R}^n\smo)$.  Then $f$ is \emph{smooth at $\vx_0$ in
direction $\vxio$} if   exists a
neighbourhood $U$ of $\vx_0$ and $V$ of $\vxi_0$ such that for every
$\phi\in \Dc(U)$ and $N\in\mathbb{R}$ there exists a constant
$C_N$ such that
\begin{equation}
\left|\Fc(\phi f)(\lambda\vxi)\right|\leq C_N(1+\abs{\lambda})^{-N}.
\end{equation}
The pair $(\vx_0,\vxio)$ is in the \emph{wavefront set,} $\wf(f)$, if
$f$ is not smooth at $\vx_0$ in direction $\vxio$.
\end{definition}
 This definition follows the intuitive idea that the elements of
$\WF(f)$ are the point--normal vector pairs above points of $X$ where
$f$ has singularities. For example, if $f$ is the characteristic
function of the unit ball in $\mathbb{R}^3$, then its wavefront set is
$\WF(f)=\{(\vx,t\vx): \vx\in S^{2}, t\neq 0\}$, the set of points on a
sphere paired with the corresponding normal vectors to the sphere.

%\begin{equation}

%\end{equation}
%That is, 

The wavefront set of a distribution on $X$ is normally defined as a
subset the cotangent bundle $T^*(X)$ so it is invariant under
diffeomorphisms, but we will continue to identify $T^*(X) = X
\times \rn$ and consider $\WF(f)$ as a subset of $X\times
\rn\smo$.

%Let $X$ and $Y$ be open subsets of $\rn$, $m \in\mathbb{R}$.

 \begin{definition}[{\cite[Definition 7.8.1]{hormanderI}}] We define
$S^m(Y\times X\times \mathbb{R}^N)$ to be the set of $a\in \Ec(Y\times
X\times \mathbb{R}^N)$ such that for every compact set $K\subset
Y\times X$ and all multi--indices $\alpha, \beta, \gamma$ the bound
\[
\left|\partial^{\gamma}_{\vy}\partial^{\beta}_{\vx}\partial^{\alpha}_{\vs}a(\vy,\vx,\vs)\right|\leq
C_{K,\alpha,\beta}(1+|\vs|)^{m-|\alpha|},\ \ \ (\vy,\vx)\in K,\
\vs\in\mathbb{R}^N,
\]
holds for some constant $C_{K,\alpha,\beta}>0$. The elements of $S^m$
are called \emph{symbols} of order $m$.
\end{definition}

Note that these symbols are sometimes denoted $S^m_{1,0}$.

\begin{definition}[{\cite[Definition
        21.2.15]{hormanderIII}}] \label{phasedef}
A function $\phi=\phi(\vy,\vx,\vs)\in
\Ec(Y\times X\times\mathbb{R}^N\smo)$ is a \emph{phase
function} if $\phi(\vy,\vx,\lambda\vs)=\lambda\phi(\vy,\vx,\vs)$, $\forall
\lambda>0$ and $\mathrm{d}\phi$ is nowhere zero. A phase function is
\emph{clean} if the critical set $\Sigma_\phi = \{ (\vy,\vx,\vs) \ : \
\mathrm{d}_\vs \phi(\vy,\vx,\vs) = 0 \}$ is a smooth manifold with tangent
space defined by $\mathrm{d} \paren{\mathrm{d}_\vs \phi}= 0$.
\end{definition}
\noindent By the implicit function theorem the requirement for a phase
function to be clean is satisfied if $\mathrm{d}\paren{\mathrm{d}_\vs
\phi}$ has constant rank.

\begin{definition}[{\cite[Definition 21.2.15]{hormanderIII} and
      \cite[Section 25.2]{hormander}}]\label{def:canon} Let $X$ and
$Y$ be open subsets of $\rn$. Let $\phi\in \Ec\paren{Y \times X \times
{\rr}^N}$ be a clean phase function.  In addition, we assume that
$\phi$ is \emph{nondegenerate} in the following sense:
\[\text{$\dd_{\vy,\vs}\phi$ and $\dd_{\vx,\vs}\phi$ are never zero.}\]
  The
\emph{critical set of $\phi$} is
\[\Sigma_\phi=\{(\vy,\vx,\vs)\in Y\times X\times\mathbb{R}^N\smo
: \dd_{\vs}\phi=0\}.\]  The
\emph{canonical relation parametrised by $\phi$} is defined as
\bel{def:Cgenl} \begin{aligned} \Cc=&\sparen{
\paren{\paren{\vy,\dd_{\vy}\phi(\vy,\vx,\vs)};\paren{\vx,-\dd_{\vx}\phi(\vy,\vx,\vs)}}:(\vy,\vx,\vs)\in
\Sigma_{\phi}},
% &\hspace{1.5cm} \vs\in \rr^N\smo,   
\end{aligned}
\end{equation}
\end{definition}

\begin{definition}\label{FIOdef}
Let $X$ and $Y$ be open subsets of $\rn$. A \emph{Fourier integral operator
(FIO)} of order $m + N/2 - n/2$ is an operator $A:\Dc(X)\to
\mathcal{D}'(Y)$ with Schwartz kernel given by an oscillatory integral
of the form
\begin{equation} \label{oscint}
K_A(\vy,\vx)=\int_{\mathbb{R}^N} e^{i\phi(\vy,\vx,\vs)}a(\vy,\vx,\vs) \mathrm{d}\vs,
\end{equation}
where $\phi$ is a clean nondegenerate phase function and $a \in S^m(Y
\times X \times \mathbb{R}^N)$ is a symbol. The \emph{canonical
relation of $A$} is the canonical relation of $\phi$ defined in
\eqref{def:Cgenl}.
\end{definition}

This is a simplified version of the definition of FIO in \cite[Section
2.4]{duist} or \cite[Section 25.2]{hormander} that is suitable for our
purposes since our phase functions are global.  For general
information about FIOs see \cite{duist, hormander, hormanderIII}.

\begin{definition}
\label{defproj} Let $\Cc\subset T^*(Y\times X)$ be the canonical
relation associated to the FIO $A:\mathcal{E}'(X)\to \mathcal{D}'(Y)$.
Then we let $\pi_L$ and $\pi_R$ denote the natural left- and
right-projections of $\Cc$, $\pi_L:\Cc\to T^*(Y)$ and $\pi_R : \Cc\to
T^*(X)$.
\end{definition}

Because $\phi$ is nondegenerate, the projections do not map to the
zero section.  
% 
% We have the following result from \cite{hormander}.
% \begin{proposition}
% \label{prop1}
% Let $\dim(X)=\dim(Y)$. Then at any point in $\Cc$:
% \begin{enumerate}[(i)]
% \item if one of $\pi_L$ or $\pi_R$ is a local diffeomorphism, then the
% other map is a local diffeomorphism (so $\Cc$ is a local canonical
% graph); 
% 
% \item if one of the projections $\pi_R$ or $\pi_L$ is singular at a
% point in $\Cc$, then so is the other. The type of the singularity may
% be different but both projections drop rank on the same set
% \begin{equation}
% \Sigma=\{(\vy,\eta; \vx,\vs)\in \Cc :
% \det(\mathrm{d}\pi_L)=0\}=\{(\vy,\eta; \vx,\vs)\in \Cc : \det
% (\mathrm{d}\pi_R)=0\}.
% \end{equation}
% \end{enumerate}
% \end{proposition}
If a FIO $\Fc$ satisfies our next definition, then $\Fc^* \Fc$ (or
$\Fc^* \phi \Fc$ if $\Fc$ does not map to $\Ec'(Y)$) is a
pseudodifferential operator \cite{GS1977, quinto}.

\begin{definition}\label{def:bolker} Let
$\Fc:\Ec'(X)\to \Dc'(Y)$ be a FIO with canonical relation $\Cc$ then
$\Fc$ (or $\Cc$) satisfies the \emph{semi-global Bolker Assumption} if
the natural projection $\pi_Y:\Cc\to T^*(Y)$ is an embedding
(injective immersion).\end{definition}

\section{Microlocal properties of translational Compton transforms}
\label{microsec1} Here we present a microlocal analysis of the toric
section transform in the translational (parallel line) scanning
geometry. Through an analysis of two separate limited data problems
for the circle transform (where the integrals over circles with
centres on a straight line are known) and using microlocal analysis,
we show that the canonical relation of the toric section transform is
2--1. The analysis follows in a similar way to the work of
\cite{webber2020microlocal}. We discuss the nonlocal artefacts inherent to the toric section inversion in section \ref{nonloc}, and then go on to explain the artefacts due to limited data in section \ref{sect:artifactsTc}.

 We
first define our geometry and formulate the toric section transform of
\cite{webber2019compton} in terms of $\delta$ functions, before
proving our main microlocal theory.

Let $r_m>1$ and define the set of points to be scanned as
\[X:=\{(x_1,x_2)\in \rtwo\st 2-r_m<x_2<1\}.\] Note that $r_m$ controls the
depth of the scanning tunnel as in figures \ref{fig1} and
\ref{fig1.1}.  Let \[ Y:=(0,\infty)\times \mathbb{R}\] then for
$j=1,2$, and $(s,\xo)\in Y$, we define the circles $C_j$ and their
centers $\vc_j$ \bel{def:Cs}
\begin{gathered}
r=\sqrt{s^2+1}, \quad \vc_j(s,x_0)=((-1)^js+x_0,2) \\
C_j(s,x_0)=\{\vx\in\mathbb{R}^2 : |\vx-\vc_j(s,x_0)|^2-s^2-1=0\}.
\end{gathered}\ee
Note that $r=\sqrt{s^2+1}$ is the radius of the circle $C_j$.  The
union of the reflected circles $C_1\cup C_2$ is called a \emph{toric
section}. Let $f\in L^2_0(X)$ be the electron charge density. To define
the toric section transform we first introduce two \emph{circle
transforms} \bel{def:Tj} \mathcal{T}_1
f(s,x_0)=\int_{C_1}f\mathrm{d}s,\ \ \ \ \ \ \ \
\mathcal{T}_2f(s,x_0)=\int_{C_2}f\mathrm{d}s.
\end{equation}
  Now we have the definition of the \emph{toric section transform}
\cite{webber2019compton} \bel{def:T} \mathcal{T}f(s,x_0)=\int_{C_1\cup
C_2}f\mathrm{d}s=\Tc_1(f)(s,x_0)+\Tc_2(f)(s,x_0) \ee where
$\mathrm{d}s$ denotes the arc element on a circle and
$(s,\xo)\in Y$.

% \tc{The Palamodov article is helpful here, but the integral using
% $\delta(|\vx-\vc_j(s,x_0)|^2-s^2-1)$ is different from
% $\delta\paren{|\vx-\vc_j(s,x_0)|-\sqrt{s^2+1}}$ just as an integral
% $\int f(t)g(t)\,dt$ is not the same as $\int f(t^2)g(t^2)\,dt$ so I
% included the change of variable factor.  This doesn't contradict
% Palamodov since these would be two different transforms in his
% case--differing by a factor that can be taken into the measure.
% Please check if I did the change of variable right.  The difference is
% just a 1/2 from what was there originally.  Of course, your choice of
% the phase is much better than using
% $\paren{|\vx-\vc_j(s,x_0)|-\sqrt{s^2+1}}$.}

We express $\mathcal{T}$ in terms of delta functions as is done for
the generalized Funk-Radon transforms studied by Palamodov.  
\cite{palamodov2012uniform}
\begin{equation}
\begin{split}
\mathcal{T}f(s,x_0)&=\mathcal{T}_1f(s,x_0)+\mathcal{T}_2f(s,x_0)\\
&=\frac{1}{2r}\sum_{j=1}^2\int_{\mathbb{R}^2}\delta(|\vx-\vc_j(s,x_0)|^2-s^2-1)f(\vx)\mathrm{d}\vx\\
%&=\frac{1}{2r}\sum_{j=1}^2\int_{\mathbb{R}^2}\delta(|\vx-((-1)^js+x_0,2)|^2-s^2-1)f(\vx)\mathrm{d}\vx\\
&=\frac{1}{2r}\sum_{j=1}^2\int_{-\infty}^{\infty}\int_{\mathbb{R}^2}e^{-i\sigma(|\vx-((-1)^js+x_0,2)|^2-s^2-1)}f(\vx)\mathrm{d}\vx\mathrm{d}\sigma.
\end{split}
\end{equation}
Note that the factor in front of the integrals comes about using the
change of variables formula and that $\Tc_j f = \int
\delta\paren{|\vx-\vc_j(s,x_0)|-\sqrt{s^2+1}} f(x)\,\dd\vx$. So the
toric section transform is the sum of two FIO's with phase functions
$$\phi_j(s,x_0,\vx,\sigma)=\sigma(|\vx-((-1)^js+x_0,2)|^2-s^2-1)$$ for
$j=1,2$.  Our distributions $f$ are supported away from the
intersection points of $C_1$ and $C_2$, and hence we can consider the
microlocal properties of $\mathcal{T}_1$ and $\mathcal{T}_2$
separately to describe the microlocal properties of $\mathcal{T}$.

\begin{proposition}
For $j=1,2$, the circle transform $\mathcal{T}_j$ is an FIO or order
$-1/2$ with canonical relation
\begin{equation}\label{def:C1C2}
\begin{aligned}
\Cc_j=\Big\{\big(&\paren{s,x_0,(-1)^{j-1}\sigma
(x_1-x_0),-\sigma((-1)^{j-1}s+x_1-x_0)} ;
\paren{\vx,-\sigma(\vx-\vc_j(s,x_0))}\big) : \\
&\quad (s,x_0)\in Y, \sigma\in \mathbb{R}\smo, \vx\in C_j(s,x_0)\cap
\{x_2<1\}\Big\}.\\
\end{aligned}
\end{equation}
Furthermore $\Cc_j$ satisfies the semi-global Bolker assumption for
$j=1,2$.
\end{proposition}

\begin{proof} First, one can check that $\phi_j$ and $\Tc_j$ both satisfy
  the restrictions in Definition \ref{FIOdef} so $\Tc_j$ is a FIO.
Using this definition again and the fact that its symbol is order zero
\cite{quinto}, one sees that it has order $-1/2$.

%NOTE:  I wasn't sure what this did so i commented it:
% We have
% \begin{equation}
% \begin{split}
%   \vx,\sigma)&=\sigma(|\vx-((-1)^js+x_0,2)|^2-s^2-1)\\
%   &=\sigma\left(|\vx|^2+2x_1((-1)^{j-1}s-x_0)-4x_2+2(-1)^jsx_0+x_0^2+3\right),
% \end{split}
% \end{equation}
%where $\vx=(x_1,x_2)$. 

A straightforward calculation using Definition \ref{def:canon} shows
that the canonical relation of $\mathcal{T}_j$ is as given in
\eqref{def:C1C2}.  Note that we have absorbed a factor of $2$ into
$\sigma$ in this calculation.  Global coordinates on $\Cc_j$ are given
by \bel{def:coords}\begin{aligned}(s,x_0,x_1,\sigma)\mapsto&
\big(s,x_0,(-1)^{j-1}\sigma (x_1-x_0),-\sigma((-1)^{j-1}s+x_1-x_0) ;\\
&\qquad(x_1,x_2),-\sigma((x_1,x_2)-\vc_j(s,x_0))\big)\\
&\qquad\qquad\text{where $x_2 =
2-\sqrt{s^2+1-(x_1-(x_0+(-1)^js))^2}$}\end{aligned}\ee because
$x_2<1$.  Recall that $\vc_j$ is given in \eqref{def:Cs}.

We now show that $\Cc_j$ satisfies the semiglobal Bolker assumption by
finding a smooth inverse in these coordinates to the projection
$\Pi_L:\Cc_j\to T^*(Y)$.  Let $\lambda= (s,x_0, \tau_1,\tau_2)\in
\Pi_L\paren{\Cc_j}$.  We solve for $x_1$ and $\sigma$ in the equation
$\Pi_L(s,x_0,x_1,\sigma)=\lambda$.  Then, $s$ and $x_0$ are known as
are
\bel{tau1 tau2}
%\begin{aligned} 
\tau_1 = (-1)^{j-1}\sigma(x_1-x_0)
%\\
\qquad
\tau_2 = -\sigma((-1)^{j-1}s+x_1-x_0).
%\end{aligned}
\ee A straightforward linear algebra exercise shows that the unique
solutions for $\sigma$ and $x_1$ are \bel{def:sigmaY}\sigma =
\frac{(-1)^{j}\tau_2-\tau_1}{s}, \qquad x_1 =
\frac{s\tau_1}{(-1)^j\tau_1-\tau_2} + x_0\end{equation} This gives a
smooth inverse to $\Pi_L$ on the image $\Pi_L\paren{\Cc_j}$ and
finishes the proof.
\end{proof}

Because $\Cc_j$ satisfies the Bolker Assumption, the composition
$\Cc^*_j\circ\Cc_j\subset \Delta$, where $\Delta$ is the diagonal in
$T^* (X)$.  Hence in a reconstruction from circular integral data with
centres on a line we would not expect to see image artefacts for
functions supported in $x-2>0$ unless one uses a sharp cutoff on the
data.

The canonical relation $\Cc$ of $\mathcal{T}$ can be written as the
disjoint union $\Cc=\Cc_1\cup\Cc_2$ since $(C_1(s,\xo)\cap C_2(s,\xo))\cap
\text{supp}(f)=\emptyset$ for any $(s,\xo)\in Y$.

%And hence the canonical relation $\Cc$ is
%either 1--1 or 2--1, as $\Cc_1$ and $\Cc_2$ are both 1--1.

For convenience, we will sometimes label the coordinate $x_0$ in
\eqref{def:coords} as $\xoo$ it is associated with $\Cc_1$ and $\xot$
if it is associated with $\Cc_2$.

\begin{theorem}
  \label{mainthm} For $j=1,2$, the projection $\pi_R:\Cc_j\to T^*(X)$
is bijective onto the set \bel{def:D} D=\sparen{(\vx,\boldsymbol{\xi})\in T^*(X)\st
\xi_2\neq 0}.\ee In addition, $\pi_R:\Cc\to T^*(X)$ is two-to one onto
$D$.
\end{theorem}

\begin{proof}
  Let $\mu= (\vx;\boldsymbol{\xi})\in T^*(X)\smo$ and let $\vx = (x_1,x_2)$ and
$\boldsymbol{\xi}=(\xi_1,\xi_2)$.  If $\mu\in \pi_R(\Cc_j)$ for either $j=1$ or
$j=2$, then $\xi_2\neq 0$ by \eqref{def:C1C2} since $x_2<2$.  For the
rest of the proof, assume $\mu$ is in the set $D$ given by
\eqref{def:D}

We will now describe the preimage of $\mu$ in $\Cc_j$.  The covector
$\mu$ is conormal to a unique circle centered on $x_2=2$, and its
center is on the line through $\vx$ and parallel $\boldsymbol{\xi}$.  If the center
has coordinates $(c,2)$, then a calculation shows that $c$ is given by
\bel{def:center} c=c(\vx,\boldsymbol{\xi})=x_1-\frac{\xi_1(x_2-2)}{\xi_2}.\ee Using
this calculation, one sees that the radius of the circle and
coordinate $s$ are given by \bel{def:rs} r=r(\vx,\boldsymbol{\xi}) =
\frac{(2-x_2)\abs{\boldsymbol{\xi}}}{\abs{\xi_2}},\quad s= s(\vx,\boldsymbol{\xi}) =
\sqrt{r^2-1}\ee and the coordinate $\xoj$ is given by
\bel{def:xoj}\xoj=\xoj(\vx,\boldsymbol{\xi})= x_1+\frac{\xi_1(2-x_2)}{\xi_2}
+(-1)^{j-1}s\ \ \text{for $j=1,2$}. \ee A straightforward calculation
shows that \bel{def:sigmaX}\sigma=\sigma(\vx,\boldsymbol{\xi}) =
\frac{-\xi_2}{2-x_2}.\ee This gives the coordinates \eqref{def:coords}
on $\Cc_j$ and shows that $\pi_R:\Cc_j\to D$ is injective with smooth
inverse.

Now, we consider the projection from $\Cc$.  Given $(\vx,\boldsymbol{\xi})\in D$,
our calculations show that the preimage in $\Cc$, in coordinates
\eqref{def:coords} is given by two \emph{distinct} points
\[(s(\vx,\boldsymbol{\xi}),\xoj\xxi, x_1,\sigma\xxi)\  \text{ for $j=1,2$ }.\]  The coordinates are
given by \eqref{def:rs}, \eqref{def:xoj} and \eqref{def:sigmaX}
respectively.
\end{proof}

The abstract adjoint $\Tc_j^t$ cannot be composed with $\Tc_i$
for $i=1,2$, because the support of $\Tc_i f$ can be unbounded in $r$,
even for $f\in \Ec'(X)$ and $\Tc_j^t$ is not defined for such
distributions.  Therefore, we introduce a smooth cutoff function.
Choose $r_M>2$ and let $\psi(s)$ be a smooth compactly supported function
equal to one for $s\in \left[1,\sqrt{1-r^2_M}\right]$ and define \bel{def:Tcstar} \Tc_j^* g =
\Tc^t_j(\psi g)\ee for all $g\in \Dc'(Y)$ because our bound on $r$
introduces a bound on $x_0$ so the integral is over a bounded set for
each $\vx\in X$.

\subsection{The nonlocal artefacts} \label{nonloc}
Now, we can state our next theorem, which describes the artifacts that
can be added to the reconstruction using the normal operator,
$\Tc^*\Tc$.

\begin{theorem}\label{mainthm2}If $f\in \Ec'(X)$ then 
  \bel{wfT*T} \wf\paren{\Tc^*\Tc f} \subset \paren{\wf(f)\cap D} \cup
\Lambda_{12}(f)\cup \Lambda_{21}(f)\ee where $D$ is given by
\eqref{def:D}, and the sets $\Lambda_{ij}$ are given for $\xxi\in D$
by \bel{def:Lambdaij} \Lij(f) = \sparen{\lij\xxi\st \xxi\in \wf(f)\cap
D}\ee where the functions $\lot$ and $\lto$ are given by
\eqref{def:lot} and \eqref{def:lto} respectively.  Note that the
functions $\lij$ are  defined for only some  $(\vx,\boldsymbol{\xi})\in D$ and
singularities at other points do not generate artifacts.
\end{theorem}

Therefore, $\Tc^*\Tc$ recovers most singularities of $f$, as indicated
in the first term in \eqref{wfT*T}, but it adds two sets of
nonlocal singularities, as given by $\Lot(f)$ and $\Lto(f)$.  Note that, even
if $\Tc_j^*$ and $\Tc_j$ are both elliptic above a covector
$(\vx,\boldsymbol{\xi})$, artifacts caused by other points could mask singularities
of $f$ that ``should'' be visible in $\Tc^*\Tc f$.

\begin{proof}
Let $f\in \Ec'(X)$.  By the H\"ormander-Sato Lemma \cite[Theorem
8.2.13]{hormanderI} We have the expansion
\begin{equation}\label{compositions}
\begin{aligned}
\wf(\Tc^*\Tc(f))\subset &\paren{\Cc^*\circ \Cc}\circ \wf(f)\\
&\quad=\bparen{(\Cc_1^*\circ\Cc_1) \cup
(\Cc_2^*\circ\Cc_2)}\circ\wf(f)\\
&\quad\quad \cup \ \paren{\Cc_2^*\circ\Cc_1}\circ\wf(f)
%\\
%&\quad\quad\quad
\ \cup \ \paren{\Cc_1^*\circ\Cc_2}\circ \wf(f)
\end{aligned}
\end{equation}
The first term in brackets in \eqref{compositions} is
$\sparen{(\vx,\boldsymbol{\xi};\vx,\boldsymbol{\xi})\st (\vx,\boldsymbol{\xi})\in D}\circ\wf(f) = \wf(f)\cap
D$.  This proves the first part of the inclusion \eqref{wfT*T}.

We now analyze the other two terms to define the functions $\lij$ and
finish the proof. Let $\xxi\in \wf(f)\cap D$.  First, consider
$\lot(\vx,\boldsymbol{\xi})=\Cc_2^*\circ \Cc_1\circ(\vx,\boldsymbol{\xi})$.\footnote{For convenience,
we will abbreviate the set theoretic composition $\Cc_i\circ
\sparen{\xxi}$ by $\Cc_i\circ \xxi$.} Using the calculations in the
proof of Theorem \ref{mainthm} one sees that $\Cc_1\circ \xxi$ is
given by \bel{C1 xxi} \begin{gathered}(s,\xo,\tau_1,\tau_2)\ \text{
where }\ s=s\xxi = \sqrt{\frac{(2-x_2)^2 \abs{\boldsymbol{\xi}}^2}{\xi_2^2} -1}\\
\xo = \xoo(\vx,\boldsymbol{\xi})= x_1+\frac{\xi_1(2-x_2)}{\xi_2} +s
%\\
\qquad \sigma=\sigma(\vx,\boldsymbol{\xi})= \frac{-\xi_2}{2-x_2}\\
\tau_1 = \sigma(x_1-x_0)\qquad
\tau_2 = -\sigma(s+x_1-x_0)
\end{gathered}
\ee where we have taken these from the proof of Theorem \ref{mainthm}.
To find $\Cc_2^*\circ \Cc_1\circ\xxi$ we calculate the composition of
the covector described in \eqref{C1 xxi} with $\Cc_2^*$.  Note that
the values of $\xo$ and $s$ are the same in both calculations and are
given by \eqref{C1 xxi}.  After using \eqref{def:sigmaY} and that
$\frac{\xi_1(2-x_2)}{\xi_2}=x_0-s-x_1$, one sees that
\bel{def:lot}\begin{aligned} \lot\xxi& = \paren{(y_1,y_2),\boldsymbol{\eta}} \
\text{ where }\\
y_1=y_1\xxi&=\frac{s(x_1-\xo)}{2(x_1-\xo)+s}+\xo\\
y_2 =y_2\xxi &= 2-\sqrt{\frac{(2-x_2)^2\abs{\boldsymbol{\xi}}^2}{\xi_2^2} -
(y_1-(\xo+s))^2}
\\
\boldsymbol{\eta}& = \paren{-2\xi_1 - \frac{s\xi_2}{2-x_2}}(\vy-\vc_2(s\xxi\xoo\xxi))
\end{aligned}
\ee where $\xo=\xoo\xxi$ and $s=s\xxi$ are given in \eqref{C1 xxi} and
$\boldsymbol{\eta}$ is calculated using the expression \eqref{def:sigmaY} with
$j=2$.

Note that the function $\lot$ is  defined for only some
$(\vx,\boldsymbol{\xi})\in D$; for example if the argument for the square root
defining $y_2\xxi$ is negative, then $y_2\xxi$ is not defined and the
point $\xxi$ will not generate artifacts in $\Lambda_{12}$.

A similar calculation shows for $\yeta\in D$ that
\bel{def:lto}\begin{aligned} \lto\yeta& = ((x_1,x_2),\boldsymbol{\xi}) \ \text{
where }\\
x_1=x_1\yeta&=\frac{s(y_1-\xo)}{-2(y_1-\xo)+s}+\xo\\
x_2 =x_2\yeta &= 2-\sqrt{s^2\yeta+1 -(x_1-(\xo-s))^2}\\
\boldsymbol{\xi}& = \paren{2\eta_1 -
\frac{s\eta_2}{2-y_2}}(\vx-\vc_1(s\yeta,\xot\yeta)) \end{aligned}
\ee
where 
\[%\begin{aligned}
s=s\yeta = \sqrt{\frac{(2-y_2)^2 \abs{\boldsymbol{\eta}}^2}{\eta_2^2} -1}
%\\
\qquad 
\xo = \xot\yeta= y_1+\frac{\eta_1(2-y_2)}{\eta_2}-s.
%\end{aligned}
\]
Note that the function $\lto$ is not defined for all $\yeta\in D$,
and other points $\yeta$ do not generate artifacts.  This is for the
same reason as for $\lot$.
\end{proof}

\begin{remark}\label{rem:strength} The artefacts caused by a
singularity of $f$ are as strong as the reconstruction of that
singularity.  To see this, first note that each $\Tc_j^*\Tc_i$
smooths of order one in Sobolev scale since it an FIO of order $-1$
\cite[Theorem 4.3.1]{Ho1971}. 

  The visible singularities in the reconstruction come from the
compositions $\Tc_1^*\Tc_1$ and $\Tc_2^*\Tc_2$ since these are
pseudodifferential operators of order $-1$. The artefacts come from
the ``cross'' compositions $\Tc_2^* \Tc_1$ and $\Tc_1^* \Tc_2$, and
they are FIO of order $-1$.  Therefore, since the terms that preserve
the real singularities of $f$, $\Tc_i^*\Tc_i$, $i=1,2$, are also of
order $-1$, $\Tc^*\Tc$ smooths each singularity of $f$ by one order in
Sobolev scale \emph{and} the composition $\Tc^*_2\Tc_1$ (corresponding
to the artifact $\lot$, if defined at this covector) can create
an artefact from that singularity that are also one order smoother than
that singularity, and similarly with the composition $\Tc^*_1\Tc_2$.  

   Second, our results are valid, not only for the normal operator
$\Tc^*\Tc$ but for any filtered backprojection method $\Tc^*P
\Tc$ where $P$ is a pseudodifferential operator.  This is true since
pseudodifferential operators have canonical relation $\Delta$ and they
do not move singularities, so our microlocal calculations are the
same.  If $P$ has order $k$, then $\Tc^*P\Tc$ decreases the Sobolev
order of each singularity of $f$ by order $(k-1)$ in Sobolev norm and
can create an artefact from that singularity of the same order.
\end{remark}

\subsection{Artifacts for $\Tc^*\Tc$ due to limited data}\label{sect:artifactsTc}
In practice we do not have access to $\mathcal{T}f(s,x_0)$ for all
$s\in (0,\infty)$ (or $r\in(1,\infty)$) and $x_0\in\mathbb{R}$, and
will have knowledge of $x_0\in(-a,a)$ and $r\in (1,r_M)$ for some
$a>0$ (see figures \ref{fig1} and \ref{fig1.1}) and maximum radius
$r_M>1$. 

We now evaluate which wavefront directions $\xxi$ will be visible from
this limited data.  Let us consider the pair $(\vx,\boldsymbol{\xi})\in
C_2(s,x_0)\times S^1$ and let $\beta$ be the angle of $\boldsymbol{\xi}$ from the
vertical as depicted in figure \ref{fig1.1}. Then
$\vc_2(s,x_0)=((2-x_2)\tan\beta+x_1,2)$ and
$$|\vx-\vc_2(s,x_0)|^2=r^2\implies (1+\tan^2\beta)(2-x_2)^2=r^2\implies \tan\beta=\sqrt{\frac{r^2}{(2-x_2)^2}-1}.$$
Let $\beta_m=\beta_m(\vx)\in (0,\pi/2)$ be defined by \bel{def:betam}
\tan\beta_m=\sqrt{\frac{r_M^2}{(2-x_2)^2}-1}\ee (noting that we only
consider $\vx$ such that $1>x_2>2-r_M$). Then the maximum directional
coverage of the singularities (wavefront set) at a given $\vx\in X$
which are resolved by the Compton data are described by the open cone
of $\boldsymbol{\xi}\in S^1$ \bel{def:C_T}C_T=\{\pm(\sin\beta,-\cos\beta) :
-\beta_m<\beta<\beta_m\},\ee and the opening angle of the cone depends
on the depth of $\vx$ (i.e. $x_2$). See figure \ref{fig1.1}. The cone
$C_T$ illustrated corresponds to the case when $\beta=\beta_m$.

In all of our numerical experiments, we
set the tunnel height as $r_m-1=6$ and the detector line width is
$2a=8$. We let $r_M>r_m$ be large enough to penetrate the entire
scanning tunnel (up to the line $\{x_2=2-r_m\}$ as highlighted in
figures \ref{fig1} and \ref{fig1.1}), so as to imply a unique
reconstruction \cite{webber2020microlocal}. Specifically we set the
maximum radius $r_M=9$ and simulate $\Tc(r,x_0)$ for $r\in \{1+0.02j :
1\leq j\leq 400\}$ and $x_0\in \{-4+0.04j : 1\leq j\leq 200\}$.
Further the densities considered are represented on
$[-2,2]\times[-3,1]$ ($200\times 200$ pixel grid) in the reconstructions shown. The machine design considered is such that for any $\vx\in[-2,2]\times
[-1.5,1]$ we have the maximal directional coverage in $C_T$ allowed
for the limited $r<r_M$ (see figure \ref{fig:directions}). With the
exception of the horizontal bar phantom depicted in figure \ref{BF1}, all objects considered for reconstruction are
approximately in this region.
\begin{figure}[!h]
\begin{subfigure}{0.32\textwidth}
\includegraphics[width=0.9\linewidth, height=4cm]{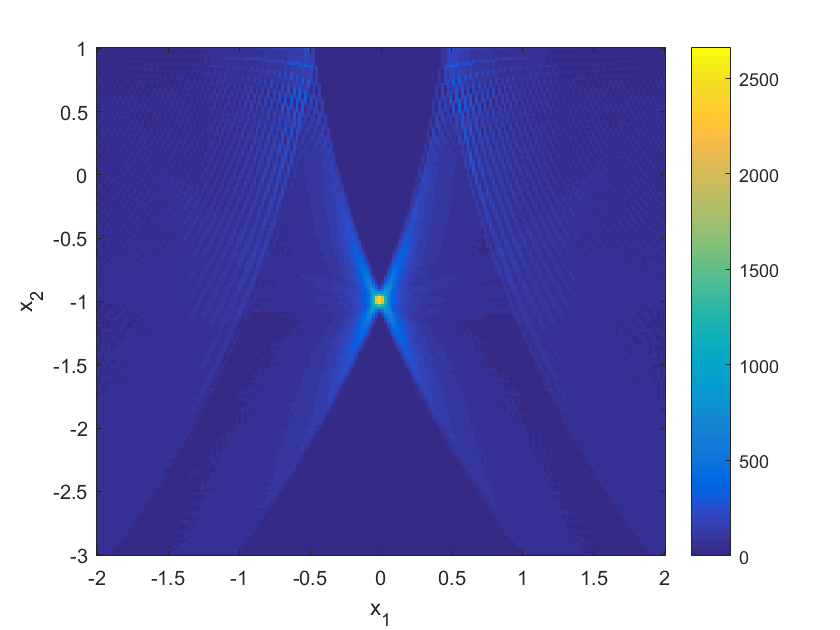}
\subcaption{$\Tc^*\Tc\delta$.}\label{FC1:A}
\end{subfigure}
\begin{subfigure}{0.32\textwidth}
\includegraphics[width=0.9\linewidth, height=4cm]{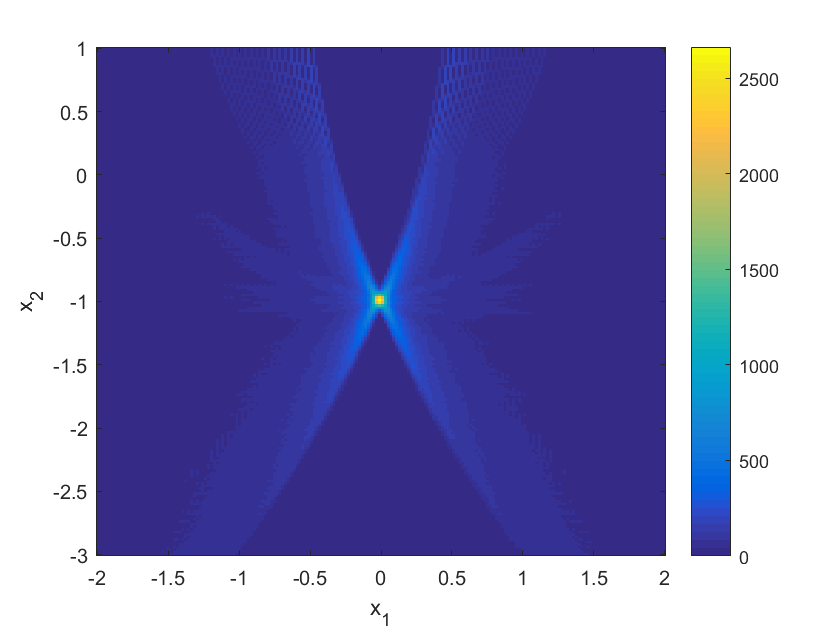}
\subcaption{$(\Tc_1^*\Tc_1+\Tc_2^*\Tc_2)\delta$.}\label{FC1:B}
\end{subfigure}
\begin{subfigure}{0.32\textwidth}
\includegraphics[width=0.9\linewidth, height=4cm]{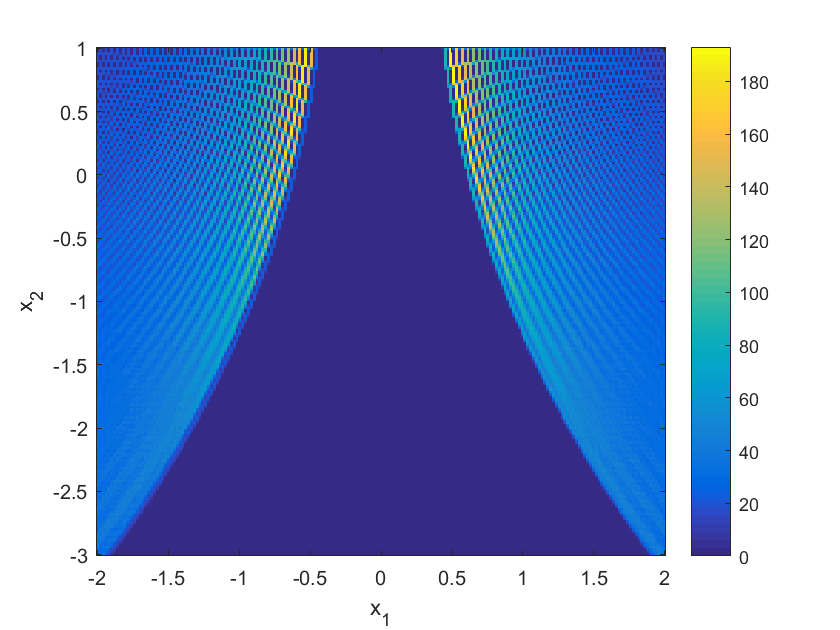}
\subcaption{$(\Tc_1^*\Tc_2+\Tc_2^*\Tc_1)\delta$.}\label{FC1:C}
\end{subfigure}
\begin{subfigure}{0.32\textwidth}
\includegraphics[width=0.9\linewidth, height=4cm]{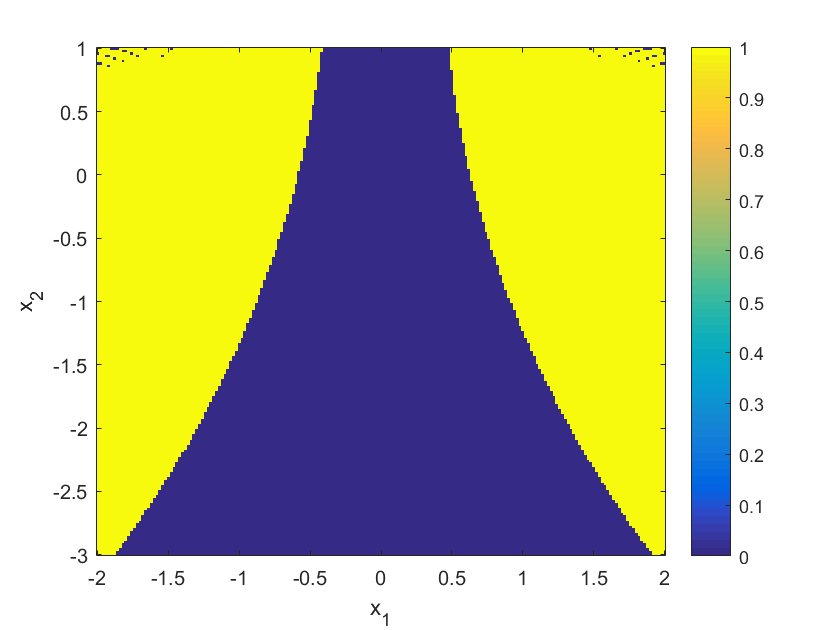}
\subcaption{$\chi_{S_{12}\cup S_{21}}$.}\label{FC1:D}
\end{subfigure}
\begin{subfigure}{0.32\textwidth}
\includegraphics[width=0.9\linewidth, height=4cm]{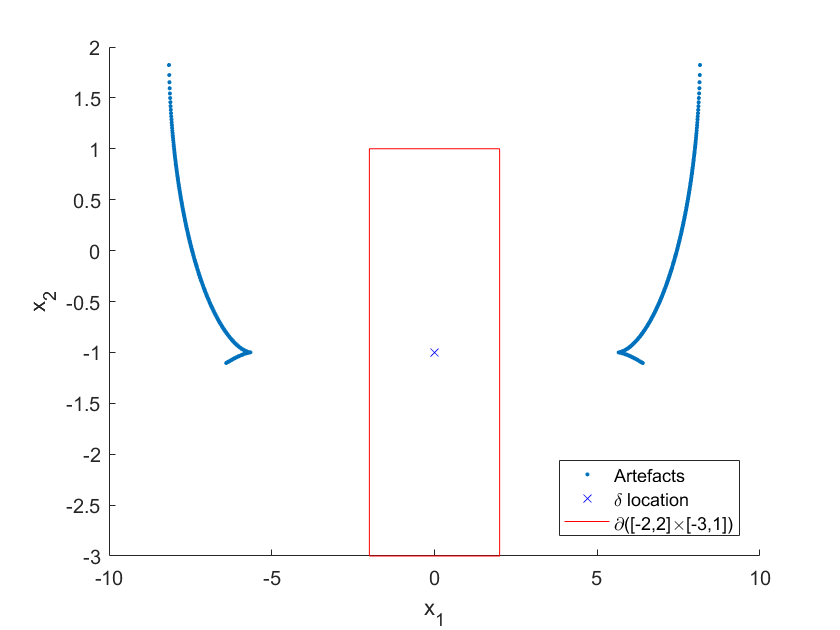}
\subcaption{$\Lambda_{ij}$ artefacts.}\label{FC1:E}
\end{subfigure}
\begin{subfigure}{0.32\textwidth}
\includegraphics[width=0.9\linewidth, height=4cm]{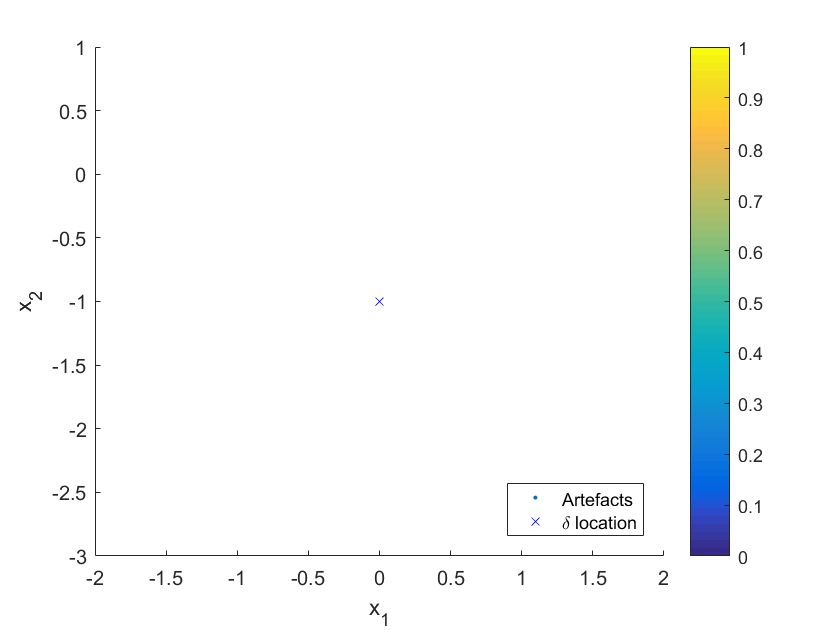}
\subcaption{$\Lambda_{ij}$ artefacts on $[-2,2]\times[-3,1]$.}\label{FC1:F}
\end{subfigure}
\caption{$\Tc^*\Tc\delta$ (the $\delta$ function is centered at
$\vO=(0,-1)$) images with the predicted artefacts due to the limited
data backprojection (on $S_{12}\cup S_{21}$) and those induced by
$\Lambda_{12}$ and $\Lambda_{21}$.} \label{FC1}
\end{figure}

To demonstrate the
artefacts, we apply a discrete form of $\mathcal{T}^*\mathcal{T}$ to a
delta function. We have the expansion
\begin{equation}
%\begin{split}
\mathcal{T}^*\mathcal{T}\ =\ (\Tc_1+\Tc_2)^*(\Tc_1+\Tc_2)
%\\
\ =\ \Tc_1^*\Tc_1+\Tc_2^*\Tc_2+\Tc_1^*\Tc_2+\Tc_2^*\Tc_1.
%\end{split}
\end{equation}
Using equations \eqref{C1 xxi}, \eqref{def:lot}, and \eqref{def:lto}
one can show for $g\in L^2(Y)$ that the backprojection operators $\Tc_j^*$, $j=1,2$
can be written
\begin{equation}\label{Tj*}
\Tc_j^*g(\vx)=\int_{-\beta_m}^{\beta_m}g\left(
\sqrt{r^2-1},x_1+(-1)^{j-1}\sqrt{r^2-1}+r\sin\beta\right)\bigg|_{r=\frac{(2-x_2)}{\cos\beta}}\mathrm{d}\beta.
\end{equation}
Note that we are not restricting $\xo$ to $[-a,a]$  but we are
restricting $s$ to $\paren{0,\sqrt{r_M^2-1}}$, and hence the cutoff function $\psi$ of equation \ref{def:Tcstar} is equal to one on the bounds of integration.

% and
% \begin{equation}\label{T2*}
% \Tc_2^*g(\vx)=\int_{-\beta_m}^{\beta_m}g\left(
% \sqrt{r^2-1},x_1-\sqrt{r^2-1}+r\sin\beta\right)\bigg|_{r=\frac{(2-x_2)}{\cos\beta}}\mathrm{d}\beta.
% \end{equation}

Now, let $f$ be a delta function at $\vy$. We calculate the artifacts 
\begin{equation}
\Tc_1^*\Tc_2f(\vx)\neq 0 \iff \exists \beta\in[-\beta_m,\beta_m]\ \
\text{s.t.}\ \ |\vy-\vc_2(s,x_0)|=r,
\end{equation}
where $r=\frac{2-x_2}{\cos\beta}$, $s=\sqrt{r^2-1}$ and
$x_0=x_1+s+r\sin\beta$. Similarly
\begin{equation}
\Tc_2^*\Tc_1f(\vx)\neq 0 \iff \exists \beta\in[-\beta_m,\beta_m]\ \
\text{s.t.}\ \ |\vy-\vc_1(s,x_0)|=r,
\end{equation}
where $r=\frac{2-x_2}{\cos\beta}$, $s=\sqrt{r^2-1}$ and
$x_0=x_1-s+r\sin\beta$. Hence the only contributions to the
  backprojection from $\Tc_1^*\Tc_2$ and $\Tc_2^*\Tc_1$   are on the
  following sets:
\begin{equation}
S_{12}=\{\vx : \exists \beta\in[-\beta_m,\beta_m]\ \ \text{s.t.}\ \
|\vy-\vc_2(s,x_0)|=r\}
\end{equation}
where $r=\frac{2-x_2}{\cos\beta}$ and $x_0 = x_1+s+r\sin\beta$
and
\begin{equation}
S_{21}=\{\vx : \exists \beta\in[-\beta_m,\beta_m]\ \ \text{s.t.}\ \
|\vy-\vc_1(s,x_0)|=r\}.
\end{equation}
where $r=\frac{2-x_2}{\cos\beta}$ and $x_0=x_1-s+r\sin\beta$.  This
means that all $\Lij$ artifacts will be in these sets.  Note that
besides the $\Lij$ artifacts shown in figure \ref{FC2:E} and
\ref{FC2:F} there are limited data artifacts caused by circles meeting
$\vy$ of radius $r_M$ (figures \ref{FC2:A}-\ref{FC2:C}) and these are
of higher strength in Sobolev norm.

\begin{figure}[!h]
%\centering
\begin{subfigure}{0.32\textwidth}
\includegraphics[width=0.9\linewidth, height=4cm]{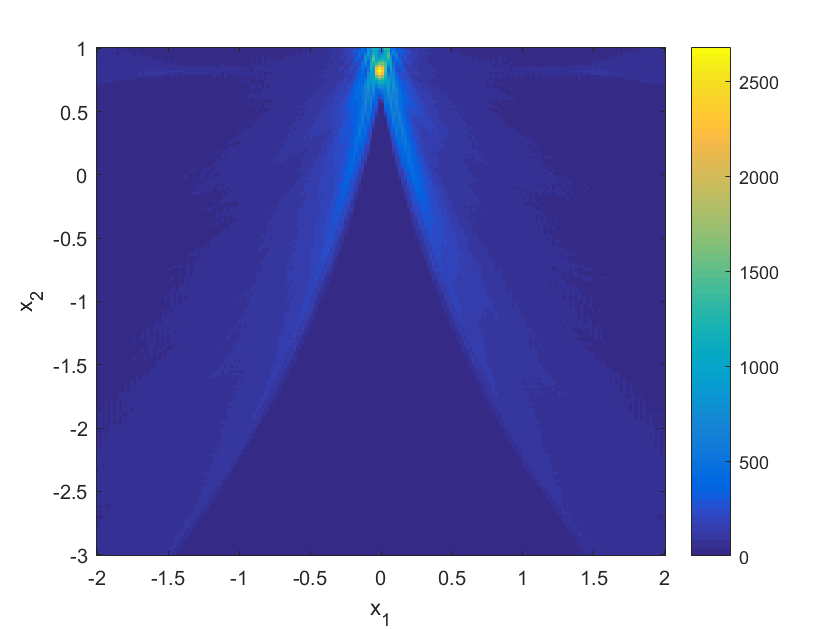}
\subcaption{$\Tc^*\Tc\delta$.}\label{FC2:A}
\end{subfigure}
\begin{subfigure}{0.32\textwidth}
\includegraphics[width=0.9\linewidth, height=4cm]{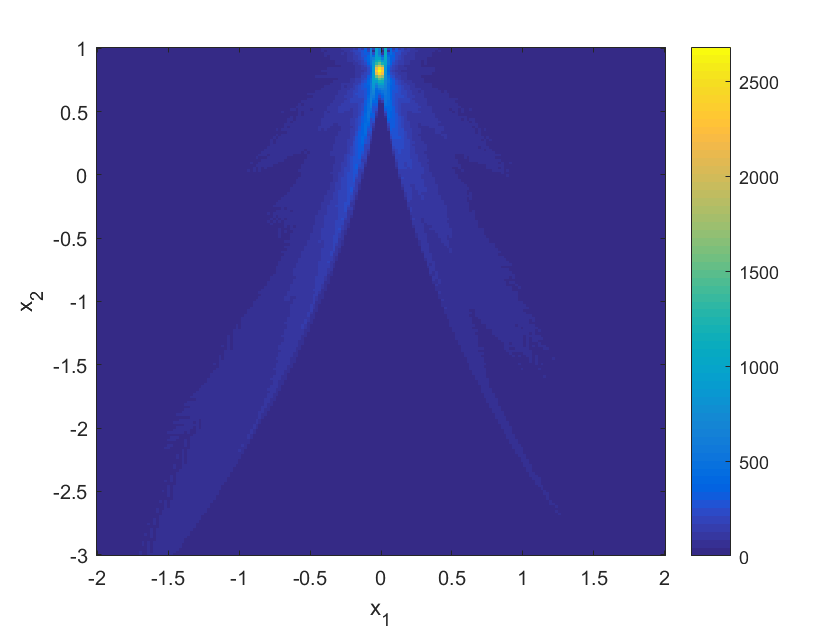}
\subcaption{$(\Tc_1^*\Tc_1+\Tc_2^*\Tc_2)\delta$.}\label{FC2:B}
\end{subfigure}
\begin{subfigure}{0.32\textwidth}
\includegraphics[width=0.9\linewidth, height=4cm]{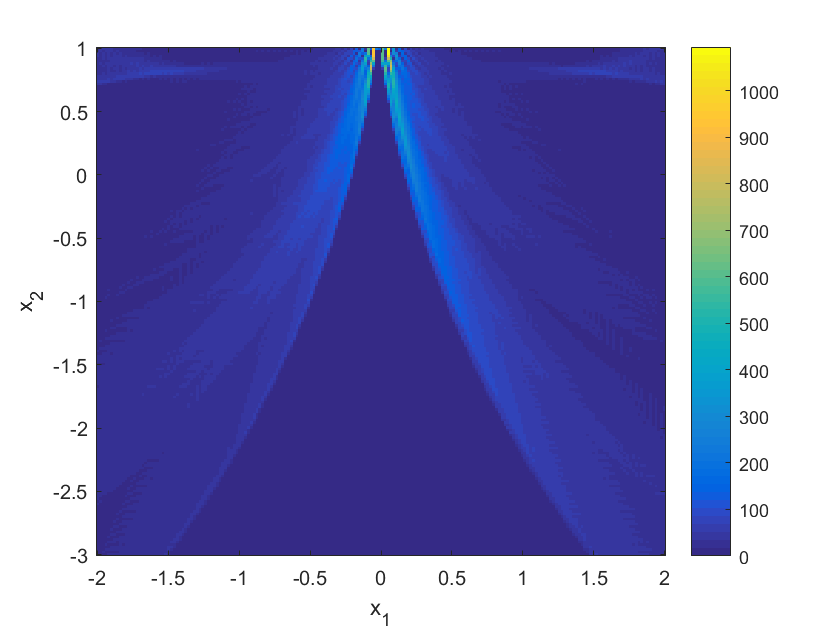}
\subcaption{$(\Tc_1^*\Tc_2+\Tc_2^*\Tc_1)\delta$.}\label{FC2:C}
\end{subfigure}
\begin{subfigure}{0.32\textwidth}
\includegraphics[width=0.9\linewidth, height=4cm]{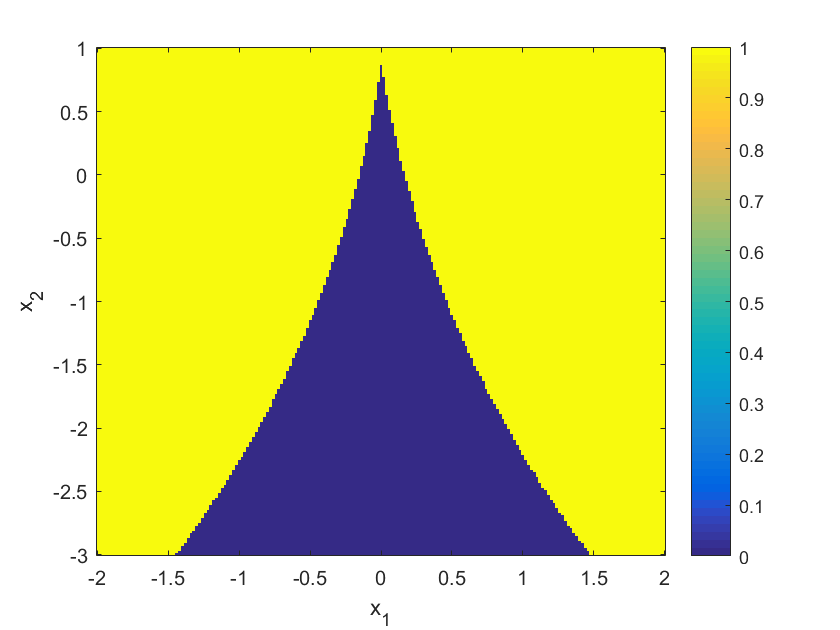}
\subcaption{$\chi_{S_{12}\cup S_{21}}$.}\label{FC2:D}
\end{subfigure}
\begin{subfigure}{0.32\textwidth}
\includegraphics[width=0.9\linewidth, height=4cm]{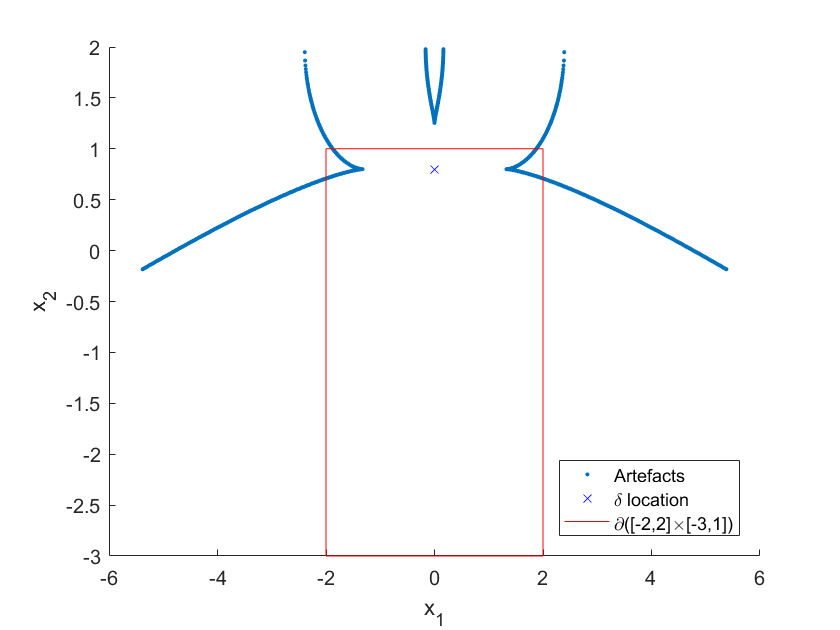}
\subcaption{$\Lambda_{ij}$ artefacts.}\label{FC2:E}
\end{subfigure}
\begin{subfigure}{0.32\textwidth}
\includegraphics[width=0.9\linewidth, height=4cm]{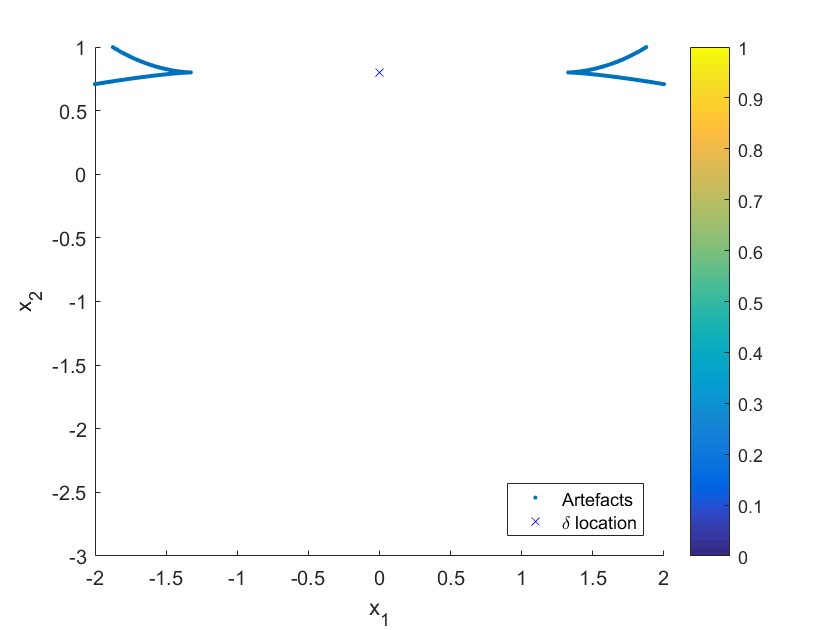}
\subcaption{$\Lambda_{ij}$ artefacts on $[-2,2]\times [-3,1]$.}\label{FC2:F}
\end{subfigure}
\caption{$\Tc^*\Tc\delta$ (the $\delta$ function is centered at
$(0,0.9)$) images with the predicted artefacts due to the limited data
backprojection (on $S_{12}\cup S_{21}$) and those induced by
$\Lambda_{12}$ and $\Lambda_{21}$.} \label{FC2}
\end{figure}

To simulate a $\delta$ function discretely we assign a value of 1 to
nine neighbouring pixels in a 200--200 grid (which will represent
$[-2,2]\times [-3,1]$) and set all other pixel values to zero. Letting our
discrete delta function be denoted by $x_{\delta}$, we approximate
$\mathcal{T}^*\mathcal{T}\delta\approx A^TAx_{\delta}$, where $A$ is
the discrete form of $\mathcal{T}$. For comparison we show images of
$$(\mathcal{T}_1^*\mathcal{T}_2+\mathcal{T}_2^*\mathcal{T}_1)\delta\approx (A_1^TA_2+A_2^TA_1)x_{\delta},$$
a characteristic function on the set $S_{12}\cup S_{21}$, and the
artefacts induced by $\Lambda_{12}$ and $\Lambda_{21}$. Here $A_j$ is the discrete
form of $\mathcal{T}_j$, for $j=1, 2$.  See figure \ref{FC1}. For example, in figure \ref{FC1:B} we see a
butterfly wing type artefact in $(A_1^TA_1+A_2^TA_2)x_{\delta}$. This is
due to the limited $r$ and $x_0$ data inherent to our acquisition geometry
(there are unresolved wavefront directions).
In the $(A_1^TA_2+A_1^TA_2)x_{\delta}$ image of figure \ref{FC1:C} we see artefacts
appearing on the set $S_{12}\cup S_{21}$ as shown in figure \ref{FC1:D}. This is as predicted by our theory.
The artefacts induced by the $\Lij$ in this case lie outside
the scanning region ($[-2,2]\times [-3,1]$), and hence they are not observed in
the reconstruction. See figures \ref{FC1:E} and \ref{FC1:F}. In figure \ref{FC2} the artefact curves intersect $[-2,2]\times [-3,1]$ in the top left and
right-hand corners respectively. See figures \ref{FC2:E} and \ref{FC2:F}. In this case the artefacts are
observed faintly in the reconstructions (their \emph{magnitude} is small compared to the delta function), and it is unclear whether they align with our predictions.
\begin{figure}[!h]
\begin{subfigure}{0.32\textwidth}
\includegraphics[width=0.9\linewidth, height=4cm]{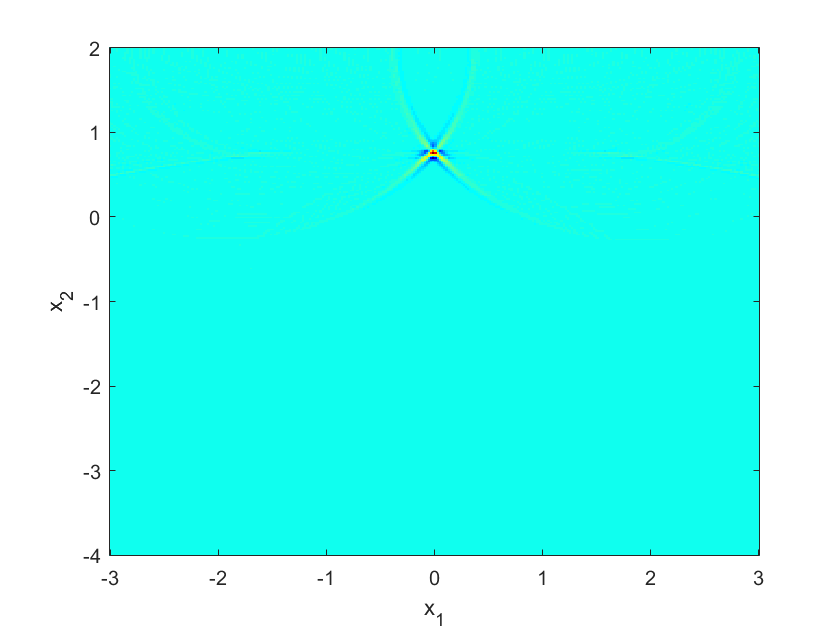}
\subcaption{$\Tc^*\Phi\Tc\delta$ (location 1).}\label{FC3:A}
\end{subfigure}
\begin{subfigure}{0.32\textwidth}
\includegraphics[width=0.9\linewidth, height=4cm]{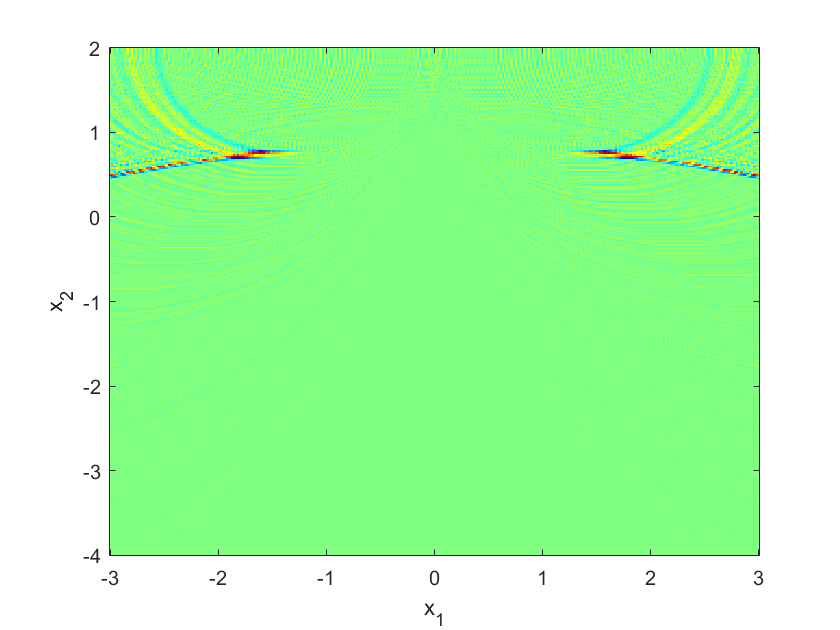}
\subcaption{$(\Tc_1^*\Phi\Tc_2+\Tc_2^*\Phi\Tc_1)\delta$.}\label{FC3:B}
\end{subfigure}
\begin{subfigure}{0.32\textwidth}
\includegraphics[width=0.9\linewidth, height=4cm]{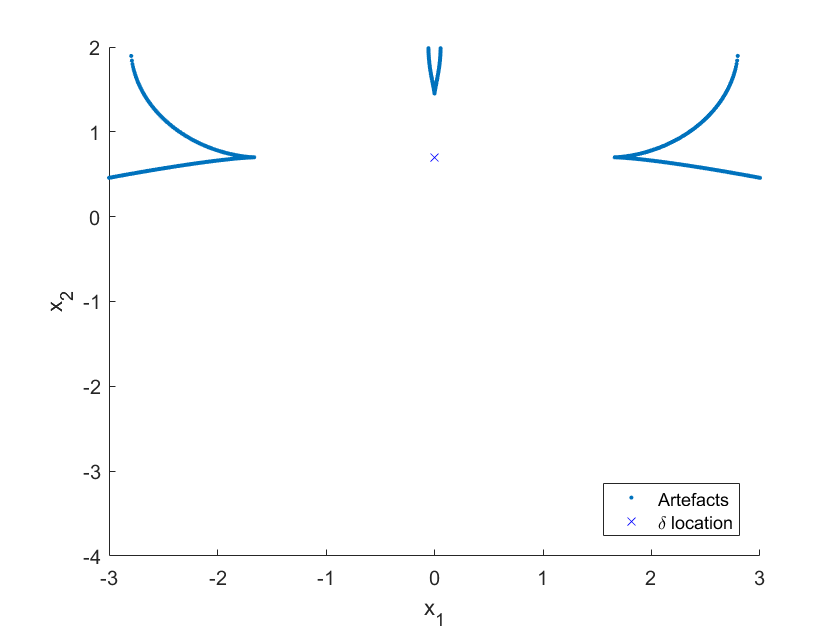}
\subcaption{$\Lij$ artefacts.}\label{FC3:C}
\end{subfigure}
\begin{subfigure}{0.32\textwidth}
\includegraphics[width=0.9\linewidth, height=4cm]{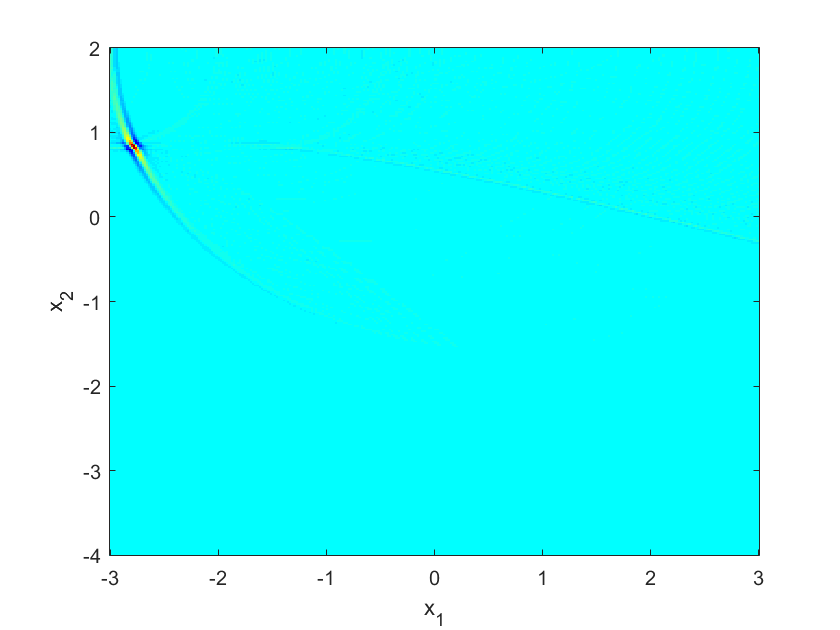}
\subcaption{$\Tc^*\Phi\Tc\delta$ (location 2).}\label{FC3:D}
\end{subfigure}
\begin{subfigure}{0.32\textwidth}
\includegraphics[width=0.9\linewidth, height=4cm]{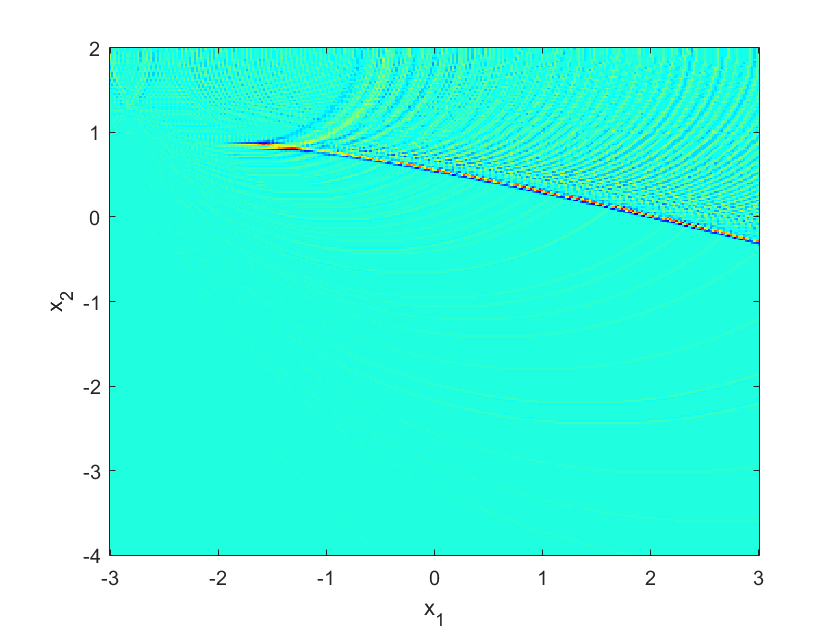}
\subcaption{$(\Tc_1^*\Phi\Tc_2+\Tc_2^*\Phi\Tc_1)\delta$.}\label{FC3:E}
\end{subfigure}
\begin{subfigure}{0.32\textwidth}
\includegraphics[width=0.9\linewidth, height=4cm]{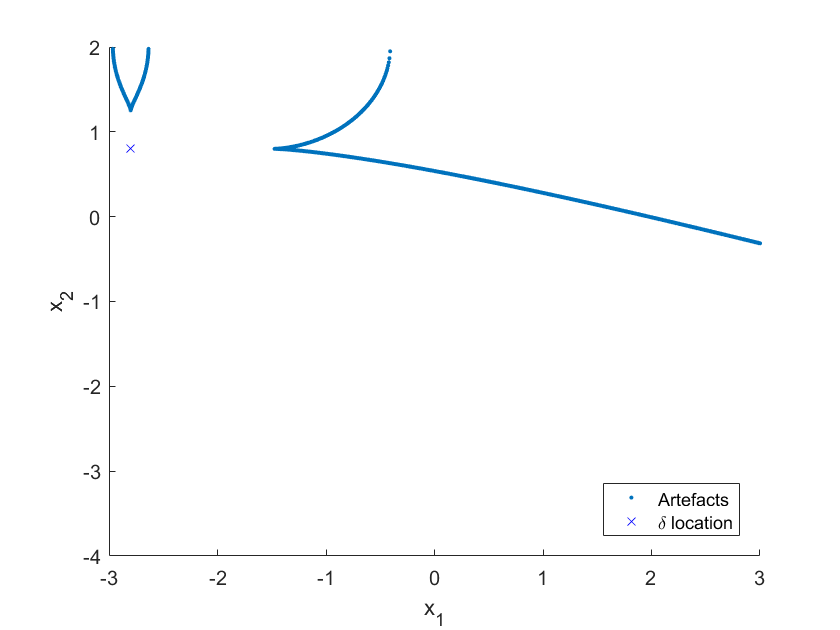}
\subcaption{$\Lambda_{ij}$ artefacts.}\label{FC3:F}
\end{subfigure}
\caption{$\Tc^*\Phi\Tc\delta$ and $(\Tc_1^*\Phi\Tc_2+\Tc_2^*\Phi\Tc_1)\delta$ images with the predicted artefacts induced by $\Lambda_{12}$ and $\Lambda_{21}$. We give examples for two $\delta$ function locations. Location 1 is (0,0.85) and location 2 is (-2.8,0.9).} \label{FC3}
\end{figure}

To show the artefacts induced by the $\Lij$ more clearly, we repeat
the analysis above using filtered backprojection, and a second derivative filter $\Phi=\frac{\mathrm{d}^2}{\mathrm{d}r^2}$. That
is we show images of
$\Tc^*\Phi\Tc\delta$. Note that $\Phi$ is applied in the variable $r$ (the torus
radius). The application of derivative filters is a common
idea in lambda tomography \cite{FFRS,FRS}, and is known to highlight
the image contours (singularities or edges) in the
reconstruction \cite[Theorem 3.5]{rigaud20183d}. See figure \ref{FC3}.
As the artefacts induced by $\Lij$ appeared to be largely outside the
scanning region ($[-2,2]\times [-3,1]$) in our previous simulations, we have
increased the scanning region size to $[-3,3]\times [-4,2]$, to show more the
effects of the $\Lij$ in the observed reconstruction. Here $\Phi$ suppresses the artefacts due to limited data, and the $\Lij$
artefacts appear as additional contours in the
reconstruction. The observed artefacts appear most clearly in figures \ref{FC3:B} and \ref{FC3:E}, and align exactly with our predictions in figures \ref{FC3:C} and \ref{FC3:F}.

% \tc{I'm guessing that the $\Lij$ artifacts are not suppressed but that
% the other artifacts are enough more singular  in Sobolev norm that
% they are more noticeable.  \newline Also, did you cut off the negative
% values in  the reconstruction?  If so, some of the artifacts could be
% there. }

\begin{remark}\label{rem:scanner size} With precise knowledge of the locations of the artefacts induced by
the $\Lij$ we can assist in the design of the proposed parallel line
scanner. That is we can choose $a$, $r_M$ and the scanning tunnel size
to minimize the presence of the nonlocal artefacts in the
reconstruction (i.e., those from $\Lij(f)$). Such advice would be of
benefit to our industrial partners in airport screening to remove the
concern for nonlocal artefacts in the image reconstruction of baggage.
Indeed the machine dimensions we have chosen seem to be suitable as
the artefacts appear largely outside the reconstruction space (see
figures \ref{FC1} and \ref{FC2}).
\end{remark}

\section{The transmission artefacts} \label{microsec2} The detector row
$D_C$ collects Compton (back) scattered data, which determines
$\mathcal{T}f(s,x_0)$ for a range of $s$ and $x_0$, where $f=n_e$ is the
electron charge density. The forward detector array $D_A$ collects
transmission (standard X-ray CT) data, which determine a set of straight
line integrals over the attenuation coefficient $f=\mu_E$, for some photon
energy $E$.  The data is limited to the set of lines which intersect $S$ (the source array) and $D_A$. This limited data can cause artefacts in the X-ray
reconstruction, and we will analyze these artefacts using the theory in
\cite{borg2018analyzing}.  Let
$L_{s,\theta}=\{\vx\in\mathbb{R}^2 : \vx\cdot \Theta=s\}$ be the line
parameterized by a rotation $\theta\in [0,\pi]$ and a directed distance
from the origin $s\in \mathbb{R}$. Here $\Theta=(\cos\theta,\sin\theta)$
and $L_{s,\theta}$ is the line containing $s\Theta$ and perpendicular to
$\Theta$.

In the scanning geometry of this article, the set $S$ of X-ray
transmitters is the segment between $(-4,3)$ and $(4,3)$ and the set
of X-ray detectors, $D_A$, is the segment between $(-4,-5)$ and
$(4,-5)$ as in figure \ref{fig1}.  For this reason, the cutoff in the
sinogram space is described by the set
\begin{equation}\label{def:H}
H=\{(s,\theta)\in \mathbb{R}\times [-\pi/2,\pi/2] : L_{s,\theta}\cap
S\neq \emptyset\ \ \text{and}\ \ L_{s,\theta}\cap D_A\neq \emptyset\}.
\end{equation} The characteristic function of $H$ appears as a 
skewed diamond shape in sinogram space.

% The characteristic function on $H$ is displayed in figure \ref{Fcut},
% and $H$ appears as a diamond shape in the sinogram space.
% \begin{figure}[!h]
% \centering
% \begin{subfigure}{0.4\textwidth}
% \includegraphics[width=1.0\linewidth, height=5.3cm]{smoothcut.png}
% \end{subfigure}
% \caption{The characteristic function of $H$. Note that, in this
% picture, we are assuming the origin in object space is at $\vO=(0,-1)$
% (the center of the square with opposite sides $S$ and $D_A$) when
% parameterizing lines.} \label{Fcut}
% \end{figure}

To illustrate the added artifacts inherent in this incomplete data
problem, we simulate reconstructions of delta functions with
transmission CT data on $H$.  That is, we apply the
backprojection operator $R_L^*R_L$ to $\delta$ functions, where $R_Lf$ denotes $Rf$ for $(s,\theta)\in H$.  See figure \ref{FA1}.  \begin{figure}[!h]
%\hspace*{-1cm}
\begin{subfigure}{0.32\textwidth}
\includegraphics[width=0.9\linewidth, height=4cm]{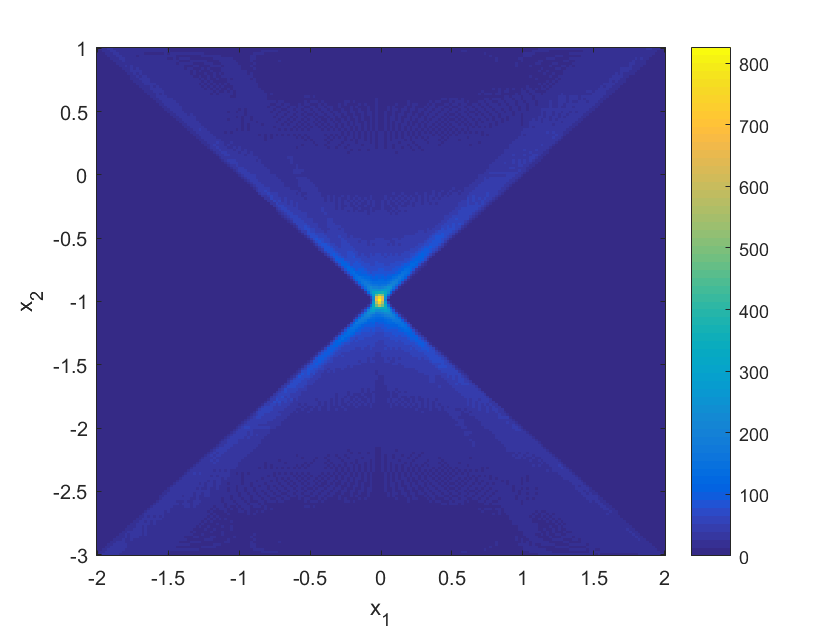}
\end{subfigure}
\begin{subfigure}{0.32\textwidth}
\includegraphics[width=0.9\linewidth, height=4cm]{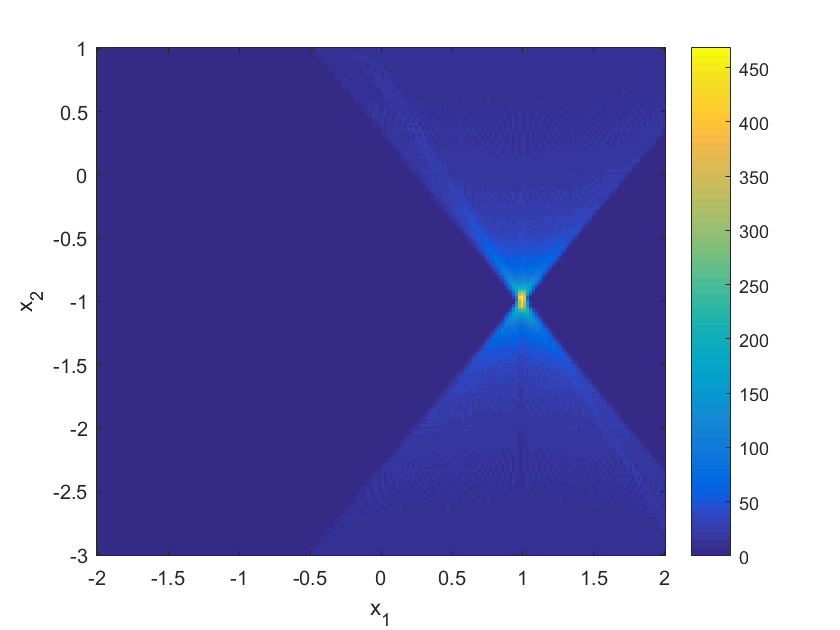} 
\end{subfigure}
\begin{subfigure}{0.32\textwidth}
\includegraphics[width=0.9\linewidth, height=4cm]{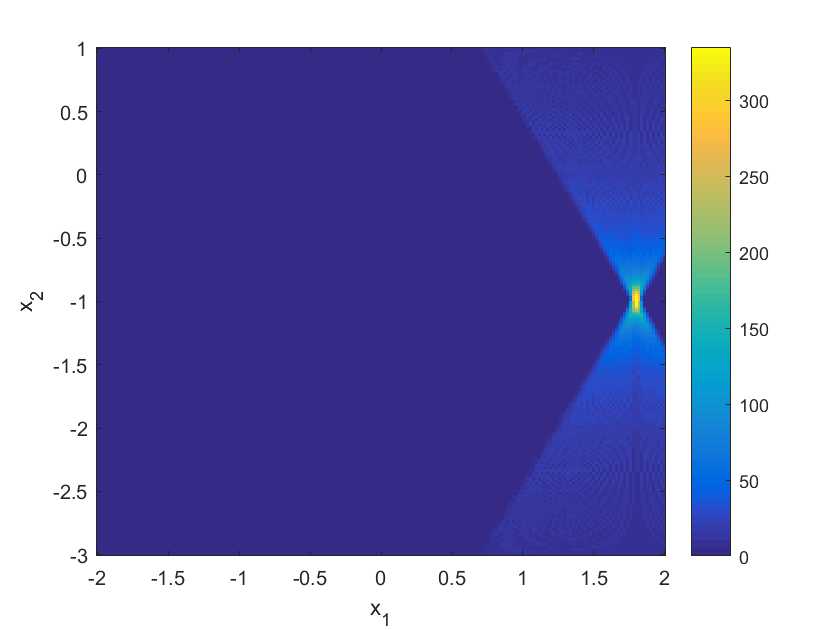}
\end{subfigure}
\caption{$R_L^TR_L\delta_{(t,-1)}$ images at varying $\delta$
  function translations along the line $x_2=-1$.} \label{FA1} \end{figure}
  By the theory
in \cite{borg2018analyzing}, artefacts caused by the incomplete data
occur on lines $L_{s,\theta}$ for $(s,\theta)$
in the boundary of $H$.  Each delta function in Figure \ref{FA1} is at
a point $(t,-1)$ for some $t$ with $0<t<2$, so the lines that meet the
support of the delta function, $(t,-1)$ that are in the boundary of
$H$ must also contain either $(4,3)$ or $(4,-5)$.  This is true
because $S$ and $D_A$ are mirror images about the line $y=-1$ and
 $t\in (0,2)$.  

Furthermore, by symmetry (the $\delta$ functions are on the
center line of the scanning region), these artifact lines will be
reflections of each other in the vertical line $x_1=t$.  This is
illustrated in our reconstructions in Figure \ref{FA1}.  The opening
angle of the cone in the delta reconstructions decreases (fewer
wavefront set directions are stably resolved) as we translate $\delta$
to the right on the line $x_1=-1$.

 \begin{example}\label{ex:joint} We now use these ideas to analyze the
visible wavefront directions for the joint problem.  Let
$\Sc=[-4,4]\times [-5,3]$ be the square between $S$ and $D_A$ and let
$\vO=(0,-1)$, the center of $\Sc$.  We consider wavefronts at points
$(x_1,x_2)\in [-2,2]\times [-3,1]$ which is a square centered at $\vO$
and the region in which our simulated reconstructions are done.

  By \eqref{def:H}, lines in the data set must intersect both $S$ and
$D_A$, so lines through $\vO$ in the data set are all lines through
$\vO$ that are more vertical than the diagonals of $\Sc$.  Because
visible wavefront directions are normal to lines in the X-ray CT data
set \cite{Q1993sing}, the wavefront directions which are resolved lie
in the horizontal open cone between normals to these diagonals.
Therefore, they are in the cone
 $$C_R=\{\pm(\cos\alpha,\sin\alpha) : -\pi/4<\alpha<\pi/4\},$$ which
 is shown in figure \ref{fig1}.
 
 An analysis of the singularities that are visible by the Compton data
was done in section \ref{sect:artifactsTc}.  For the point $\vO$, the
angle defined by \eqref{def:betam} gives $\beta_m=1.23$ and the cone
of visible directions given by \eqref{def:C_T} is the vertical cone
with angles from the vertical between $-\beta_m$ and $\beta_m$ since
the parameter $r_M=9$.  A calculation shows that $2(\pi/4
+\beta_m)>\pi$ and this implies $ C_R\cup C_T=S^1$ and we have a full
resolution of the image singularities at $\vO$. 
\begin{figure}[h!]
\includegraphics[width = 0.4\linewidth]{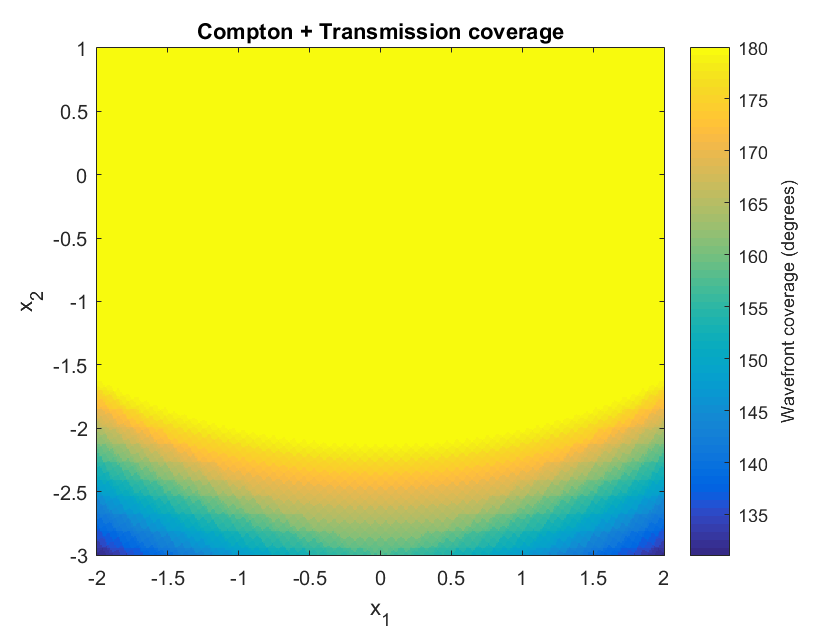} \caption{Picture of
the range of wavefront directions that are visible at points in $[-2,2]\times [-3,1]$ from the joint data. Angles are measured from $0^\circ$ = no coverage to
$360^\circ$ = full coverage. We let $\Gamma$ denote the solid light-colored (yellow) region
(roughly the top 3/4 of the figure) in
which all wavefront directions are recovered.} \label{fig:directions}\end{figure}

However, for points near the bottom $x_2=-3$, there are invisible
singularities that are not visible from either the Compton or X-ray
data.  For example, the vertical direction $(0,1)$ is not normal to
any circle in the Compton data set at any point $(0,x_2)$ for $x_2\in
[-2.5,-3]$.  Figure \ref{fig:directions} shows the points for which
all wavefront directions are visible at those points (yellow
color--roughly for points $(x_1,x_2)$ for $x_2>-2$) and near the
bottom of the reconstruction region, there are more missing
directions.  
\end{example}

%Also, please delete ``on the market''.}}

% \tc{Note there are no artifacts at top and bottom corners since
%   $f=\delta$ and the two lines that do appear are the ones we see.
%   Talk about how singularity at the origin gives horizontal corner
%   and for vertical need support away from these corners but in a .
%   Note that the sources and detectors are outside the reconstruction
%   region.  artifacts come because of lines at the boundary for data
%   set and these are ones meeting the $\pm a$ or the ends of the
%   detector region.  See p. 13 for the description. }
\section{A joint reconstruction approach and results}
\label{results}
In this section we detail our joint reconstruction scheme and lambda tomography regularization technique, and show the effectiveness of our methods in combatting the artefacts observed and predicted by our microlocal theory. We first explain the physical relationship between $\mu_E$ and $n_e$, which will be needed later in the formulation of our regularized inverse problem.
\subsection{Relating $\mu_E$ and $n_e$}
The attenuation coefficient and electron density satisfy the formula \cite[page 36]{wadeson2011modelling}
\begin{equation}
\label{mu-n}
\mu_E(Z)=n_e(Z)\sigma_E(Z),
\end{equation}
where $\sigma_E$ denotes the electron cross section, at energy $E$.
Here $Z$ denotes the effective atomic number. In the proposed
application in airport baggage screening (among many other
applications such as medical CT) we are typically interested in the
materials with low effective $Z$. Hence we consider the materials with
$Z<20$ in this paper.  For large enough $E$ and $Z<20$, $\sigma_E(Z)$
is approximately constant as a function of $Z$.  Equivalently $\mu_E$
and $n_E$ are approximately proportional for low $Z$ and high $E$ by
equation \ref{mu-n}. See figure \ref{fig2}. We see a strong
correlation between $n_e$ and $\mu_E$ when $E=100$keV and $Z<20$, and
even more so when $E$ is increased to $E=1$MeV. The sample of
materials considered consists of 153 compounds (e.g. wax, soap, salt,
sugar, the elements) taken from the NIST database
\cite{hubbell1995tables}. In this case $\sigma_E\approx\nu$ for some
$\nu\in\mathbb{R}$ is approximately constant and we have
$\mu_E\approx\nu n_e$. 
\begin{figure}[!h]
\centering
\begin{subfigure}{0.32\textwidth}
\includegraphics[width=0.9\linewidth, height=4cm]{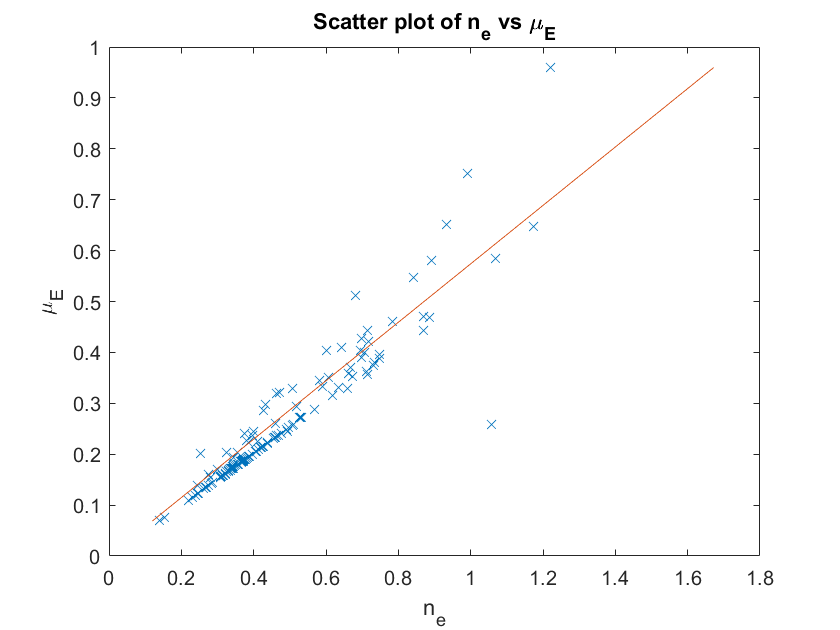} 
\end{subfigure}\hspace{5mm}
\begin{subfigure}{0.32\textwidth}
\includegraphics[width=0.9\linewidth, height=4cm]{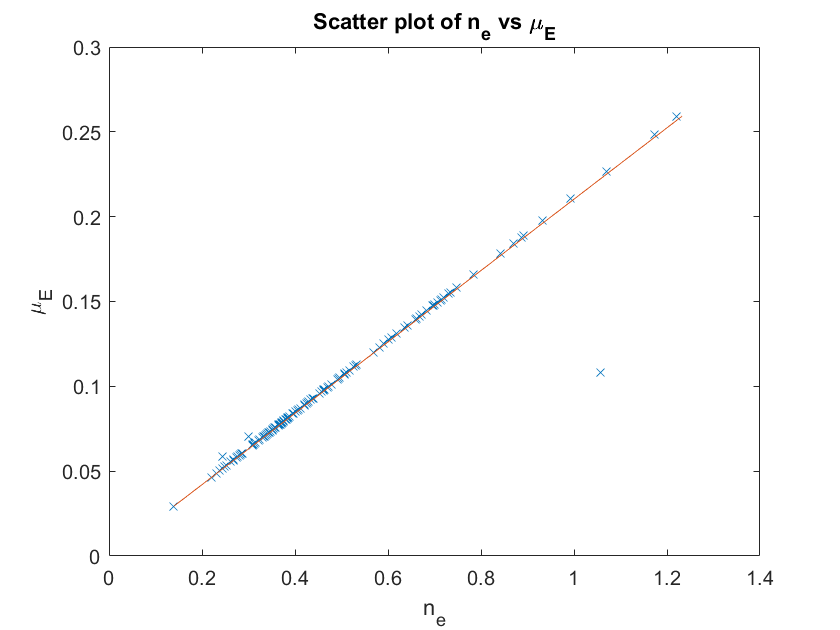} 
\end{subfigure}
\caption{Scatter plot of $n_e$ vs $\mu_E$ for $E=100$keV (left) and $E=1$MeV (right), for 153 compounds with effective $Z<20$ taken from the NIST \cite{hubbell1969photon} database. The correlation is R=0.93 (left) and R=0.98 (right).}
\label{fig2}
\end{figure}
For a given energy $E$, $\nu$ is the slope of
the straight line fit as in figure \ref{fig2}. Throughout the rest of
this paper, we set $\nu$ as the slope of the straight line in the left
hand of figure \ref{fig2} (i.e. $\nu\approx 0.57$), and present
reconstructions of $\mu_E$ for $E=100$keV. 

\subsection{Lambda regularization; the idea}
In sections \ref{microsec1} and \ref{microsec2} we discovered that the
$R_L\mu _E$ and $\Tc n_e$ data provide complementary information regarding
the detection and resolution of edges in an image. More specifically the
line integral data resolved singularities in an open cone $C_R$ with
central axis $x_1$ and the toric section integral data resolved
singularities in a cone $C_T$ with central axis $x_2$. So the overlapping
cones $C_R\cup C_T$ give a greater coverage of $S^1$ than when considered
separately. In figures \ref{FC1}, \ref{FC2} and \ref{FA1}, this theory was
later verified through reconstructions of a delta functions by (un)filtered
backprojection. 

For a further example, let us consider a more complicated phantom than a
delta function, one which is akin to densities considered later for testing
our joint reconstruction and lambda regularization method. In figure
\ref{Fcont} we have presented reconstructions of an image phantom $f$ (with
no noise) from $R_Lf$ (transmission data--middle figure) using FBP, and
from $\Tc f$ (Compton data--right figure) by an application of
$\Tc^*\frac{\mathrm{d}^2}{\mathrm{d}r^2}$ (a contour reconstruction). In
the reconstruction from Compton data, we see that the image singularities
are well resolved in the vertical direction ($x_2$), and conversely in the
horizontal direction ($x_1$) in the reconstruction from transmission data.
In the middle picture (reconstruction from $R_L$), the visible
singularities of the object are tangent to lines in the data set (normal
wavefront set) and the artifacts are along lines at the end of the data set
that are tangent to boundaries of the objects.  In the right-hand
reconstruction from Compton data, the visible boundaries are tangent to
circles in the data set and the streaks are along circles at the end of the
data set.  Note that the visible boundaries in each picture complement each
other and together, image the full objects.  This is all as predicted by
the theory of sections \ref{microsec1} and \ref{microsec2} (and is
consistent with the theory in \cite{borg2018analyzing} and \cite{FrQu2015})
and highlights the complementary nature of the Compton and transmission
data in their ability to detect and resolve singularities.

%For example, with the chosen machine dimensions of $2a=8$ and $r_m-1=6$ (as described in section \ref{microsec1}), when $\vx=O$ the corresponding values for $\beta_m$ (as in section \ref{microsec1}) and $\alpha_m$ (see figure \ref{fig1}) are $\beta_m=1.23$ and $\alpha_m=\frac{\pi}{4}$. Hence with Compton and transmission data we have a 

%In terms of wavefront sets we can write
%$$\text{WF}(f)\subseteq \text{WF}(R^*_L\Phi_1R_L f)\cup \text{WF}(\Tc^*\Phi_2\Tc f),$$
%where $\Phi_1$ denotes the ramp filter and $\Phi_2=\frac{\mathrm{d}^2}{\mathrm{d}r^2}$. Note that $\text{WF}(f)$ is only contained in (as opposed to equal to) the above union given the additional contours induced by the $\alpha_j$ in $\text{WF}(\Tc^*\Phi_2\Tc f)$. Equality occurs when the artefacts induced by the $\alpha_j$ lie outside $\text{supp}(f)$ (e.g. in the case of figure \ref{FC1}). 
%This is as predicted by the theory presented in sections \ref{microsec1} and \ref{microsec2}. 
Given the complementary edge resolution capabilities of $R_L\mu_E$ and $\Tc n_e$, and given the approximate linear relationship between $\mu_E$ and $n_e$, we can devise a joint linear least squares reconstruction scheme with the aim to recover the image singularities stably in all directions in the $n_e$ and $\mu_E$ images simultaneously. 
\begin{figure}[!h]
%\hspace*{-1cm}
\begin{subfigure}{0.32\textwidth}
\includegraphics[width=0.9\linewidth, height=4cm]{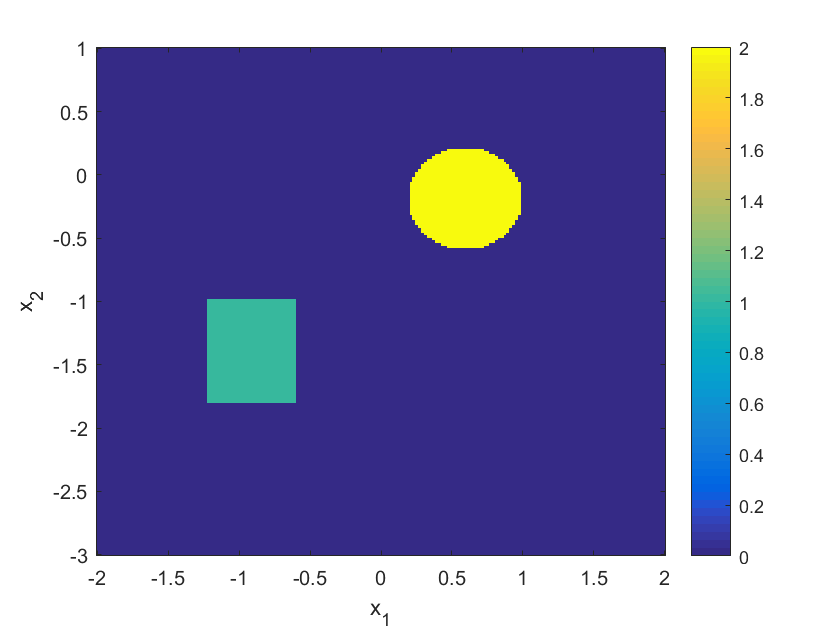}
\end{subfigure}
\begin{subfigure}{0.32\textwidth}
\includegraphics[width=0.9\linewidth, height=4cm]{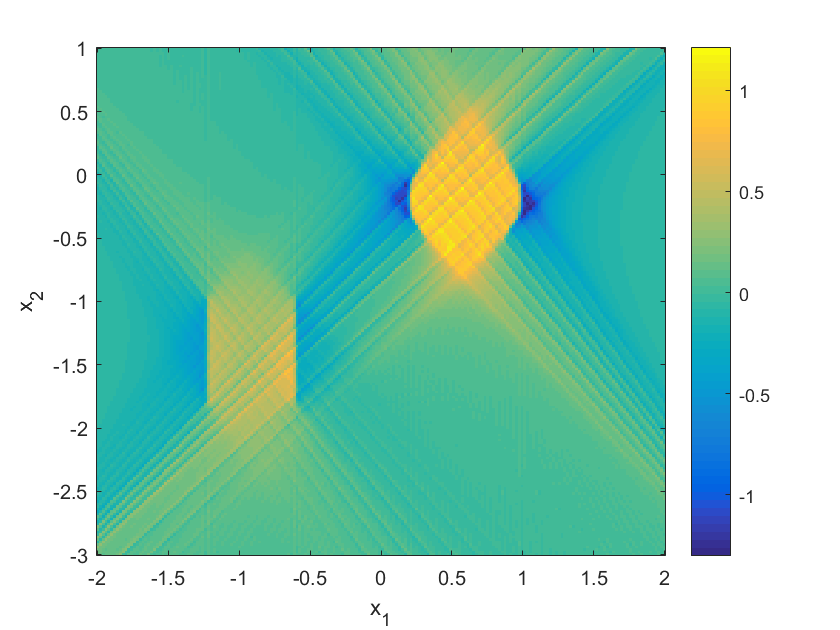} 
\end{subfigure}
\begin{subfigure}{0.32\textwidth}
\includegraphics[width=0.9\linewidth, height=4cm]{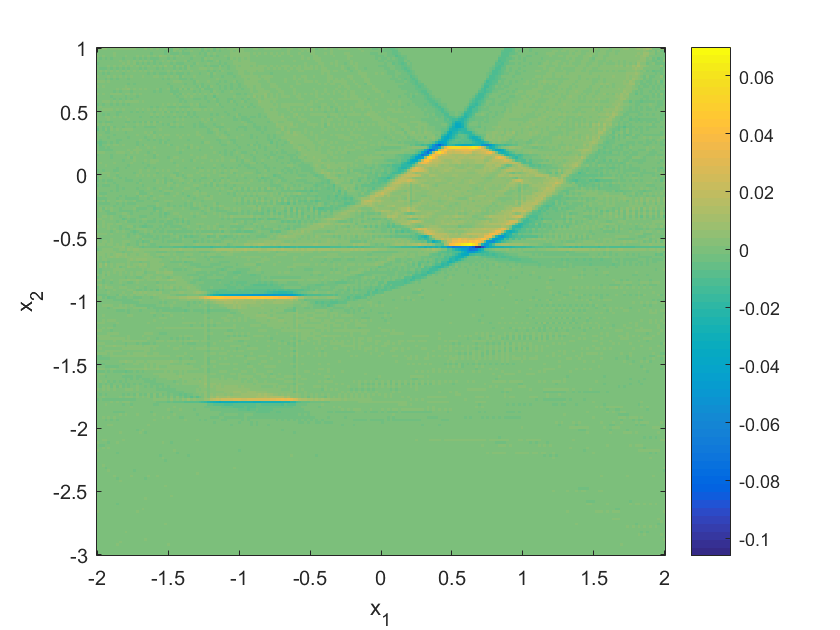}
\end{subfigure}
\caption{Image phantom $f$ (left), a reconstruction from $R_Lf$ using FBP (middle) and $\Tc^*\frac{\mathrm{d}^2}{\mathrm{d}r^2}\Tc f$ (right).}
\label{Fcont}
\end{figure}
To this end we employ ideas in lambda tomography and microlocal analysis. 

Let $f\in\mathcal{E}(\mathbb{R}^n)$ and let
$Rf(s,\theta)=Rf_{\theta}(s)=\int_{L_{s,\theta}}f\mathrm{d}l$ denote
the hyperplane Radon transform of $f$, where $L_{s,\theta}$ is as
defined in section \ref{microsec2}. The Radon projections
$Rf_{\theta}$ detect singularities in $f$ in the direction
$\Theta=(\cos\theta,\sin\theta)$ (i.e. the elements
$(\vx,\Theta)\in\text{WF}(f)$). Applying a derivative filter
$\frac{\mathrm{d}^m}{\mathrm{d}s^m}R_{\theta}$, for some $m\geq1$,
increases the strength of the singularities in the $\Theta$ direction
by order $m$ in Sobolev scale. 
\begin{figure}[!h]
\centering
\begin{subfigure}{0.35\textwidth}
\includegraphics[width=1\linewidth, height=4.8cm]{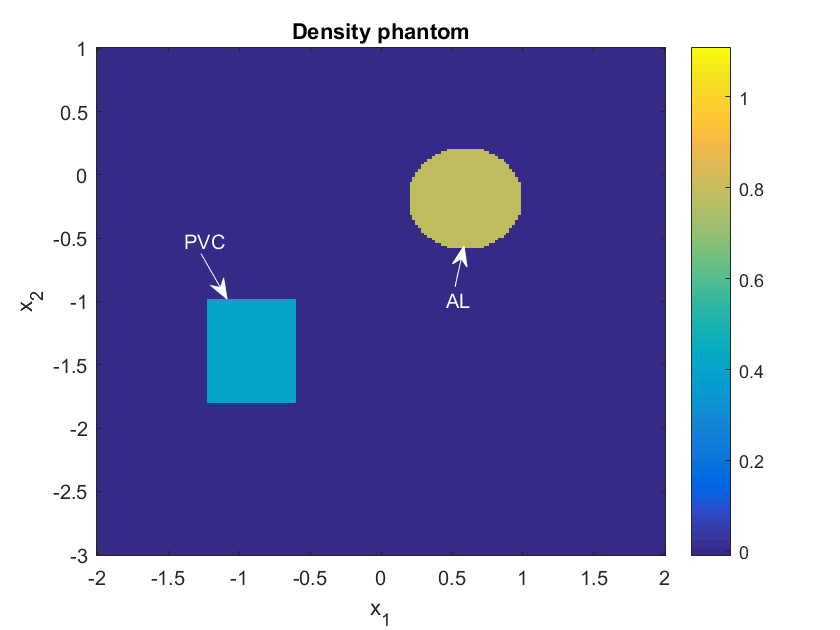} 
\end{subfigure}
\begin{subfigure}{0.35\textwidth}
\includegraphics[width=1\linewidth, height=4.8cm]{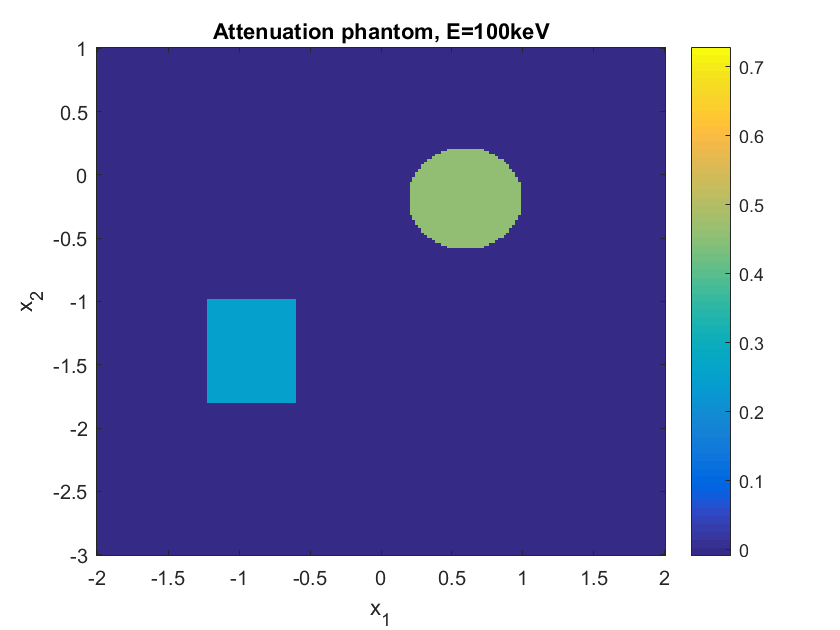}
\end{subfigure}
\begin{subfigure}{0.35\textwidth}
\includegraphics[width=1\linewidth, height=4.8cm]{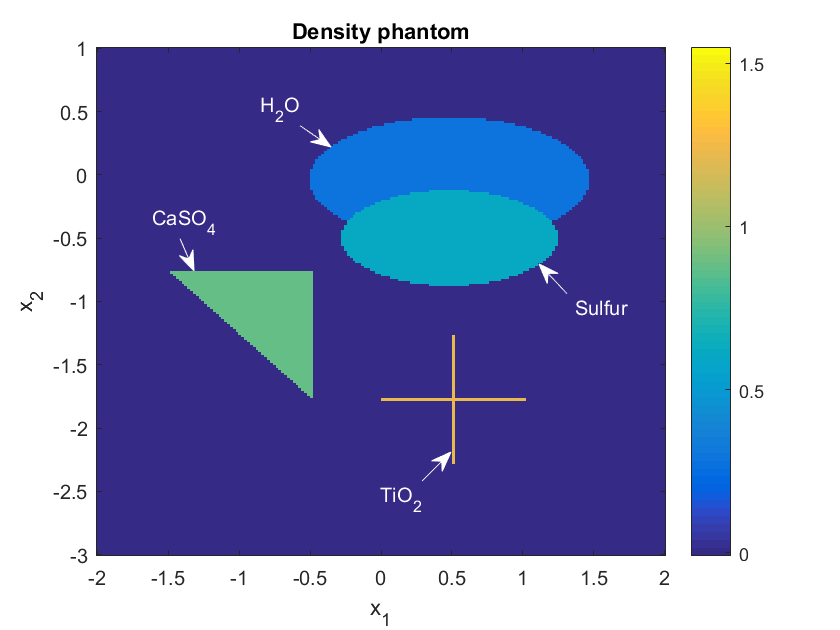} 
\end{subfigure}
\begin{subfigure}{0.35\textwidth}
\includegraphics[width=1\linewidth, height=4.8cm]{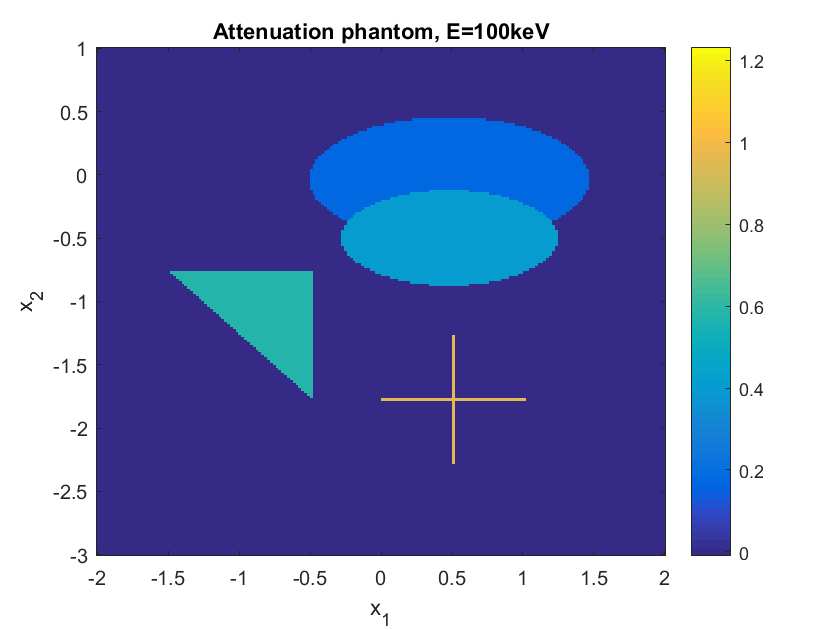}
\end{subfigure}
\caption{Top row -- Simple density (left) and attenuation (right) phantoms. Bottom row -- Complex density (left) and attenuation (right) phantoms. The associated materials are labelled in each case.}
\label{figPhan}
\end{figure}Given
$f,g\in\mathcal{E}(\mathbb{R}^n)$, we aim to enforce a similarity in
$\text{WF}(f)$ and $\text{WF}(g)$ through the addition of the penalty
term
$\|\frac{\mathrm{d}^m}{\mathrm{d}s^m}R(f-g)\|_{L^2(\mathbb{R}\times
S^{1})}$ to the least squares solution. Note that we are integrating
over all directions in $S^{1}$ to enforce a full directional
similarity in $\text{WF}(f)$ and $\text{WF}(g)$. Specifically in our
case $f=\mu_E$, $g=n_e$ and we aim to minimize the quadratic
functional
\begin{equation}
\label{JLAM}
\argmin_{\mu_E,n_e}\left\|\begin{pmatrix}
wR_L & 0\\
0 & \mathcal{T}\\
%\alpha_1[I & -\nu I] \\
\alpha[\frac{\mathrm{d}^m}{\mathrm{d}s^m}R & -\nu\frac{\mathrm{d}^m}{\mathrm{d}s^m}R]
\end{pmatrix}\begin{pmatrix}
\mu_E\\
n_e\\
\end{pmatrix}-\begin{pmatrix}
wb_1\\
b_2\\
%0\\
0
\end{pmatrix}\right\|^2_2,
\end{equation}
where $R_L$ denotes a discrete, limited data Radon operator, $R$ is
the discrete form of the full data Radon operator, $\mathcal{T}$ is
the discrete form of the toric section transform, $b_1$ is known
transmission data and $b_2$ is the Compton scattered data. Here
$\alpha$ is a regularization parameter which controls the level of similarity in the image wavefront sets. The lambda regularizers enforce the soft constraint that $\mu_E=\nu n_e$ (since $\frac{\mathrm{d}^m}{\mathrm{d}s^m}Rf=0\iff f=0$ for $f\in L^2_c(X)$), but with emphasis on the location, direction and magnitude of the image singularities in the comparison.
\begin{table}[!h]
\hspace{-1.05cm}
\begin{subtable}{.49\linewidth}\centering
{\begin{tabular}{| c | c | c | c | c | c | c | c |}
\hline
$\epsilon$ & TV  & JLAM   & JTV & LPLS  \\ \hline
$n_e$     & $.26$ & $.14$ & $.05$ & $.03$ \\ 
$\mu_E$  & $.40$ & $.15$ & $.05$ & $.03$  \\ \hline
\end{tabular}}
%\caption{$\epsilon$}
\end{subtable}
\begin{subtable}{.49\linewidth}\centering
{\begin{tabular}{| c | c | c | c | c | c | c | c |}
\hline
$F$-score & TV  & JLAM   & JTV & LPLS  \\ \hline
$\text{supp}(n_e)$     & $.81$ & $.99$ & $.99$ & $\sim 1$ \\ 
$\nabla n_e$  & $.76$ & $.87$ & $.82$ & $.87$  \\ 
  $\text{supp}(\mu_E)$   & $.78$ & $\sim 1$ & $\sim 1$ & $\sim 1$ \\ 
$\nabla\mu_E$  & $.49$ & $.87$ & $.82$ & $.86$  \\\hline
\end{tabular}}
%\caption{$F$-score}
\end{subtable}
\caption{Simple phantom $\epsilon$ and $F$-score comparison using TV, JLAM, JTV and LPLS.}
\label{T1}
\end{table}
\begin{table}[!h]
\hspace{-1.05cm}
\begin{subtable}{.49\linewidth}\centering
{\begin{tabular}{| c | c | c | c | c | c | c | c |}
\hline
$\epsilon_{\pm}$ & JLAM   & JTV & LPLS  \\ \hline
$n_e$     & $.15\pm.01$ & $.05\pm.005$ & $.06\pm.002$ \\ 
$\mu_E$  & $.16\pm.02$ & $.05\pm.005$ & $.06\pm.03$ \\ \hline
\end{tabular}}
%\caption{$\epsilon$}
\end{subtable}
\begin{subtable}{.49\linewidth}\centering%49
{\begin{tabular}{| c | c | c | c | c | c | c | c |}
\hline
$F_{\pm}$ & JLAM   & JTV & LPLS  \\ \hline
$\text{supp}(n_e)$     & $.99\pm.01$ & $.99\pm.004$ & $\sim 1\pm.005$ \\ 
$\nabla n_e$  & $.86\pm.01$ & $.84\pm.02$ & $.86\pm.04$ \\
$\text{supp}(\mu_E)$    & $.98\pm.04$ & $.99\pm.005$ & $\sim1\pm.005$ \\ 
$\nabla\mu_E$  & $.86\pm.01$ & $.85\pm.02$ & $.87\pm.04$ \\ \hline
\end{tabular}}
%\caption{$F$-score}
\end{subtable}
\caption{Randomized simple phantom $\epsilon_{\pm}$ and $F_{\pm}$ comparison over 100 runs using JLAM, JTV and LPLS.}
\label{T1new}
\end{table}
Further we expect the lambda regularizers to have a smoothing effect given the nature of $\frac{\mathrm{d}^m}{\mathrm{d}s^m}R$ as a differential operator (i.e. the inverse is a smoothing operation). Hence we expect $\alpha$ to also act as a smoothing parameter. The weighting $w=\|\mathcal{T}\|_2/\|R_L\|_2$ is included so as to give equal weighting to the transmission and scattering datasets in the inversion. We denote the joint reconstruction method using lambda tomography regularizers as ``JLAM". A common choice for $m$ in lambda tomography applications is $m=2$ \cite{denisyuk1994inversion,rigaud20183d} (hence the name ``lambda regularizers").  
\begin{figure*}
\setlength{\tabcolsep}{5pt}
\begin{tabular}{ c|cc }  %x}{\textwidth}{@{}c*{2}{C}@{}}
  &  Density $n_e$ & Attenuation $\mu_E$, $E=100$keV \\ \hline \\[-0.4cm] %& error $n_e$ & error $\mu_E$ \\ 
 \rotatebox{90}{\hspace{2.05cm} TV}  & 
   \includegraphics[ width=0.4\linewidth, height=0.4\linewidth, keepaspectratio]{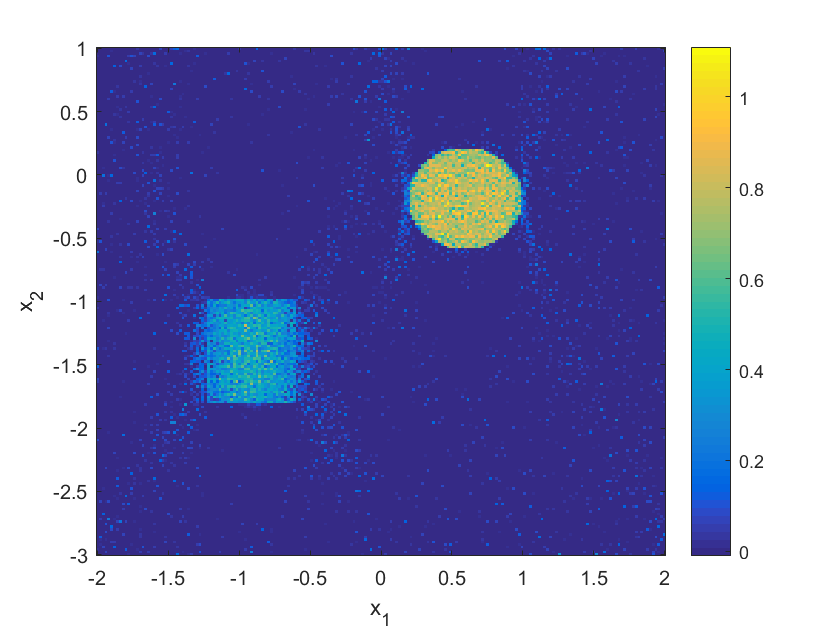} &
   \includegraphics[ width=0.4\linewidth, height=0.4\linewidth, keepaspectratio]{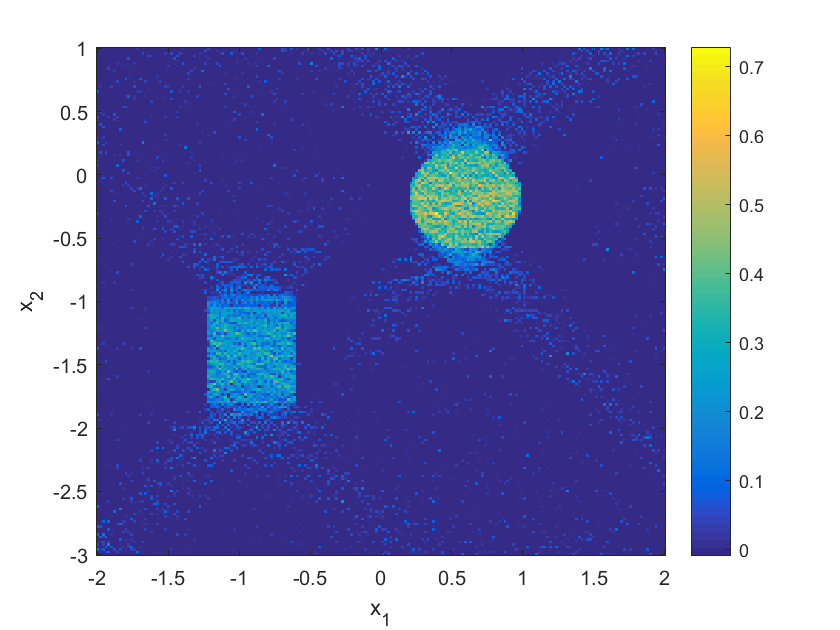} \\ %&
 \rotatebox{90}{\hspace{1.8cm} JLAM} &
  \includegraphics[ width=0.4\linewidth, height=0.4\linewidth, keepaspectratio]{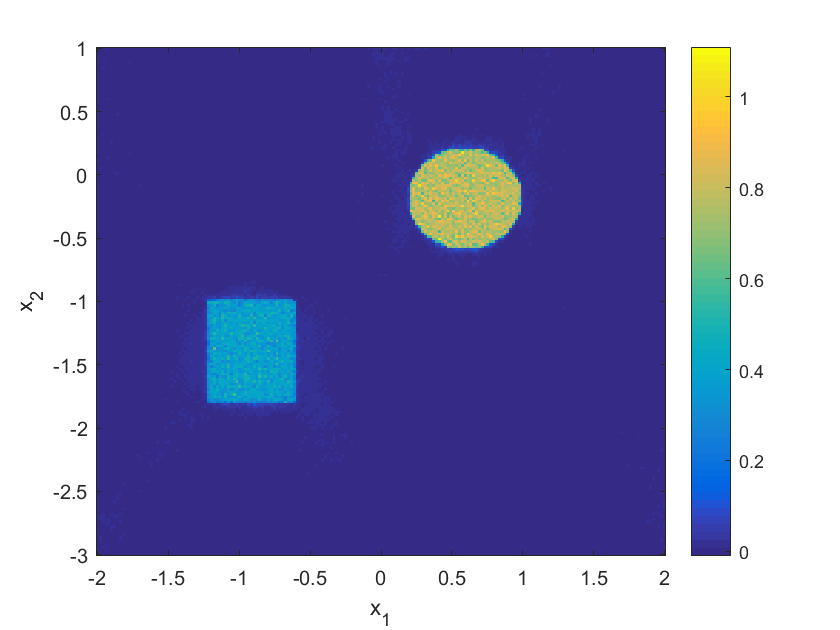} &
   \includegraphics[ width=0.4\linewidth, height=0.4\linewidth, keepaspectratio]{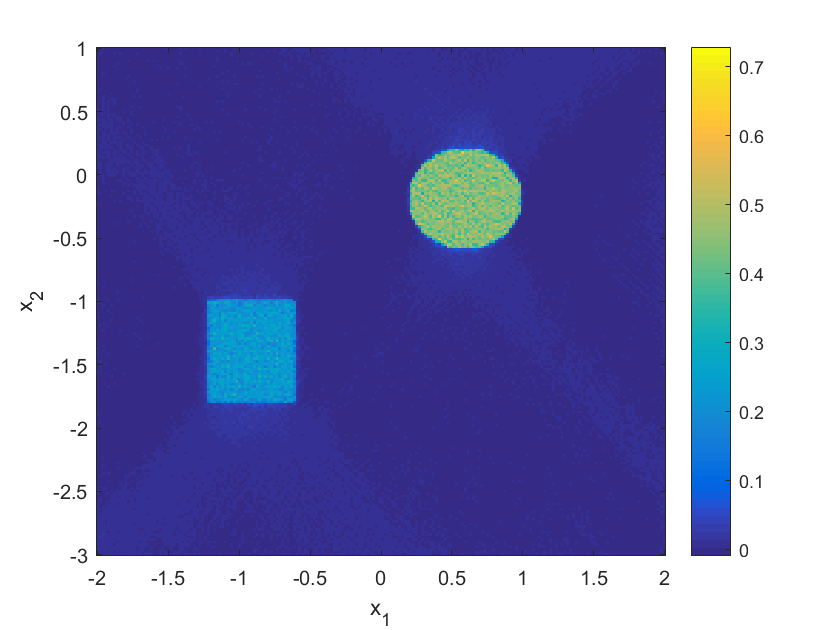} \\ %&
 \rotatebox{90}{\hspace{2cm} JTV} &
   \includegraphics[ width=0.4\linewidth, height=0.4\linewidth, keepaspectratio]{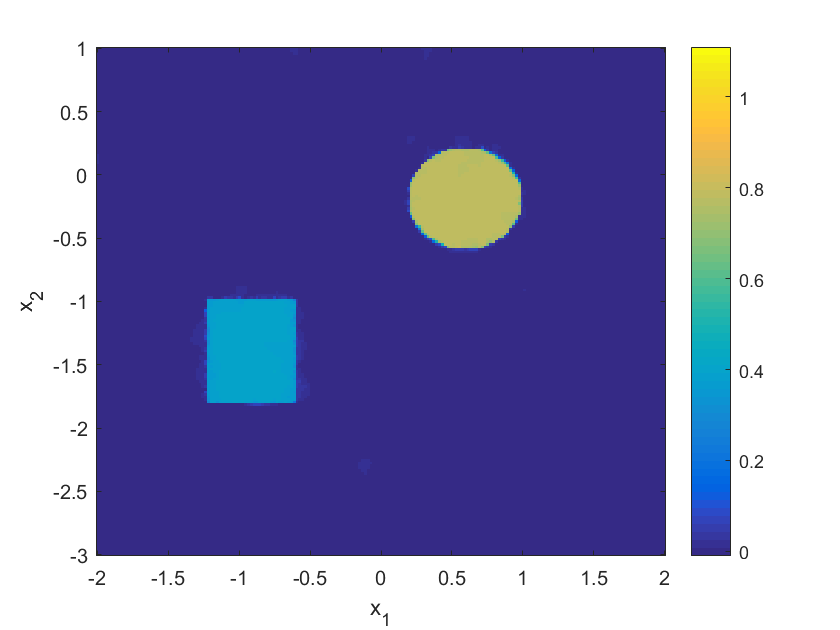} &
   \includegraphics[ width=0.4\linewidth, height=0.4\linewidth, keepaspectratio]{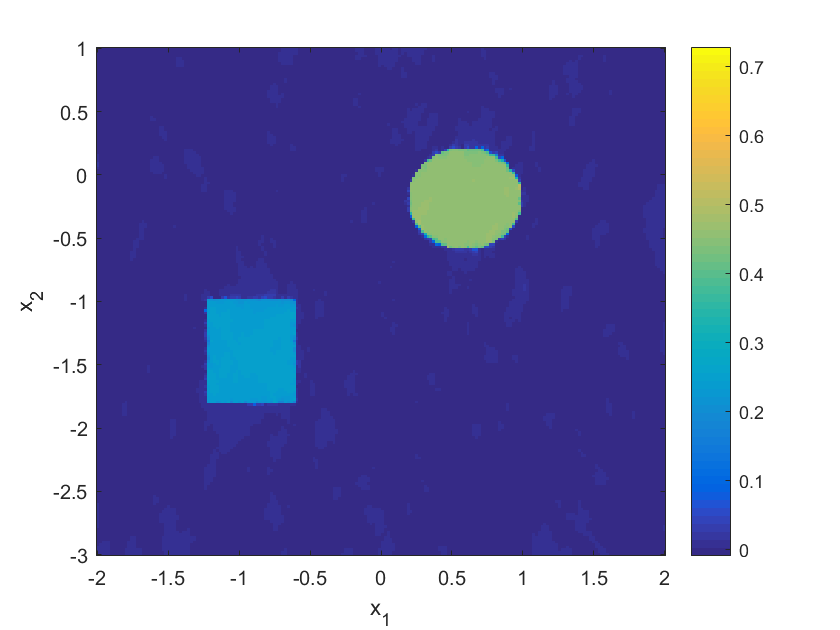} \\ %&
 \rotatebox{90}{\hspace{1.8cm} LPLS} &
   \includegraphics[ width=0.4\linewidth, height=0.4\linewidth, keepaspectratio]{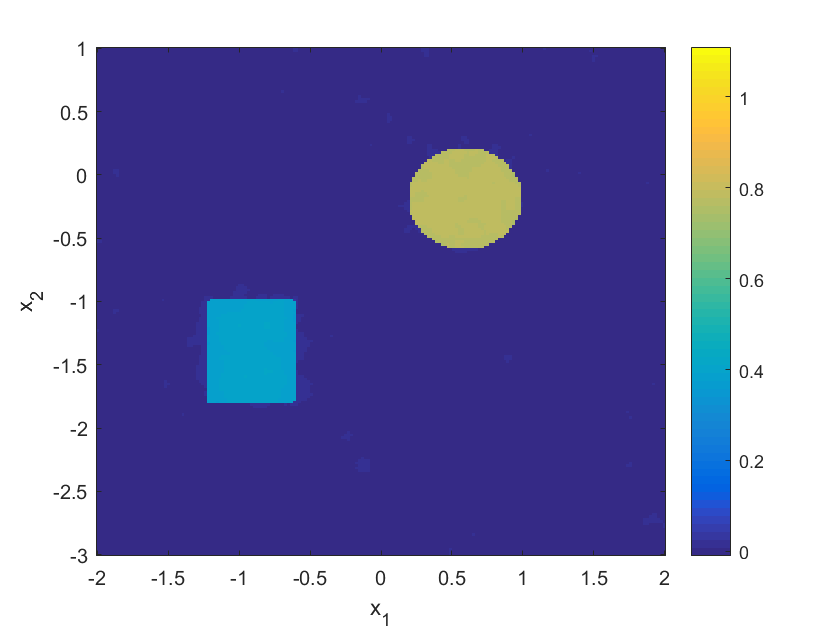} &
   \includegraphics[ width=0.4\linewidth, height=0.4\linewidth, keepaspectratio]{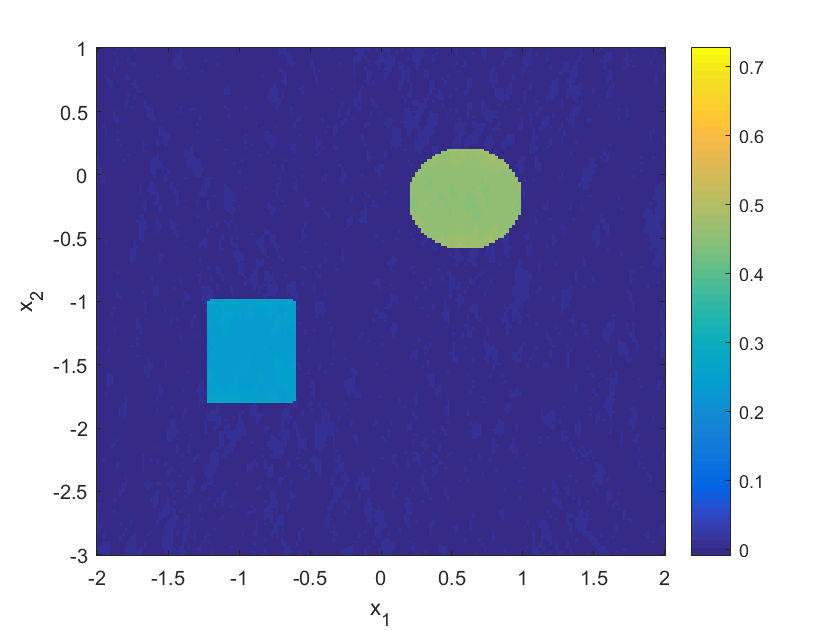} \\ %&
   %\includegraphics[ width=0.4\linewidth, height=0.4\linewidth, keepaspectratio]{JRDen2.png} &
   %\includegraphics[ width=0.4\linewidth, height=0.4\linewidth, keepaspectratio]{JRDen2.png} \\
 %\rotatebox{90}{JLAM} &
   %\includegraphics[ width=0.4\linewidth, height=0.4\linewidth, keepaspectratio]{JRDen2.png} &
   %\includegraphics[ width=0.4\linewidth, height=0.4\linewidth, keepaspectratio]{JRDen2.png} \\ %&
   %\includegraphics[ width=0.4\linewidth, height=0.4\linewidth, keepaspectratio]{JRDen2.png} &
   %\includegraphics[ width=0.4\linewidth, height=0.4\linewidth, keepaspectratio]{JRDen2.png} \\
% \hline
\end{tabular}
\caption{Simple phantom reconstructions, noise level $\eta=0.1$. Comparison of methods TV, JLAM, JTV and LPLS.}
\label{F4}
\end{figure*}
With complete X-ray data, the application of a Lambda term yields this
$R^* \frac{d^2}{ds^2} R f = {-4\pi}\sqrt{-\Delta} f$ \cite[Example
9]{krishnan2014microlocal}, so the singularities of $f$ are preserved
and emphasized by order $1$ in Sobolev scale, so they will dominate
the Lambda reconstruction.  Hence choosing $m=2$ is sufficient for a
full recovery of the image singularities.  Since the singularities are
dominant in the lambda term, they are matched accurately in
\eqref{JLAM}. Indeed we have already seen the effectiveness of such a
filtering approach in recovering the image contours earlier in the
right hand of figure \ref{Fcont}.  We find that setting $m=2$ here
works well as a regularizer on synthetic image phantoms and simulated
data with added pseudo random noise, as we shall now demonstrate. We
note that the derivative filters for $m\neq 2$ are also worth
exploration but we leave such analysis for future work.

% ORIGINAL TEXT:
%recovers
% the singularities of target function modulo smoothing \cite[Theorem 3.5]{rigaud20183d}, namely
% $$R^*\frac{\mathrm{d}^2}{\mathrm{d}s^2}Rf=f+\mathcal{E}f,$$
% where $\tred{\mathcal{E} :L^2(X) \to H^1_{\text{loc}}(X)}$ is a
% continuous integral operator with smooth kernel. Given the smoothness
% of $\mathcal{E}f$, the image edges are fully recovered in the
% microlocal inversion, although there may be an error in the recovery
% of the smooth parts of $f$. Hence choosing $m=2$ is sufficient for a
% full recovery of the image singularities. Indeed we have already seen
% the effectiveness of such a filtering approach in recovering the image
% contours earlier in the right hand of figure \ref{Fcont}.  We find
% that setting $m=2$ here works well as a regularizer on synthetic image
% phantoms and simulated data with added pseudo random noise, as we
% shall now demonstrate. We note that the derivative filters for $m\neq
% 2$ are also worth exploration but we leave such analysis for future
% work. } 

%In lambda tomography applications, this would result in a larger amplification in the measurement error. In our case, $\frac{\mathrm{d}^m}{\mathrm{d}s^m}R$ is used as a regularization term, and hence using a larger $m$ applies more smoothing to our solution (the inverse of $\frac{\mathrm{d}^m}{\mathrm{d}s^m}R$ is a smoothing operator of higher order). 

\subsection{Proposed testing and comparison to the literature}
To test our reconstruction method, we first consider two
test phantoms, one simple and one complex (as in
\cite{webber2020microlocal}). See figure \ref{F4}. The phantoms considered are supported
on $\Gamma$, the region in figure \ref{fig:directions} in which there
is full wavefront coverage from joint X-ray and Compton scattered
data. The simple density phantom consists of a Polyvinyl Chloride (PVC) cuboid and an Aluminium sphere with an approximate density ratio of 1:2 (PVC:Al). The complex density phantom consists of a water ellipsoid, a Sulfur ellipsoid, a Calcium sulfate ($\text{CaSO}_4$)
right-angled-triangle and a thin film of Titanium dioxide ($\text{TiO}_2$) in the shape of a cross. The density ratio of the materials which compose the complex phantom is approximately 1:2:3:4 ($\text{H}_2\text{O}$:S:$\text{CaSO}_4$:$\text{TiO}_2$). The density values used are those of figure \ref{fig2} taken from the NIST database \cite{hubbell1969photon}, and the background densities ($\approx 0$) were set to the density of dry air (near sea level). The corresponding attenuation coefficient
phantoms are simulated similarly. The materials considered are widely used in practice. For example $\text{CaSO}_4$ is used in the production of plaster of Paris and stucco (a common construction material) \cite{wirsching2000calcium}, and $\text{TiO}_2$ is used in the making of decorative thin films (e.g. topaz) and in pigmentation \cite[page 15]{winkler2014titanium}.

To simulate data we set
\begin{equation}
\label{data} b=\begin{pmatrix} b_1\\
b_2\\
\end{pmatrix}=\begin{pmatrix}
R_L\mu_E\\
\mathcal{T}n_e,\\
\end{pmatrix}
\end{equation}
and add a Gaussian noise
\begin{equation}
\label{noise}
b_{\eta}=b+\eta\|b\|_2\frac{v_{G}}{\sqrt{l}},
\end{equation}
for some noise level $\eta$, where $l$ is the length of $b$ and $v_{G}$ is a vector of length $l$ of draws from $\mathcal{N}(0,1)$. For comparison we present separate reconstructions of $\mu_E$ and $n_e$ using Total Variation (TV regularizers). That is we will find
\begin{equation}
\label{Sep1}
\argmin_{\mu_E}\left\|
R_L \mu_E-b_1\right\|^2_2+\alpha\text{TV}(\mu_E)
\end{equation}
to reconstruct $\mu_E$ and 
\begin{equation}
\label{Sep2}
\argmin_{n_e}\left\|
\mathcal{T} n_e-b_2\right\|^2_2+\alpha\text{TV}(n_e)
\end{equation}
for $n_e$, where $\text{TV}(f)=\|\nabla f\|_1$ and $\alpha>0$ is a regularization parameter. We will denote this method as ``TV". In addition we present reconstructions using the state-of-the-art joint reconstruction and regularization techniques from the literature, namely the Joint Total Variation (JTV) methods of \cite{JR2} and the Linear Parallel Level Sets (LPLS) methods of \cite{JR1}. To implement JTV we minimize
\begin{equation}
\label{JTV}
\argmin_{\mu_E,n_e}\left\|\begin{pmatrix}
wR_L & 0\\
0 & \mathcal{T}\\
\end{pmatrix}\begin{pmatrix}
\mu_E\\
n_e\\
\end{pmatrix}-\begin{pmatrix}
wb_1\\
b_2\\
\end{pmatrix}\right\|^2_2+\alpha\text{JTV}_{\beta}(\mu_E,n_e),
\end{equation}
where $w=\|\mathcal{T}\|_2/\|R_L\|_2$ as before, and 
\begin{equation}
\text{JTV}_{\beta}(\mu_E,n_e)=\int_{[-2,2]\times [-3,1]}\paren{\|\nabla\mu_E(\vx)\|_2^2+\|\nabla n_e(\vx)\|_2^2+\beta^2}^{\frac{1}{2}}\mathrm{d}\vx,
\end{equation}
where $\beta>0$ is an additional hyperparameter included so that the gradient of $\text{JTV}_{\beta}$ is defined at zero. This allows one to apply techniques in smooth optimization to solve \eqref{JTV}.
\begin{figure*}
\setlength{\tabcolsep}{5pt}
\begin{tabular}{ c|cc }  %x}{\textwidth}{@{}c*{2}{C}@{}}
  &  Density $n_e$ & Attenuation $\mu_E$, $E=100$keV \\ \hline \\[-0.4cm] %& error $n_e$ & error $\mu_E$ \\ 
 \rotatebox{90}{\hspace{2.05cm} TV}  & 
   \includegraphics[ width=0.4\linewidth, height=0.4\linewidth, keepaspectratio]{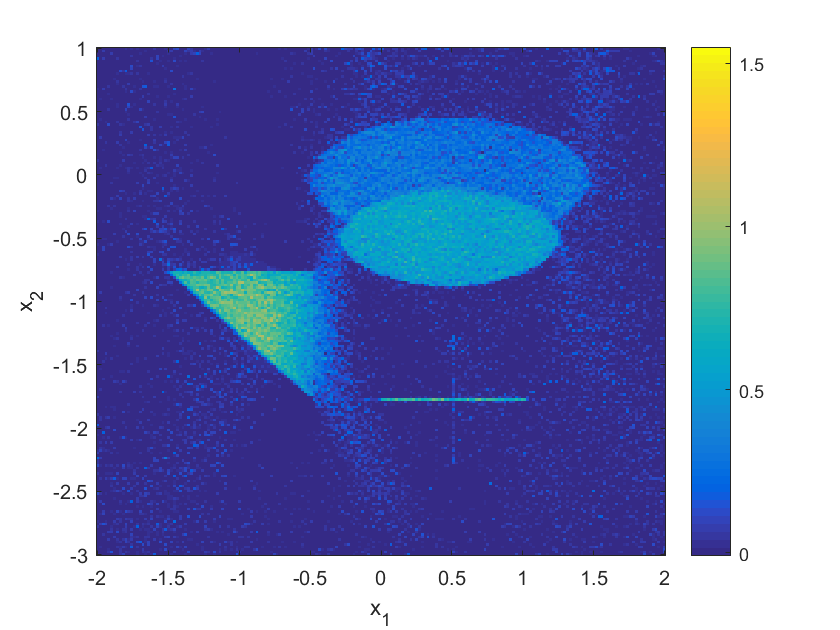} &
   \includegraphics[ width=0.4\linewidth, height=0.4\linewidth, keepaspectratio]{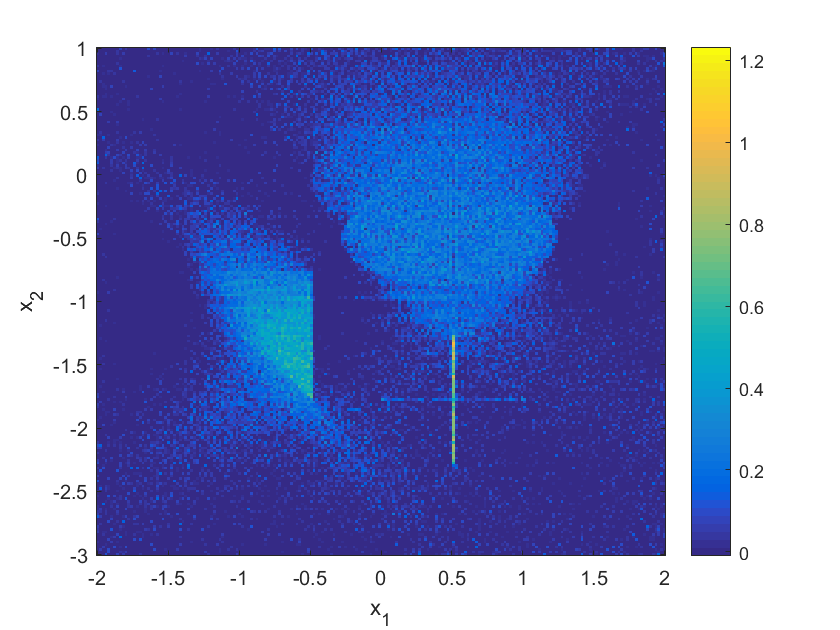} \\ %&
 \rotatebox{90}{\hspace{1.8cm} JLAM} &
  \includegraphics[ width=0.4\linewidth, height=0.4\linewidth, keepaspectratio]{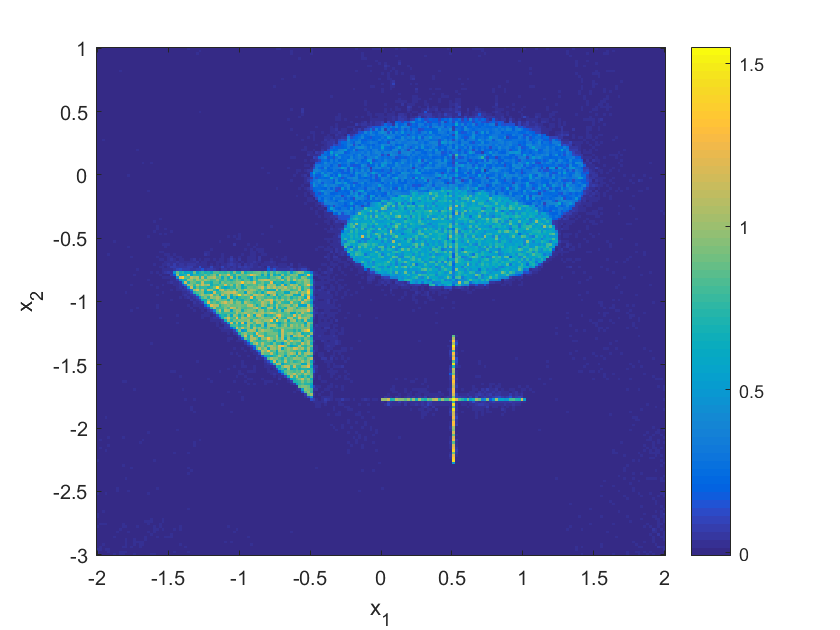} &
   \includegraphics[ width=0.4\linewidth, height=0.4\linewidth, keepaspectratio]{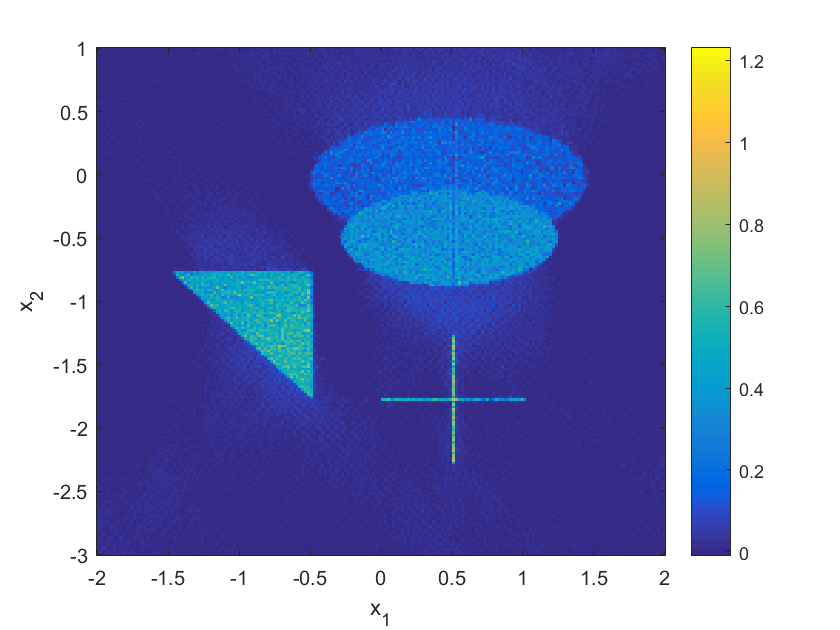} \\ %&
 \rotatebox{90}{\hspace{2cm} JTV} &
   \includegraphics[ width=0.4\linewidth, height=0.4\linewidth, keepaspectratio]{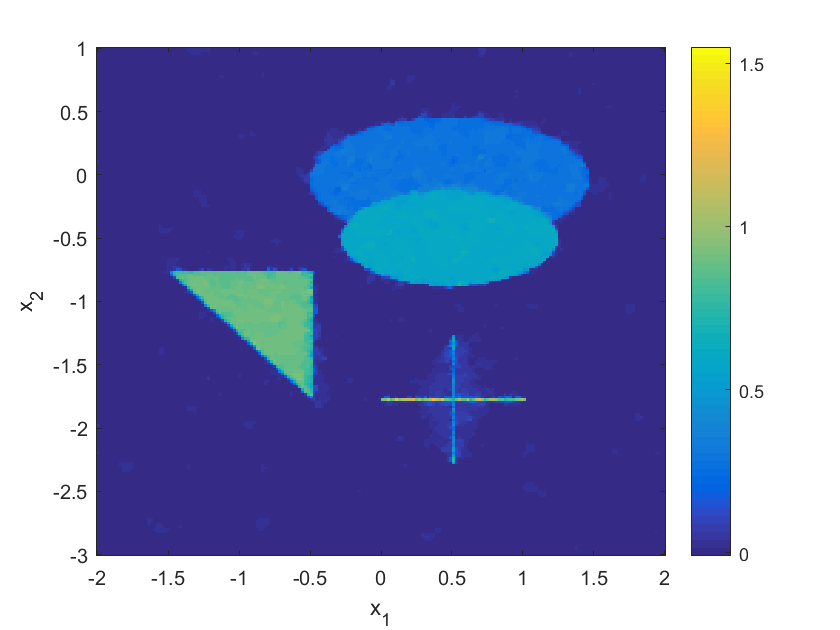} &
   \includegraphics[ width=0.4\linewidth, height=0.4\linewidth, keepaspectratio]{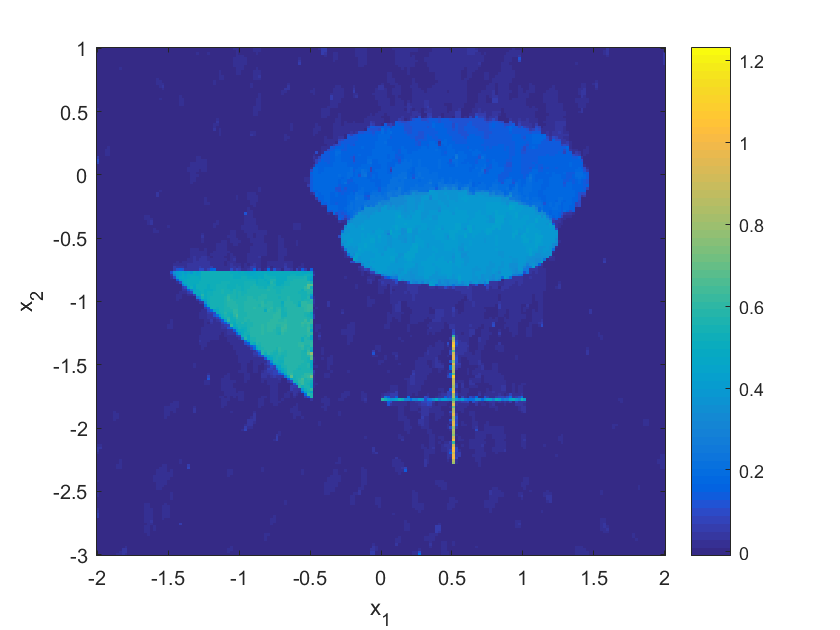} \\ %&
 \rotatebox{90}{\hspace{1.8cm} LPLS} &
   \includegraphics[ width=0.4\linewidth, height=0.4\linewidth, keepaspectratio]{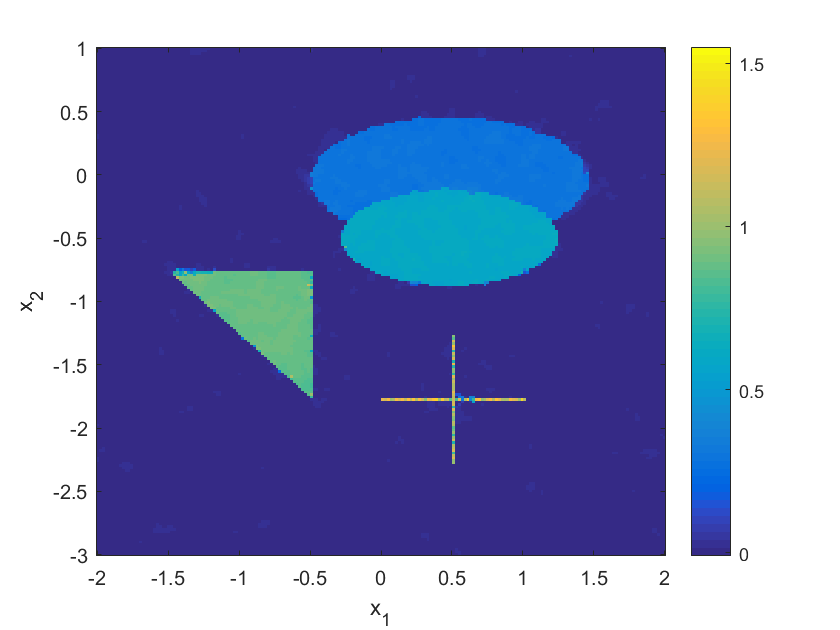} &
   \includegraphics[ width=0.4\linewidth, height=0.4\linewidth, keepaspectratio]{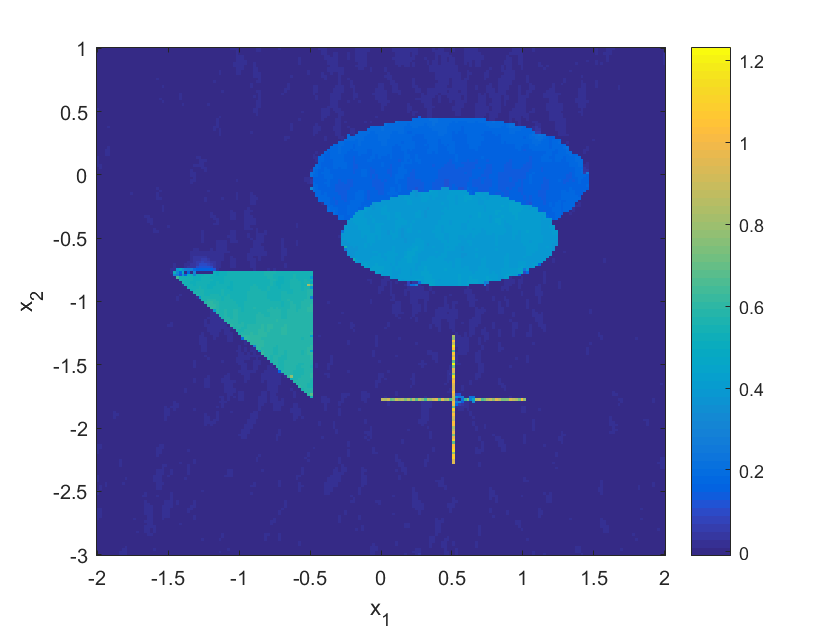} \\ %&
   %\includegraphics[ width=0.4\linewidth, height=0.4\linewidth, keepaspectratio]{JRDen2.png} &
   %\includegraphics[ width=0.4\linewidth, height=0.4\linewidth, keepaspectratio]{JRDen2.png} \\
 %\rotatebox{90}{JLAM} &
   %\includegraphics[ width=0.4\linewidth, height=0.4\linewidth, keepaspectratio]{JRDen2.png} &
   %\includegraphics[ width=0.4\linewidth, height=0.4\linewidth, keepaspectratio]{JRDen2.png} \\ %&
   %\includegraphics[ width=0.4\linewidth, height=0.4\linewidth, keepaspectratio]{JRDen2.png} &
   %\includegraphics[ width=0.4\linewidth, height=0.4\linewidth, keepaspectratio]{JRDen2.png} \\
% \hline
\end{tabular}
\caption{Complex phantom reconstructions, noise level $\eta=0.1$. Comparison of methods TV, JLAM, JTV and LPLS.}
\label{CF1}
\end{figure*}

To implement LPLS we minimize
\begin{equation}
\label{LPLS}
\argmin_{\mu_E,n_e}\left\|\begin{pmatrix}
wR_L & 0\\
0 & \mathcal{T}\\
\end{pmatrix}\begin{pmatrix}
\mu_E\\
n_e\\
\end{pmatrix}-\begin{pmatrix}
wb_1\\
b_2\\
\end{pmatrix}\right\|^2_2+\alpha\text{LPLS}_{\beta}(\mu_E,n_e),
\end{equation}
where
\begin{equation}
\text{LPLS}_{\beta}(\mu_E,n_e)=\int_{[-2,2]\times [-3,1]}\|\nabla\mu_E(\vx)\|_{\beta}\|\nabla n_e(\vx)\|_{\beta}-|\nabla\mu_E(\vx)\cdot\nabla n_e(\vx)|_{\beta^2}\mathrm{d}\vx,
\end{equation}
where $|\vx|_{\beta}=\sqrt{|\vx|^2+\beta^2}$ and $\|\vx\|_{\beta}=\sqrt{\|\vx\|^2_2+\beta^2}$ for $\beta>0$. The JTV and LPLS penalties seek to impose soft constraints on the equality of the image wavefront sets of $\mu_E$ and $n_e$. For example setting $\beta=0$ in the calculation of $\text{LPLS}_{\beta}$ yields
\begin{equation}
\label{DP}
\begin{split}
\text{LPLS}_{\beta}(\mu_E,n_e)&=\|\nabla\mu_E\|_2\|\nabla n_e\|_2-|\nabla\mu_E\cdot\nabla n_e|\\
&=\|\nabla\mu_E\|_2\|\nabla n_e\|_2(1-|\cos\theta|),
\end{split}
\end{equation}
where $\theta$ is the angle between $\nabla n_e$ and $\nabla \mu_E$. Hence \eqref{DP} is minimized for the gradients which are parallel (i.e. when $\theta=0,\pi$), and thus using $\text{LPLS}_{\beta}$ as regularization serves to enforce equality in the image gradient direction and location (i.e. when $\text{LPLS}_{\beta}$ is small, the gradient directions are approximately equal). 

We wish to stress that the comparison with JTV and LPLS is included
purely to illustrate the potential advantages (and disadvantages) of
the lambda regularizers when compared to the state-of-the-art
regularization techniques. Namely is the improvement in image quality
due to joint data, lambda regularizers or are they both beneficial?
We are not claiming a state-of-the-art performance using JLAM, but our
results show JLAM has good performance, and it is numerically easier
to implement, requiring only least squares solvers. There are two
hyperparameters ($\alpha$ and $\beta$) to be tuned in order to
implement JTV and LPLS, which is more numerically intensive (e.g.
using cross validation) in contrast to JLAM with only one
hyperparameter ($\alpha$). Moreover the LPLS objective is non-convex
\cite[appendix A]{JR1}, and hence there are potential local minima to
contend with, which is not an issue with JLAM, being a simple
quadratic objective.

\begin{table}[!h]
\hspace{-1.05cm}
\begin{subtable}{.49\linewidth}\centering
{
\begin{tabular}{| c | c | c | c | c | c | c | c |}
\hline
$\epsilon$ & TV  & JLAM   & JTV & LPLS  \\ \hline
$n_e$     & $.36$ & $.24$ & $.16$ & $.09$ \\ 
$\mu_E$  & $.63$ & $.30$ & $.19$ & $.13$  \\ \hline
\end{tabular}
}
%\caption{$\epsilon$}
\end{subtable}
\begin{subtable}{.49\linewidth}\centering
{\begin{tabular}{| c | c | c | c | c | c | c | c |}
\hline
$F$-score & TV  & JLAM   & JTV & LPLS  \\ \hline
$\text{supp}(n_e)$     & $.78$ & $.98$ & $.97$ & $.99$ \\ 
$\nabla n_e$  & $.73$ & $.83$ & $.84$ & $.84$  \\ 
  $\text{supp}(\mu_E)$   & $.65$ & $.94$ & $.98$ & $.99$ \\ 
$\nabla\mu_E$  & $.39$ & $.83$ & $.84$ & $.85$  \\\hline
\end{tabular}}
%\caption{$F$-score}
\end{subtable}
\caption{Complex phantom $\epsilon$ and $F$-score comparison using TV, JLAM, JTV and LPLS.}
\label{T2}
\end{table}
\begin{table}[!h]
\hspace{-1.05cm}
\begin{subtable}{.49\linewidth}\centering
{\begin{tabular}{| c | c | c | c | c | c | c | c |}
\hline
$\epsilon_{\pm}$  & JLAM   & JTV & LPLS  \\ \hline
$n_e$     & $.25\pm.02$ & $.13\pm.03$ & $.13\pm.03$ \\ 
$\mu_E$  & $.31\pm.03$ & $.16\pm.02$ & $.17\pm.04$ \\ \hline
\end{tabular}}
%\caption{$\epsilon$}
\end{subtable}
\begin{subtable}{.49\linewidth}\centering
{\begin{tabular}{| c | c | c | c | c | c | c | c |}
\hline
$F_{\pm}$ & JLAM   & JTV & LPLS  \\ \hline
$\text{supp}(n_e)$     & $.98\pm.01$ & $.98\pm.005$ & $.99\pm.004$ \\ 
$\nabla n_e$  & $.76\pm.06$ & $.78\pm.05$ & $.75\pm.05$ \\
$\text{supp}(\mu_E)$    & $.90\pm.06$ & $.96\pm.03$ & $.98\pm.007$ \\ 
$\nabla\mu_E$  & $.73\pm.07$ & $.77\pm.05$ & $.74\pm.05$ \\ \hline
\end{tabular}}
%\caption{$F$-score}
\end{subtable}
\caption{Randomized complex phantom $\epsilon_{\pm}$ and $F_{\pm}$ comparison over 100 runs using JLAM, JTV and LPLS.}
\label{T2new}
\end{table}

To minimize \eqref{JLAM}, we store the discrete forms of $R_L$, $R$
and $\mathcal{T}$ as sparse matrices and apply the Conjugate Gradient
Least Squares (CGLS) solvers of \cite{hansen,AIRtools} (specifically
the ``IRnnfcgls" code) with non-negativity constraints (since the
physical quantities $n_e$ and $\mu_E$ are known a-priori to be
nonnegative). To solve equations \eqref{Sep1} and \eqref{Sep2} we
apply the heuristic least squares solvers of \cite{hansen,AIRtools}
(specifically the ``IRhtv" code) with TV penalties and non-negativity
constraints. To solve \eqref{JTV} and \eqref{LPLS} we apply the
codes of \cite{JR1}, modified so as to suit a Gaussian noise model (a
Poisson model is used in \cite[equation 3]{JR1}). The relative
reconstruction error $\epsilon$ is calculated as
$\epsilon=\|\vx-\vy\|_2/\|\vx\|_2$, where $\vx$ is the ground truth
image and $\vy$ is the reconstruction. For all methods compared
against we simulate data and added noise as in equations \eqref{data}
and \eqref{noise}, and the noise level added for each simulation is
$\eta=0.1$ ($10\%$ noise). We choose $\alpha$ for each method
such that $\epsilon$ is minimized for a noise level of $\eta=0.1$.
That is we are comparing the best possible performance of each method.
We set $\beta$ for JTV and LPLS to the values used on the ``lines2"
data set of \cite{JR1}. We do not tune $\beta$ to the best performance
(as with $\alpha$) so as to give fair comparison between TV, JLAM, JTV
and LPLS. After the optimal hyperparameters were selected, we
performed 100 runs of TV, JLAM, JTV and LPLS on both phantoms for 100
randomly selected sets of materials. That is, for 100 randomly chosen
sets of values from figure \ref{fig2} and the NIST database, with the
NIST values corresponding to the nonzero parts of the phantoms. We
present the mean ($\mu_{\epsilon}$) and standard deviation
($\sigma_{\epsilon}$) relative errors over 100 runs in the left-hand
of tables \ref{T1new} and \ref{T2new} for the simple and complex
phantom respectively. The results are given in the form
$\epsilon_{\pm}=\mu_{\epsilon}\pm\sigma_{\epsilon}$ for each method.
In addition to the relative error $\epsilon$, we also provide metrics
to measure the structural accuracy of the results. Specifically we
will compare $F$-scores on the image gradient and support, as is done
in \cite{DICE1,DICE2}. The gradient $F$-score \cite{DICE2} measures
the wavefront set reconstruction accuracy, and the support $F$-score
\cite[page 5]{DICE1} (see DICE score) is a measure of the geometric
accuracy. That is, the support $F$-score checks whether the
reconstructed phantoms are the correct shape and size. The $F$-score
takes values on $[0,1]$. For this metric, values close to one indicate
higher performance, and conversely for values close to zero. Similarly
to $\epsilon$, we present the $F$-scores of the randomized tests in
the form $F_{\pm}=\mu_{F}\pm\sigma_{F}$, where $\mu_F$ and $\sigma_F$
are the mean and standard deviation $F$-scores respectively. In all
tables, the support $F$-scores are labelled by $\text{supp}(n_e)$ and
$\text{supp}(\mu_E)$, and by $\nabla n_e$ and $\nabla\mu_E$ for the
gradient $F$-scores.

\subsection{Results and discussion}
\label{RnD}
See figure \ref{F4} for image reconstructions of the simple phantom
using TV, JLAM, JTV and LPLS, and see table \ref{T1} for the corresponding $\epsilon$ and $F$-score values. See table \ref{T1new} for the $\epsilon_{\pm}$ and $F_{\pm}$ values calculated over 100 randomized simple phantom reconstructions. For the
complex phantom, see figure \ref{CF1} for image reconstructions, and table \ref{T2} for the $\epsilon$ and $F$-score values. See table \ref{T2new} for $\epsilon_{\pm}$ and $F_{\pm}$. In the
separate reconstruction of $n_e$ (using method TV) we see a blurring
of the ground truth image edges (wavefront directions) in the
horizontal direction and there are artefacts in the reconstruction due to limited data, as predicted by our microlocal theory. In the TV reconstruction of $\mu_E$ we see a similar effect, but in this case we fail to resolve the wavefront directions in the vertical direction due to limited line integral data. This is as predicted by the theory of section \ref{microsec2} and \cite{borg2018analyzing}.
\begin{figure}[!h]
\centering
\begin{subfigure}{0.42\textwidth}
\includegraphics[width=1.0\linewidth, height=5.5cm]{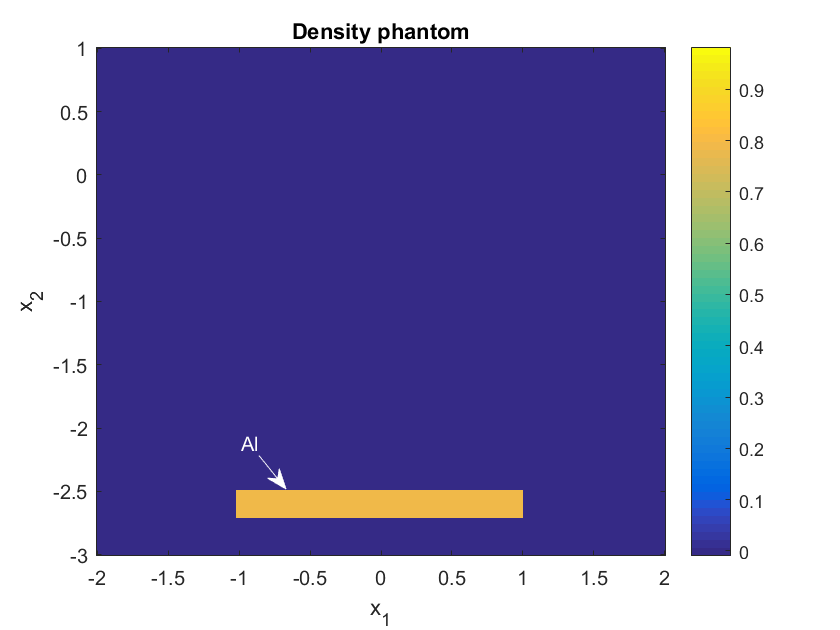} 
\end{subfigure}\hspace{5mm}
\begin{subfigure}{0.42\textwidth}
\includegraphics[width=1.0\linewidth, height=5.5cm]{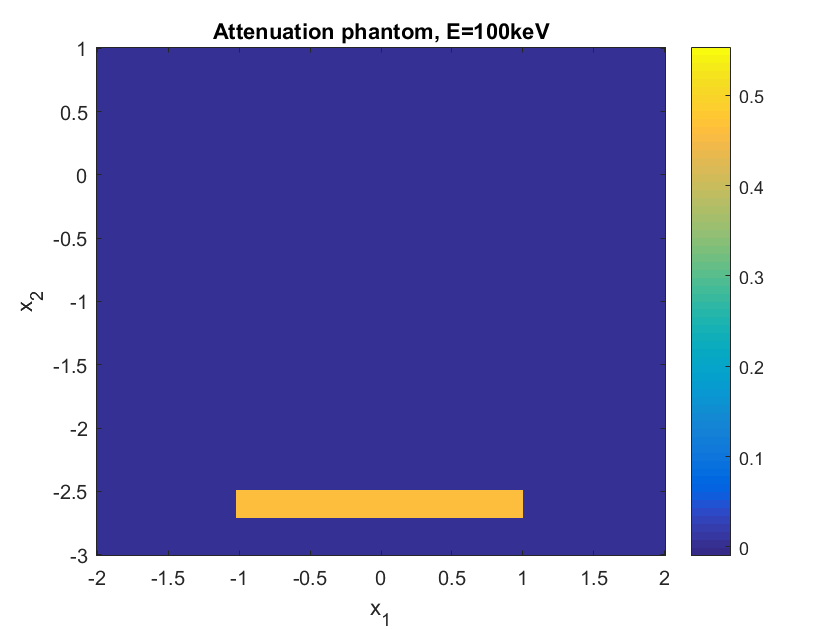} 
\end{subfigure}
\caption{Horizontal Al bar density (left) and attenuation (right) phantoms.}
\label{Barphan}
\end{figure}
In section \ref{microsec1} we discovered the existence also of
nonlocal artefacts in the $n_e$ reconstruction, which were induced
by the mappings $\lambda_{ij}$. However these were found to lie largely outside the imaging space unless the singularity in question $(\vx,\boldsymbol{\xi})\in \text{WF}(n_e)$ were such that $\vx$ is close to the detector array (see figures \ref{FC1} and \ref{FC2}). Hence why we do not see the effects of the $\lambda_{ij}$ in the phantom reconstructions, as the phantoms are bounded sufficiently away from the detector array. The added regularization may smooth out such artefacts also, which was found to be the case in \cite{webber2020microlocal}. 

Using the joint reconstruction methods
(i.e. JLAM, JTV and LPLS) we see a large reduction in the image
artefacts in $n_e$ and $\mu_E$, since with joint data we are able to
stably resolve the image singularities in all directions. The
improvement in $\epsilon$ and the $F$-score is also significant, particularly in the $\mu_E$ reconstruction. Thus it seems that
the joint data is the greater contributor (over the
regularization) to the improvement in the image quality, as the approaches with joint
data each perform well. 
\begin{figure*}
\setlength{\tabcolsep}{5pt}
\begin{tabular}{ c|cc }  %x}{\textwidth}{@{}c*{2}{C}@{}}
  &  Density $n_e$ & Attenuation $\mu_E$, $E=100$keV \\ \hline \\[-0.4cm] %& error $n_e$ & error $\mu_E$ \\ 
 \rotatebox{90}{\hspace{2.05cm} TV}  & 
   \includegraphics[ width=0.4\linewidth, height=0.4\linewidth, keepaspectratio]{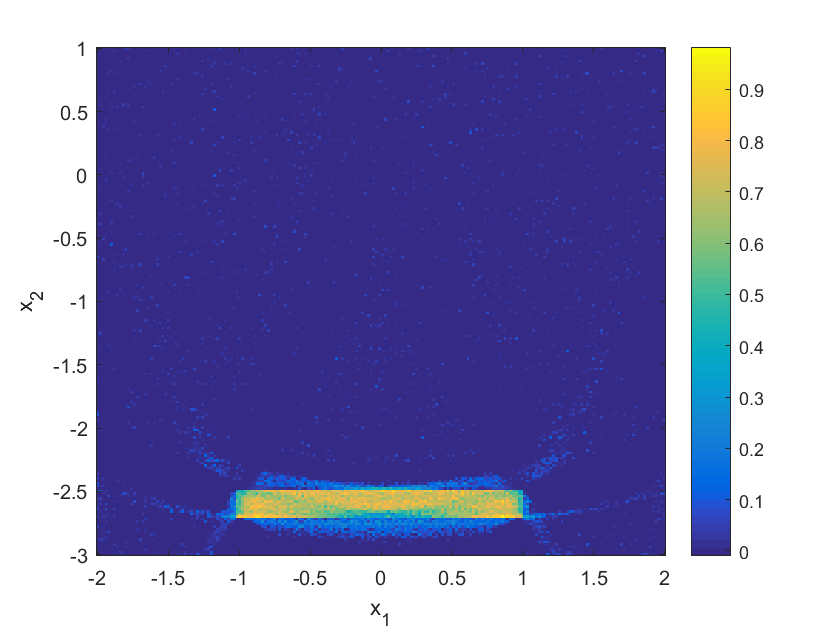} &
   \includegraphics[ width=0.4\linewidth, height=0.4\linewidth, keepaspectratio]{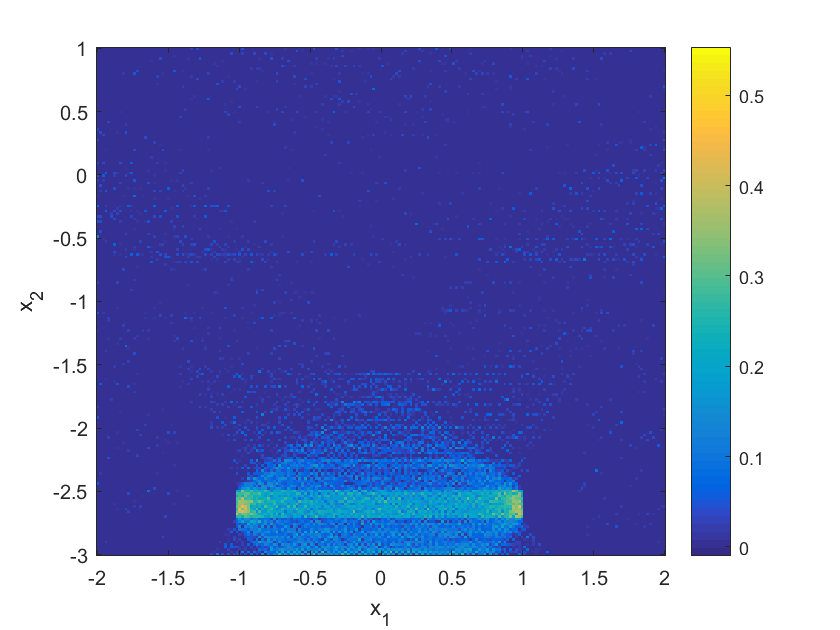} \\ %&
 \rotatebox{90}{\hspace{1.8cm} JLAM} &
  \includegraphics[ width=0.4\linewidth, height=0.4\linewidth, keepaspectratio]{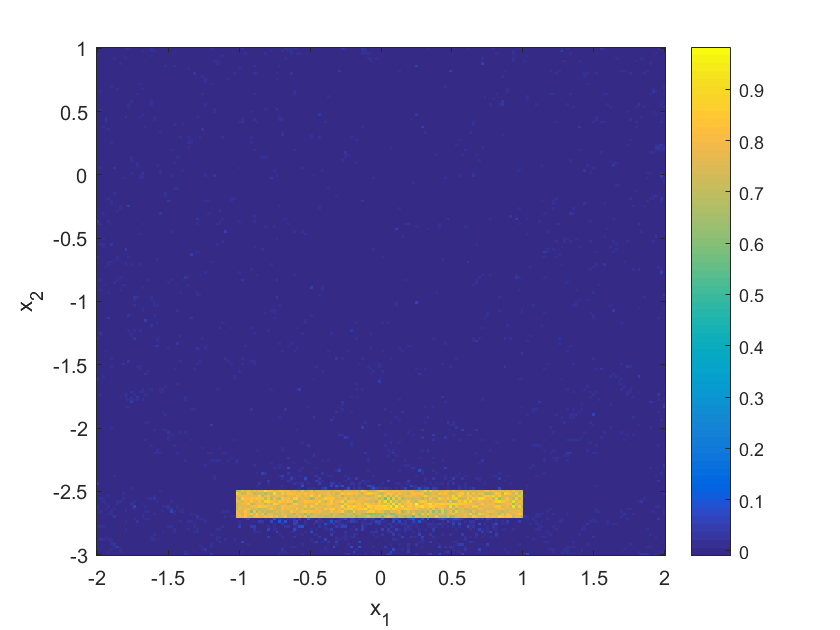} &
   \includegraphics[ width=0.4\linewidth, height=0.4\linewidth, keepaspectratio]{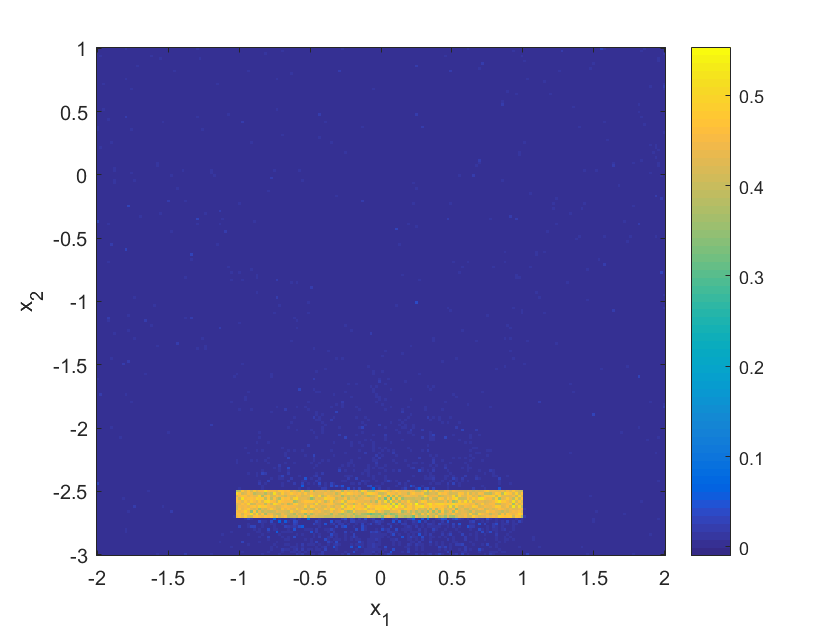} \\ %&
 \rotatebox{90}{\hspace{2cm} JTV} &
   \includegraphics[ width=0.4\linewidth, height=0.4\linewidth, keepaspectratio]{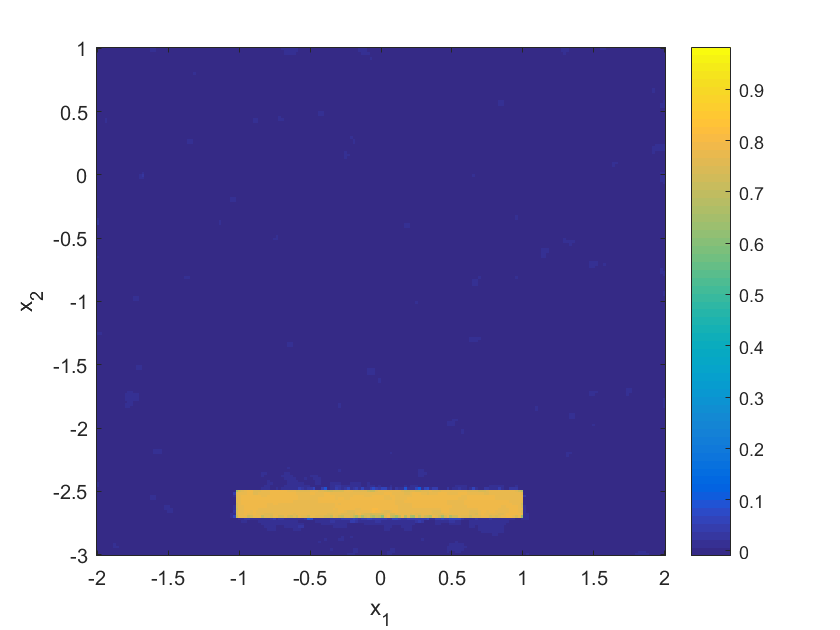} &
   \includegraphics[ width=0.4\linewidth, height=0.4\linewidth, keepaspectratio]{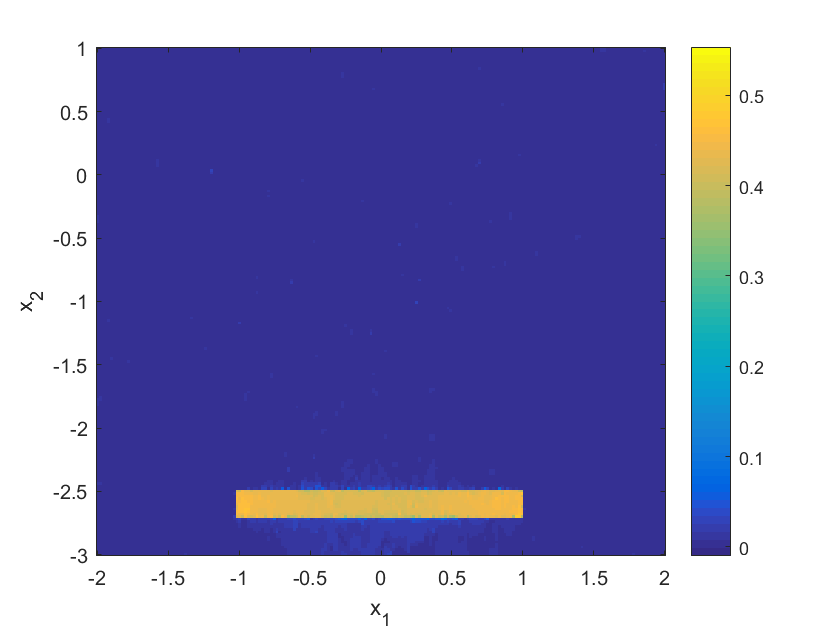} \\ %&
 \rotatebox{90}{\hspace{1.8cm} LPLS} &
   \includegraphics[ width=0.4\linewidth, height=0.4\linewidth, keepaspectratio]{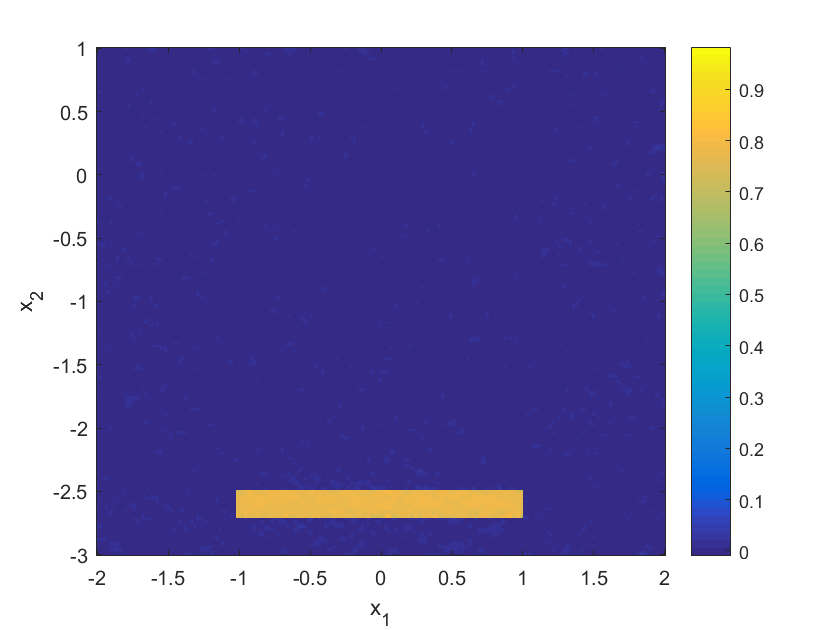} &
   \includegraphics[ width=0.4\linewidth, height=0.4\linewidth, keepaspectratio]{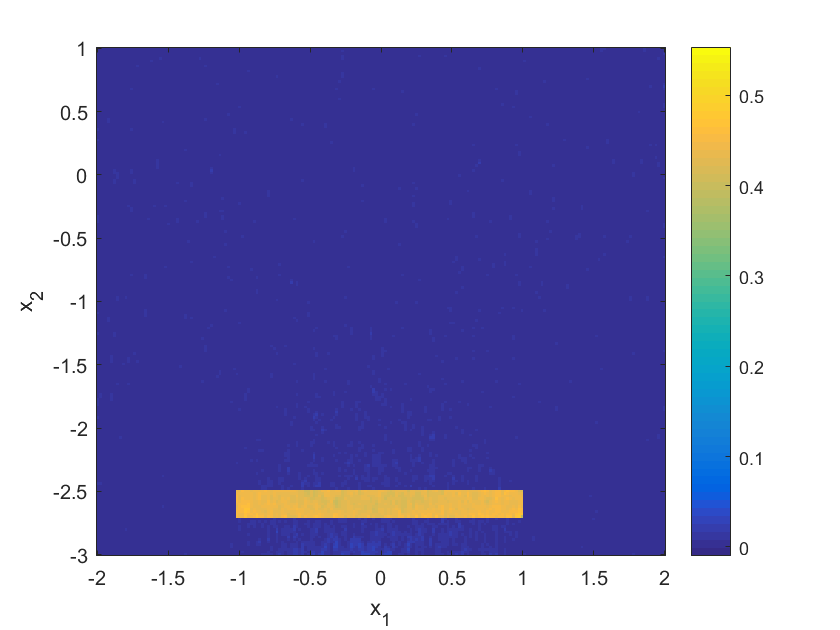} \\ %&
   %\includegraphics[ width=0.4\linewidth, height=0.4\linewidth, keepaspectratio]{JRDen2.png} &
   %\includegraphics[ width=0.4\linewidth, height=0.4\linewidth, keepaspectratio]{JRDen2.png} \\
 %\rotatebox{90}{JLAM} &
   %\includegraphics[ width=0.4\linewidth, height=0.4\linewidth, keepaspectratio]{JRDen2.png} &
   %\includegraphics[ width=0.4\linewidth, height=0.4\linewidth, keepaspectratio]{JRDen2.png} \\ %&
   %\includegraphics[ width=0.4\linewidth, height=0.4\linewidth, keepaspectratio]{JRDen2.png} &
   %\includegraphics[ width=0.4\linewidth, height=0.4\linewidth, keepaspectratio]{JRDen2.png} \\
% \hline
\end{tabular}
\caption{Horizontal bar phantom reconstructions, noise level $\eta=0.1$. Comparison of methods TV, JLAM, JTV and LPLS.}
\label{BF1}
\end{figure*}
Upon comparison of JLAM, JTV and LPLS, the $\epsilon$ metrics are
significantly improved when using JTV and LPLS over JLAM, but the
image quality and $F$-scores are highly comparable. This indicates
that, while the noise in the reconstruction is higher using JLAM, the
recovery of the image edges and support is similar using JLAM, JTV and
JLAM. As theorized, the lambda regularizers were successful in
preserving the wavefront sets of $\mu_E$ and $n_e$. However there is a
distortion present in the nonzero parts of the JLAM reconstruction.
This is the most notable difference in JLAM and JTV/LPMS, and is
likely the cause of the $\epsilon$ discrepancy. So while the edge
preservation and geometric accuracy of JLAM is of a high quality (and
this was our goal), the smoothing properties of JLAM are not up to par
with the state-of-the-art currently. We note however that the JTV and
LPLS objectives are nonlinear (with LPLS also non-convex) and require
significant additional machinery (e.g. in the implementation of the
code of \cite{JR1} used here) in the inversion when compared to JLAM,
which is a straight forward implementation of linear least squares
solvers.
\begin{table}[!h]
\hspace{-1.05cm}
\begin{subtable}{.49\linewidth}\centering
{
\begin{tabular}{| c | c | c | c | c | c | c | c |}
\hline
$\epsilon$ & TV  & JLAM   & JTV & LPLS  \\ \hline
$n_e$     & $.28$ & $.09$ & $.04$ & $.02$ \\ 
$\mu_E$  & $.68$ & $.11$ & $.09$ & $.07$  \\ \hline
\end{tabular}}
\end{subtable}
\begin{subtable}{.49\linewidth}\centering
{\begin{tabular}{| c | c | c | c | c | c | c | c |}
\hline
$F$-score & TV  & JLAM   & JTV & LPLS  \\ \hline
$\text{supp}(n_e)$     & $.71$ & $.99$ & $\sim 1$ & $\sim 1$ \\ 
$\nabla n_e$  & $.64$ & $.82$ & $.79$ & $.83$  \\ 
  $\text{supp}(\mu_E)$   & $.54$ & $\sim 1$ & $\sim 1$ & $\sim 1$ \\ 
$\nabla\mu_E$  & $.53$ & $.82$ & $.80$ & $.90$  \\\hline
\end{tabular}}
%\caption{$F$-score}
\end{subtable}
\caption{Al bar phantom $\epsilon$ and $F$-score comparison using TV, JLAM, JTV and LPLS.}
\label{TB1}
\end{table}
\begin{table}[!h]
\hspace{-1.05cm}
\begin{subtable}{.49\linewidth}\centering
{
\begin{tabular}{| c | c | c | c | c | c | c | c |}
\hline
$\epsilon_{\pm}$  & JLAM  & JTV & LPLS  \\ \hline
$n_e$     & $.10\pm.02$ & $.04\pm.004$ & $.03\pm.02$ \\ 
$\mu_E$  & $.12\pm.03$ & $.12\pm.03$ & $.08\pm.03$ \\ \hline
\end{tabular}}
\end{subtable}
\begin{subtable}{.49\linewidth}\centering
{\begin{tabular}{| c | c | c | c | c | c | c | c |}
\hline
$F_{\pm}$ & JLAM   & JTV & LPLS  \\ \hline
$\text{supp}(n_e)$     & $.99\pm.02$ & $\sim 1\pm.003$ & $\sim 1\pm.001$ \\ 
$\nabla n_e$  & $.83\pm.02$ & $.79\pm.02$ & $.84\pm.02$ \\
$\text{supp}(\mu_E)$    & $.98\pm.06$ & $.99\pm.02$ & $\sim 1\pm.008$ \\ 
$\nabla\mu_E$  & $.81\pm.03$ & $.77\pm.02$ & $.85\pm.02$ \\ \hline
\end{tabular}}
%\caption{$F$-score}
\end{subtable}
\caption{Randomized bar phantom $\epsilon_{\pm}$ and $F_{\pm}$ comparison over all NIST materials considered  (153 runs) using JLAM, JTV and LPLS.}
\label{TB2}
\end{table}

\subsection{Reconstructions with limited data}
The simple and complex phantoms considered thus far are supported within $\Gamma$ (the yellow region of figure \ref{fig:directions}) so as to allow for a full wavefront coverage in the reconstruction. To investigate what happens when the object is supported outside of $\Gamma$, we present additional reconstructions of an Aluminium bar phantom with support towards the bottom (close to $x_2=-3$) of the reconstruction space. See figure \ref{Barphan}. In this case we have limited data and the full wavefront coverage is not available with the combined X-ray and Compton data sets. Image reconstructions of the Al bar phantom are presented in figure \ref{BF1}, and the corresponding $\epsilon$ and $F$-score values are displayed in table \ref{TB1}. See table \ref{T2new} for the $\epsilon_{\pm}$ and $F_{\pm}$ values corresponding to the randomized bar phantom reconstructions. In this case $\epsilon_{\pm}$ and $F_{\pm}$ were calculated from reconstructions of 153 bar phantoms (we used 100 runs previously), replacing the Al density value of figure \ref{Barphan} with one of each NIST value considered (153 in total). The reconstruction processes and hyperparameter selection applied here were exactly the same as for the simple and complex phantom. In this case we see artefacts in the Compton reconstruction along curves which follow the shape of the boundary of $\Gamma$, and the X-ray artefacts constitute a vertical blurring as before. The $\epsilon$ error when using JLAM, JTV and LPLS is more comparable in this example (compared to tables \ref{T1} and \ref{T2}), particularly in the case of the $\mu_E$ phantom. The image quality and $F$-scores are again similar as with the simple and complex phantom examples. All joint reconstruction methods were successful in removing the image artefacts observed in the separate reconstructions, and thus can offer satisfactory image quality under the constraints of limited data. However this is only a single test of the capabilities of JLAM, JTV and LPLS with limited data and we leave future work to conclude such analysis.
%\clearpage

\section{Conclusions and further work}\label{conclusion}
Here we have introduced a new joint reconstruction method ``JLAM" for low effective $Z$ imaging ($Z<20$), based on ideas in lambda tomography. We considered primarily the ``parallel line segment" geometry of \cite{webber2019compton}, which is motivated by system architectures for airport security screening applications. In section \ref{microsec1} we gave a microlocal analysis of the toric section transform $\Tc$, which was first proposed in \cite{webber2019compton} for a CST problem. Explicit expressions were provided for the microlocal artefacts and verified through simulation. Section \ref{microsec2} explained the X-ray CT artefacts using the theory of \cite{borg2018analyzing}. Following the theory of sections \ref{microsec1} and \ref{microsec2}, we detailed the JLAM algorithm in section \ref{results}. Here we conducted simulation testing and compared JLAM to separate reconstructions using TV, and to the nonlinear joint inversion methods, JTV \cite{JR2} and LPLS \cite{JR1} from the literature. The joint inversion methods considered (i.e. JLAM, JTV and LPLS) were successful in preserving the image contours in the reconstruction, as predicted. However the smoothing applied by JLAM was not as effective as JTV and LPLS, and we saw a distortion in the JLAM reconstruction (see figures \ref{F4} and \ref{CF1}). JTV and LPLS were thus shown to offer better performance than JLAM, with LPLS producing the best results overall. The advantages of JLAM over JTV and LPLS are in the fast, linear inversion, and the reduction in tuning parameters (one for JLAM, two for JTV/LPLS). Given the linearity of JLAM, the ideas of JTV and LPLS can be easily integrated with lambda regularization to modify the objectives of the literature and improve further the edge resolution of the reconstruction. To preserve the linearity of JLAM we could also combine JLAM with a Tikhonov regularizer. This may help smooth out the distortion observed in the JLAM reconstruction. We leave such ideas for future work.

 \section*{Acknowledgements} We would like to thank the journal reviewers for their helpful comments and insight towards this article, particularly in regards to the simulation study. This material is based upon work
supported by the U.S.\ Department of Homeland Security, Science and
Technology Directorate, Office of University Programs, under Grant
Award 2013-ST-061-ED0001. The views and conclusions contained in this
document are those of the authors and should not be interpreted as
necessarily representing the official policies, either expressed or
implied, of the U.S. Department of Homeland Security.  The work of the
second author was partially supported by U.S.\ National Science
Foundation grant DMS 1712207.  The authors thank the referees for
thorough reviews and thoughtful comments that improved the article.
\bibliographystyle{abbrv} \bibliography{referencesMicCom}

\appendix \section{Generating the plots of figure \ref{fig2}}
\label{plots} The generation of the plots of figure \ref{fig2} is
explained in more detail here. We will explain the generation of the
plot for $E=100$keV. Refer to figure \ref{Fplots}. We first plotted
$\mu_E$ for $E=100$keV against $n_e$ for all materials in the NIST
database \cite{hubbell1969photon} with effective $Z$ less than 20.
This is the left hand plot of figure \ref{Fplots}. The set of
materials with effective $Z<20$ was
$$Z_{\text{eff}}=\{Z : \sigma_E(Z)<\sigma_E(20), E=100\text{keV}\},$$
where $\sigma_E$ is the electron cross section. We noticed a large outlier (coal, or amorphous Carbon) which corrupts the correlation in our favour, and hence we chose to remove the material from consideration in simulation. The outlier is highlighted in the left hand plot. After the outlier was removed we noticed a number of materials located at the origin (with negligible attenuation coefficient and density, such as air) in the middle scatter plot of figure \ref{Fplots}. As such materials again bias the correlation and plot standard deviation in our favour, these were removed to produce the left hand plot of figure \ref{fig2} in the right hand of figure \ref{Fplots}. The same points were removed in the generation of the right hand plot of figure \ref{fig2} also, for $E=1$MeV.
\begin{figure}[!h]
%\hspace*{-1cm}
\begin{subfigure}{0.32\textwidth}
\includegraphics[width=0.9\linewidth, height=4cm]{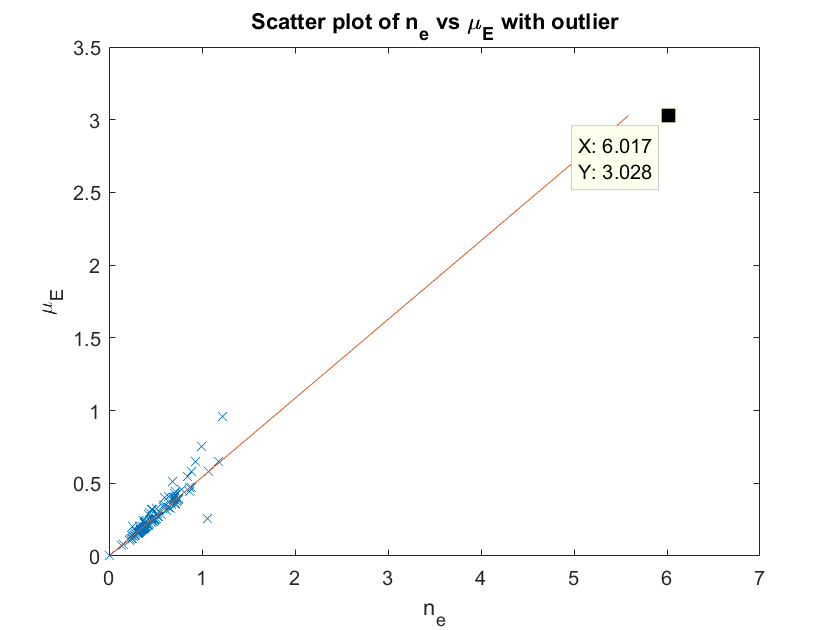}
\end{subfigure}
\begin{subfigure}{0.32\textwidth}
\includegraphics[width=0.9\linewidth, height=4cm]{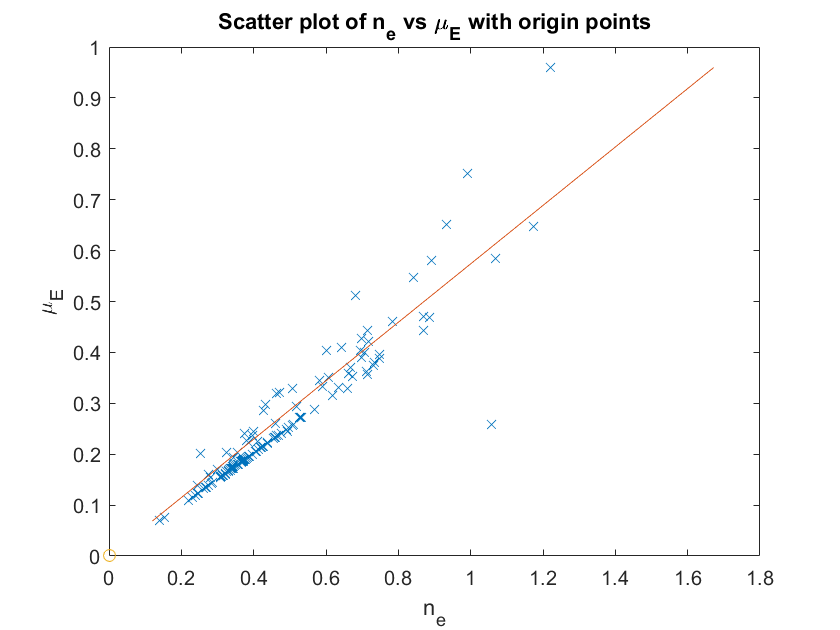} 
\end{subfigure}
\begin{subfigure}{0.32\textwidth}
\includegraphics[width=0.9\linewidth, height=4cm]{scat.png}
\end{subfigure}
\caption{Scatter plot with outlier and origin points included (left, R=0.98), scatter plot with the outlier removed and origin points included, the origin points highlighted by an orange circle (middle, R=0.95), and the scatter plot of figure \ref{fig2} with outliers and origin points removed (right, R=0.93).}
\label{Fplots}
\end{figure}

\end{document}